%% file: PATH.tex
\documentclass[11pt]{article}

\usepackage{amssymb}
\usepackage{amsmath}
\usepackage{ifthen}
\usepackage{bbm}
\usepackage[noend]{algorithmic}
\usepackage{algorithm}
\usepackage[plain]{fullpage}
\usepackage{graphicx}
\usepackage[footnotesize,hang]{caption}
\usepackage{mathtools}
\usepackage{url}
\usepackage{sectsty}
\usepackage{enumitem}
\usepackage{yzdefs}
\usepackage{marginnote}
\usepackage[T1]{fontenc}

\allsectionsfont{\sffamily}

\algsetup{indent=2em}
\renewcommand\ij{i\kern -.1em j}

\newcommand\ceil[1]{\left\lceil #1 \right\rceil}
\newcommand\floor[1]{\left\lfloor #1 \right\rfloor}

\newcommand{\width}{\textup{width}}
\newcommand{\length}[1][\relax]{%
  \textup{len}%
  \ifthenelse{\equal{#1}{\relax}}{#1}{\left(#1\right)}%
}
\newcommand{\pdspan}{\textup{span}} 
\newcommand{\pw}{\textup{pw}}

\newcommand{\size}{\textup{\textup{size}}}
\newcommand{\mbalpha}{\mbox{\boldmath$\alpha$}}
\newcommand{\maxalpha}{\widehat{\alpha}}
\newcommand{\minbeta}{\widehat{\beta}}

\newcommand{\border}{\delta}
\newcommand{\cf}[1]{\textbf{\textsc{#1}}}

\newcommand\YES{\mbox{\normalfont\textsc{yes}}}

\newcommand{\qed}{\hspace*{\fill}\nolinebreak\ensuremath{\Box}}

\newcommand{\PUP}{\mbox{\normalfont\textsc{PUP}}}
\newcommand{\PathDec}{\mbox{\normalfont\textsc{PD}}}
\newcommand{\PathLen}{\mbox{\normalfont\textsc{MLPD}}}
\newcommand{\PathLenFixed}[1]{\PathLen_{#1}}
\newcommand{\LCPD}{\mbox{\normalfont\textsc{LCPD}}}
\newcommand{\LCPDFixed}[1]{\LCPD_{#1}}

\newcommand{\bA}{\mathbb{A}}
\newcommand{\bB}{\mathbb{B}}
\newcommand{\bC}{\mathbb{C}}

\newcommand{\bagsize}[1]{e_{#1}}

\setitemize{topsep=0pt,parsep=0pt,itemsep=0pt}
\setenumerate{itemsep=0pt,topsep=0pt}

\begin{document}

\newcommand\thetitle{Minimum length path decompositions}

\title{\textbf{\thetitle}}
\author{Dariusz Dereniowski\footnote{Gda\'{n}sk University of Technology, Gda\'{n}sk, Poland, deren@eti.pg.gda.pl} \and
Wieslaw Kubiak\footnote{Memorial University, St. John's, Canada,
wkubiak@mun.ca}\and Yori Zwols\footnote{This research was conducted while this author was at McGill University and Concordia University.}
}

\maketitle
\begin{abstract}
We consider a bi-criteria generalization of the pathwidth problem, where, for given integers $k,l$ and a graph $G$, we ask whether there exists a path decomposition $\cP$ of $G$ such that the width of $\cP$ is at most $k$ and the number of bags in $\cP$, i.e., the \emph{length} of $\cP$, is at most $l$.

We provide a complete complexity classification of the problem in terms of $k$ and $l$ for general graphs.
Contrary to the original pathwidth problem, which is fixed-parameter tractable with respect to $k$, we prove that the generalized problem is NP-complete for any fixed $k\geq 4$, and is also NP-complete for any fixed $l\geq 2$.
On the other hand, we give a polynomial-time algorithm that, for any (possibly disconnected) graph $G$ and integers $k\leq 3$ and $l>0$, constructs a path decomposition of width at most $k$ and length at most $l$, if any exists.

As a by-product, we obtain an almost complete classification of the problem in terms of $k$ and $l$ for connected graphs. Namely, the problem is NP-complete for any fixed $k\geq 5$ and it is polynomial-time for any $k\leq 3$. This leaves open the case $k=4$ for connected graphs.
\end{abstract}


\section{Introduction} \label{sec:introduction}
\input{introduction}

\subsection{Applications} \label{sec:applications}
\input{applications}

\subsection{Preliminaries} \label{sec:preliminaries}
\input{preliminaries}

\section{$\PathLenFixed{k}$ is NP-hard for $k\geq 4$} \label{sec:disconnected_five}
\input{disconnected_five}

\section{$\PathLenFixed{k}$ for connected graphs, $k\leq 3$} \label{sec:conn_disconn_four}
Section \ref{sec:disconnected_five} dealt with the entries marked with ``NP-hard'' in Table \ref{tbl:results}. 
In the remainder of this paper, we will prove the ``poly-time'' entries in the table. Therefore, from now on, all path 
decompositions that we deal with have width at most $3$.
The main result of this section is the following theorem.

\begin{theorem} \label{thm:connected}
Let $k\in\{1,2,3\}$.
\begin{list}{}{}
\item[If: ]
for each $k'\in\{0,\ldots,k-1\}$, 
there exists a polynomial-time algorithm that, for any graph $G$ either constructs a minimum length path decomposition of width $k'$ of $G$, 
or concludes that no such path decomposition exists, 
\item[then:]
there exists a polynomial-time algorithm that, for any  \textbf{connected} graph $G$, either constructs a minimum-length path decomposition of width $k$ of $G$, or  concludes that no such path decomposition exists.
\end{list}
\end{theorem}

We take several steps to prove this theorem.
In Section~\ref{sec:general_alg} we formulate an algorithm that outlines the main idea of our method, but whose running time is not necessarily polynomial.
This algorithm constructs a directed graph $\cG_k$ such that the directed paths leading from its source $s$ to its sink $t$ correspond to path decompositions of width $k$ of $G$.
Moreover, the length of a directed $s{-}t$ path in $\cG_k$ equals the length of the corresponding path decomposition of $G$.
Hence, our problem reduces to computing a shortest path in $\cG_k$.
The running time of this algorithm  is, in general, not polynomial since the size of $\cG_k$ may be exponential in the size of $G$.
Hence, the remainder of Section~\ref{sec:conn_disconn_four} is devoted to providing a different construction of $\cG_k$ that preserves the above-mentioned relation between shortest paths in $\cG_k$ and path decompositions of $G$, and furthermore ensures that the size of $\cG_k$ is polynomial in the size of $G$.
To that end we develop some notation and obtain several properties of minimum-length path decomposition of width at most $3$ of a connected graph (Sections~\ref{sec:bottleneck} and~\ref{sec:clean}).
Finally, Section~\ref{sec:connected_four} provides the polynomial-time algorithm and proves its correctness.
Our proof of \ref{thm:connected} is constructive provided that the algorithms from the `if' part of this theorem exist.
We deal with the latter in Sections~\ref{sec:one_big_component} and~\ref{sec:disconnected}.

\subsection{A generic (non-polynomial) algorithm} \label{sec:general_alg}
\input{generic_alg}

\input{properties}

\subsection{Well-arranged path decompositions} \label{sec:clean}
\input{clean}

\subsection{An algorithm for connected graphs} \label{sec:connected_four}
\input{connected_four}

\section{$\PathLenFixed{k}$ for graphs with one big component, $k\leq 3$} \label{sec:one_big_component}
In this section and in Section \ref{sec:disconnected}, we apply the results from 
Section \ref{sec:conn_disconn_four} to develop an algorithm for general graphs. 
We will do it in two steps. This section
adapts the algorithm from Section \ref{sec:conn_disconn_four} so that it can handle
graphs that consist of one component with more than two vertices and perhaps a number of isolated
vertices and isolated edges. We call such a graph a \textit{chunk graph}. Section \ref{sec:disconnected} uses the algorithm for chunk graphs to obtain an algorithm for general graphs.

To be precise, let $G$ be a graph. A connected component of $G$ is called \emph{big} if it has at least three vertices, and it is called \emph{small} otherwise. Clearly, small components are either isolated vertices or isolated edges; we refer to them as $K_1$-components (isolated vertices) and $K_2$-components (isolated edges). 
A \textit{chunk graph} is a graph that has exactly one big component.

Our polynomial-time algorithm for chunk graphs needs to meet additional  restrictions on the sizes of the first and the last bags of the resulting path decompositions. Thus, we need to extend the path decomposition definition as follows. Let $\bagsize{1},\bagsize{2}$ be integers.
We  call a path decomposition $\cP = (X_1, \hdots, X_l)$ a 
\textit{$(\bagsize{1},\bagsize{2})$-path decomposition} if $|X_1| \leq\bagsize{1}$ and $|X_l| \leq\bagsize{2}$.

The main result of this section is:
\begin{theorem} \label{thm:chunkgraphs}
Let $k\in\{1,2,3\}$.
\begin{list}{}{}
\item[If: ]
for each $k'\in\{0,\ldots,k-1\}$, 
there exists a polynomial-time algorithm that, for any graph $G$
either constructs a minimum length path decomposition of width $k'$ of $G$, 
or concludes that no such path decomposition exists, 
\item[then: ]
there exists a polynomial-time algorithm that, for any given \textbf{chunk graph} $C$ and for any $\bagsize{1},\bagsize{2}\in\{2,\ldots,k+1\}$, either constructs a minimum-length $(\bagsize{1},\bagsize{2})$-path decomposition of width at most $k$ of $C$, or  concludes that no such path decomposition exists.
\end{list}
\end{theorem}

Here, 
by the minimum length $(\bagsize{1},\bagsize{2})$-path decomposition of width at most $k$ of $C$, 
we mean a path decomposition that is shortest among all
$(\bagsize{1},\bagsize{2})$-path decompositions of width at most $k$ of $C$.

We extend the notion of clean path decompositions to $(\bagsize{1},\bagsize{2})$-path decompositions.
We say that a $(\bagsize{1},\bagsize{2})$-path decomposition $\cP$ of $G$ is \emph{clean} if, among
all $(\bagsize{1},\bagsize{2})$-path decompositions of $G$, it satisfies conditions
\ref{it:clean:min-size} and \ref{it:clean:min-lex} in \ref{def:clean}.

Before continuing with the algorithm, we want to point out that, in general, a minimum- length path decomposition of a chunk graph cannot be obtained by simply constructing a minimum-length path decomposition of its big component and filling up non-full bags, if any, and prepending or appending new bags with the vertices of the small components. The following example illustrates this fact.

\begin{example*}
Consider the following graph $C$, which is the disjoint union of the 
graph $G$ in Figure~\ref{fig:example}, two isolated edges, and three isolated vertices.
\begin{figure}[ht]
\begin{center}
\includegraphics[scale=0.75]{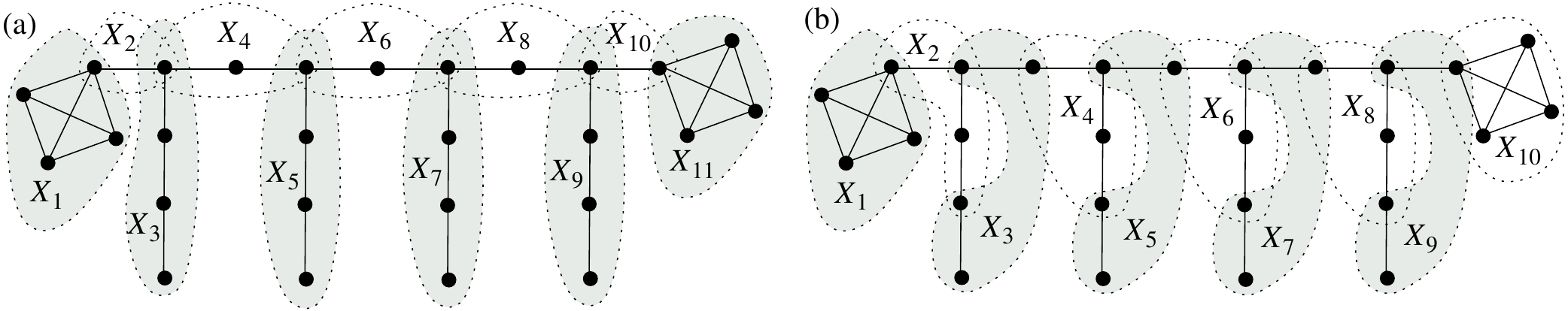}
\end{center}
\caption{(a) a path decomposition of length $11$ of $G$; (b) minimum length path decomposition of $G$}
\label{fig:example}
\end{figure}
Figure~\ref{fig:example}(b) shows a minimum length path decomposition of $G$ (its length is $10$). Since all bags are of size $4$ in this decomposition, the small components must use two additional bags, resulting in a path decomposition of $C$ of length $12$. The verification that each bag is of size $4$ in any minimum length path decomposition of $G$ is left for the reader. Now, consider a path decomposition of $G$ of length $11$ in Figure~\ref{fig:example}(a). There are two bags of size $2$ and three of size $3$ in this decomposition which readily accommodates all small components resulting in a path decomposition of $C$ of length $11$.
\end{example*}

\subsection{An algorithm for chunk graphs} \label{sec:connected_four_with_small}
\input{connected_four_with_small}

\section{$\PathLenFixed{k}$ for general graphs, $k\leq 3$} \label{sec:disconnected}
In Section~\ref{sec:one_big_component} we developed an algorithm that finds a minimum length $(\bagsize{1},\bagsize{2})$-path decomposition for a chunk graph.
We use this algorithm as a subroutine to obtain a polynomial-time algorithm for general graphs in this section.

The key idea of the algorithm for a general graph $G$ with $\pw(G)\leq 3$ is as follows. We look at a graph $G$ as a disjoint union of chunk graphs $C^1, \hdots, C^c$. Each $C^i$ consists of a big component $G_i$ and possibly some $K_1$- and $K_2$-components.  We show in \ref{lem:big_not_in_parallel} that we can limit ourselves to minimum length path decompositions of width $k\leq 3$ for $G$ where  big components are never `processed in parallel' (see Section~\ref{sec:disconnected_four} for formal definitions). It is therefore natural to construct a minimum length path decomposition of $G$ by first constructing minimum length path decompositions $\cQ_1, \hdots, \cQ_c$ of the chunk graphs $C^1, \hdots, C^c$, respectively, and then by sequencing them one after another in some order, and finally by concatenating consecutive path decompositions. Two crucial issues need to be resolved however by this approach. The first consists in how many $K_1$- and $K_2$-components to add to $G_i$ to make up a chunk graph $C^i$, $i=1,\ldots,c$. This issue is resolved by a dynamic program given in Subsection \ref{sec:disconnected_general}. The second issue consists in how to sequence and concatenate $\cQ_1, \hdots, \cQ_c$ for given chunk graphs $C^1, \hdots, C^c$, so that the resulting decomposition of $G$ has minimum length for the given $\cQ_1, \hdots, \cQ_c$. The latter is illustrated for a given order of $\cQ_1, \hdots, \cQ_c$ as follows. For a graph $G$ which breaks up into two chunk graphs $C^1$ and $C^2$ we can concatenate $\cQ_1$ and $\cQ_2$ by taking first the bags of $\cQ_1$ and then the bags of $\cQ_2$. However, if the last bag of $\cQ_1$ and the first bag of $\cQ_2$ both have size two, then we can save one bag by replacing the last bag of $\cQ_1$ and the first bag of $\cQ_2$ by their union. In general, we can save some bags by placing $\cQ_1, \hdots, \cQ_c$ in an appropriate order and by applying an appropriate concatenation. The order is dealt with in Subsection~\ref{sec:ordering}, and the concatenation in Subsection~\ref{sec:optimal_characteristics}.

The main result of this section is the following theorem that essentially reduces $\PathLenFixed{k}$ for general graphs to $\PathLenFixed{k}$ for chunk graphs, $k=1,2,3$. 

\begin{theorem} \label{thm:anygraph}
Let $k\in\{1,2,3\}$.
\begin{list}{}{}
\item[If: ]
there exists a polynomial-time algorithm that, for any chunk graph $C$, and for any $\bagsize{1},\bagsize{2}\in\{2,\ldots,k+1\}$,
either constructs a minimum length $(\bagsize{1},\bagsize{2})$-path decomposition of width $k$ of $C$, 
or concludes that no such path decomposition exists, 
\item[then:]
there exists a polynomial-time algorithm that, for \textbf{any} graph $G$, either constructs a minimum-length path decomposition of width $k$ of $G$, or concludes that no such path decomposition exists.
\end{list}
\end{theorem}

Since it is straightforward to obtain a minimum length path decomposition for graphs with no big components, we assume that the input graph has at least one big component.

\subsection{Avoiding parallel processing of big components}
 \label{sec:disconnected_four}
\input{disconnected_four}

\subsection{Type-optimal path decompositions of chunk graphs} \label{sec:optimal_characteristics}
\input{optimal_characteristics}

\subsection{A dynamic programming algorithm for general graphs} \label{sec:disconnected_general}
\input{disconnected_general}

\section{Conclusions and open problems} \label{sec:conclusions}
In this paper, we have considered a bi-criteria generalization of the pathwidth problem, where, for given integers $k,l$ and a graph $G$, we ask the question whether there exists a path decomposition $\cP$ of $G$ such that the width of $\cP$ is at most $k$ and  the length of $\cP$, is at most $l$. We have shown that the minimum length path decomposition can be found in polynomial time provided that $k\leq 3$, and that the minimum-length path decomposition problem becomes NP-hard for $k\geq 4$. Also, we have shown that the minimum-width path decomposition problem becomes NP-hard for $l\geq 2$. Though these results provide a complete complexity classification of the bi-criteria problem for general graphs,
we point out some open problems and interesting directions for further research:

\begin{itemize}[leftmargin=*]
 \item[$\circ$] The most immediate open question is the complexity status of the minimum-length path decomposition of width $k=4$ for connected graphs. Also, given our focus on the structural properties and complexity status of the special cases with fixed pathwidth parameter $k$ in this paper, our algorithms, although polynomial in the size of $G$ for $k\leq 3$, are not very efficient. Hence, an interesting and challenging open question remains about the existence of low-degree polynomial-time solutions for $\PathLenFixed{k}$ for connected and disconnected graphs and  $k\leq 3$.
 \item[$\circ$] Another research direction is the study of approximate solutions to $\PathLenFixed{k}$ and the trade-offs between the width $k$ and the length $l$. More precisely, whenever an efficient optimization algorithm for a case of $\PathLenFixed{k}$ is unlikely to exist,  it is justifiable to design approximation algorithms that find path decompositions whose width and length are within some, preferably provable, bounds from the optima.
 \item[$\circ$] Since the $\PathLenFixed{k}$ problem has appeared in a different context as a combinatorial problem motivated by an  industrial application \cite{ADGJT11}, one may search for efficient algorithms for special classes of graphs particularly relevant for this and other real-life applications.
\end{itemize}

\begin{center}{\bf Acknowledgments}\end{center}
This research has been supported by the Natural Sciences and Engineering Research Council of Canada (NSERC) Grant  OPG0105675,
and has been partially supported by Narodowe Centrum Nauki under contract DEC-2011/02/A/ST6/00201.
Dariusz Dereniowski has been partially supported by a scholarship for outstanding young researchers founded by the Polish Ministry of Science and Higher Education.

\bibliographystyle{plain}
\bibliography{search}

\end{document}

%% file: introduction.tex
The notions of pathwidth and treewidth of graphs have been introduced in a series of graph minor papers by Robertson and Seymour, starting with~\cite{RobertsonSeymour83}. Since then the pathwidth and treewidth of graphs have been receiving growing interest due to their connections to several other combinatorial problems and numerous practical applications. In particular, pathwidth is closely related to the interval thickness, the gate matrix layout problem, the vertex separation number, the node search number and narrowness for instance, see e.g. \cite{EST94,Kinnersley92,KirousisPapadimitriou85,searching_and_pebbling,KornaiTuza92,Mohring90}.

In this paper we focus on computing minimum width path decompositions whose length is minimum. More formally, the input to the decision version of this problem consists of a graph $G$ and two integers $k$ and $l$. The question is whether there exists a path decomposition $\cP$ of $G$ such that the width of $\cP$ is at most $k$ and the number of bags in $\cP$ is at most $l$. Clearly, this decision problem is NP-complete, because the pathwidth computation problem itself is an NP-hard problem~\cite{Proskurowski87} (see also \cite{BrandenburgHerrmann06}). On the other hand, it can be decided in linear time whether the pathwidth of a given graph $G$ is at most $k$ for any fixed $k$, and if the answer is affirmative, then a path decomposition of width at most $k$ can be also computed in linear time \cite{Bodlaender96,Bodlaender12}. However, as we prove in this paper, finding the minimum length path decomposition of width $k$ is an NP-hard problem for any fixed value of $k\geq 4$, which answers one of the open questions stated in \cite{BrandenburgHerrmann06}. For a detailed analysis of the complexity of the  pathwidth computation problem see e.g. \cite{Bodlaender07,Kloks94}.

To the best of our knowledge, no algorithmic results are known for the minimum length path decompositions. However, some research has been done on simultaneously bounding the diameter and the width of tree decompositions. In particular, \cite{Bodlaender89} gives a (parallel) algorithm that transforms a given tree decomposition of width $k$ for $G$ into a binary tree decomposition of width at most $3k+2$ and depth $O(\log n)$, where $n$ is the number of vertices of $G$. A more detailed analysis of the trade-off between the width and the diameter of tree decompositions can be found in \cite{BodlaenderH98}.

%% file: applications.tex
In this section,  we briefly describe selected applications of the minimum length path decomposition problem. However, the applications are not limited to those --- see also \cite{TileAssembly12} for another application.

\subsubsection*{The Partner Units Problem}

We recall the description of the Partner Unit Problem ($\PUP$) from \cite{ADGJT11}. We are given a set $S$ of sensors and a set $Z$ of zones. Each zone in $Z$ contains several sensors from $S$ and each sensor may belong to arbitrary number of zones. A feasible solution to the problem consists of a set $U$ of (control) units such that
\begin{itemize}
 \item each unit contains at most $c_1$ zones and at most $c_1$ sensors,
 \item each unit is connected to at most $c_2$ other units,
 \item each zone and each sensor belongs to exactly one unit,
 \item if $s\in S$ and $z\in Z$ belong to different units and $s$ belongs to $z$, then the two units must be connected.
\end{itemize}
From the graph-theoretic point of view, the feasible solution is a graph $G$, called \emph{unit graph}, with vertex set $S\cup Z\cup U$ such that $\{u,v\}$ is an edge of $G$ if $u\in U$ and either $v\in S\cup Z$ belongs to $u$, or $v$ is an unit that is connected to $u$. The $\PUP$ problem asks for the feasibility, i.e., whether a feasible solution exists, and if the answer is affirmative then a natural goal is to find a solution that minimizes the number of units.

The \cite{ADGJT11} considers a special case of $c_2=2$ that it solves via path decompositions. Then, the number of bags in a path decomposition corresponds to the number of units in the solution to $\PUP$.

\subsubsection*{Scheduling and register allocation}

Several applications of minimum length path decompositions can be found in operations research, in scheduling in particular. The general link between  path decomposition and scheduling is as follows. Consider a graph $G$ that represents the dependencies between non-preemptive jobs (i.e., two vertices of $G$ are adjacent if and only if their corresponding jobs are dependent). A job can start at any time but it can only be completed if all its dependent jobs have started as well yet not completed by the job's start. In other words, the execution intervals of two dependent jobs need to overlap. This requirement may be due, for instance, to the fact that some data exchanges between the jobs are required. When a job  starts, some resources necessary for its execution, for instance a processor, must be allocated to the job and thus become unavailable to other jobs until the job's completion. If the number of available processors is limited by $k$, and each job requires a single processor, then at most $k$ jobs can be executed in parallel. A schedule for the given set of jobs is \emph{feasible} if the dependencies are met and the number of jobs executed simultaneously at any given time does not exceed $k$. It can be shown that there exists a path decomposition of width $k$ if and only if there exists a feasible schedule for $G$. Here, the width of the corresponding path decomposition is directly related to the number of processors, while the length of the decomposition is related to the maximum completion time, or make-span, of all jobs.

The resources can also be the registers or the cache memory available during the execution of computer processes or database queries, see e.g. \cite{QueryOptimization01,CacheScheduling96,Sethi75}, and $k$ can be the number of registers or the size of the cache.

\subsubsection*{Graph searching games}

The problem of finding minimum width path decomposition is closely related to the problems of computing several search numbers of a graph, e.g. the node search number, the edge search number, the mixed search number or the connected search number \cite{deren_stacs11,Kinnersley92,KirousisPapadimitriou85,searching_and_pebbling,Mohring90}. Despite the fact that the number of searchers is the classical, well-motivated, and most investigated criterion for graph searching games, other criteria are also interesting. One of them is the length of search strategy. For instance, in the node searching problem, the number of moves of placing a searcher on a vertex equals $n$, the number of the vertices of $G$. However, if we allow the searchers to move simultaneously, i.e., in each step any number of searchers can be placed/removed on/from the vertices of $G$, then the length minimization of a path decomposition is equivalent to the minimization of the number of steps of the corresponding (parallel) search strategy.

For a more detailed description of graph searching games, as well as the corresponding graph width-like parameters, see e.g. \cite{Bienstock_survey91,DendrisKT97,Fomin98,guaranteed_graph_searching}.

%% file: preliminaries.tex
We now formally introduce the essential graph theoretic notation used, and the problems studied in this paper.

Let $G=(V(G),E(G))$ be a simple graph and let $X\subseteq V(G)$. We denote by $G[X]$ the subgraph \emph{induced} by $X$, i.e., $G[X]=(X,\{e\in E(G)\st e\subseteq X\})$ and by $G-X$ the subgraph obtained by removing the vertices in $X$ (together with the incident edges) from $G$, i.e., $G-X=G[V(G)\setminus X]$. Given $v\in V(G)$, $N_G(v)$ is the \emph{neighborhood} of $v$ in $G$, that is, the set of vertices adjacent to $v$ in $G$, and $N_G(X)=(\bigcup_{v\in X}N_G(v))\setminus X$ for any $X\subseteq V(G)$. We say that a vertex $v$ of $G$ is \emph{universal} if $N_G(v)=V(G)\setminus\{v\}$. A maximal connected subgraph of $G$ is called a \emph{connected component} of $G$. A graph $G$ is \emph{connected} if it has at most one connected component. Given a subgraph $H$ of $G$ we refer to the set of vertices of $H$ that have a neighbor in $V(G)\setminus V(H)$ as the \emph{border} of $H$ in $G$, and denote it by $\border_G(H)$. We define $\cC_G(X)$ to be the set of connected components $H$ of $G-X$ such that $N_G(V(H))= X$. Thus, for each $H\in \cC_G(X)$, every vertex in $X$ has a neighbor in $V(H)$. Moreover, let $\cC_G^1(X)\subseteq\cC_G(X)$ denote the set of connected components that consist of a single vertex, and let $\cC_G^2(X) = \cC_G(X)\setminus \cC_G^1(X)$.
We sometimes drop the subscript $G$ whenever $G$ is clear from the context.

For a positive integer $n$ by $K_n$ we denote a complete graph on $n$ vertices, and by $P_n$ a path graph on $n$ vertices.
For any graph $G$, any of its complete subgraphs is called a \emph{clique} of $G$. We now define a path decomposition of a graph.

\begin{definition} \label{def:path_dec}
A \emph{path decomposition} of a simple graph $G=(V(G),E(G))$ is a sequence $\cP=(X_1,\ldots,X_l)$, where $X_i\subseteq V(G)$ for each $i=1,\ldots,l$, and
\begin{enumerate}[label={\normalfont\textbf{(PD\arabic*)}},align=left,leftmargin=*]
 \item $\bigcup_{i=1,\ldots,l}X_i=V(G)$, \label{pathaxiom1}
 \item for each $\{u,v\}\in E(G)$ there exists $i\in\{1,\ldots,l\}$ such that $u,v\in X_i$, \label{pathaxiom2}
 \item for each $i,j,k$ with $1\leq i\leq j\leq k\leq l$ it holds that $X_i\cap X_k\subseteq X_j$. \label{pathaxiom3}
\end{enumerate}
The \emph{width} (respectively the \emph{length}) of the path decomposition $\cP$ is $\width(\cP)=\max_{i=1,\ldots,l}|X_i|-1$ ($\length(\cP)=l$, respectively). The \emph{pathwidth} of $G$, $\pw(G)$, is the minimum width over all path decompositions of $G$.
The \emph{size} of $\cP$, denoted by $\size(\cP)$, is given by
$\size(\cP) = \sum_{t=1}^{l} |X_t|$.
\end{definition}

We observe that condition \ref{pathaxiom3} is equivalent to the following condition:
\begin{enumerate}[label={\normalfont\textbf{(PD3')}},align=left,leftmargin=*]
 \item for each $i,k$ with $1\leq i\leq k\leq l$, if $v\in X_i$ and $v\in X_k$, then $v\in X_j$ for all $i\leq j\leq k$. \label{pathaxiom3'}
\end{enumerate}

We also make the following useful observation:
\begin{observation}\label{obs:O2}
Let  $\cP=(X_1,\ldots,X_l)$ be a path decomposition of a connected graph $G$. If $a\in X_i$ and $b\in X_j$ for some $1\leq i \leq j \leq l$, then any path $P$ between $a$ and $b$ in $G$ has a non-empty intersection with each $X_k$ for $i\leq k \leq j$.
\end{observation}

Given a simple graph $G$ and an integer $k$, in the problem $\PathDec$ (\emph{Path Decomposition}) we ask whether $\pw(G)\leq k$.

In the optimization problem $\PathLen$ (\emph{Minimum Length Path Decomposition}) the goal is to compute, for a given simple graph $G$ and an integer $k$, a minimum length path decomposition $\cP$ of $G$ such that $\width(\cP)\leq k$.
In the corresponding decision problem $\LCPD$ (\emph{Length-Constrained Path Decomposition}), a simple graph $G$ and integers $k,l$ are given, and we ask whether there exists a path decomposition $\cP$ of $G$ such that $\width(\cP)\leq k$ and $\length(\cP)\leq l$.

Finally, in the optimization problem $\PathLenFixed{k}$ the goal is to compute, for a given simple graph $G$, a minimum length path decomposition $\cP$ of $G$ such that $\width(\cP)\leq k$.
The corresponding decision problem $\LCPDFixed{k}$ the input consists of a simple graph $G$ and an integer $l$ and the question is whether there exists a path decomposition $\cP$ of $G$ such that $\width(\cP)\leq k$ and $\length(\cP)\leq l$. 

Note that the difference between $\PathLen$ and $\PathLenFixed{k}$ (and, similarly, $\LCPD$ and $\LCPDFixed{k}$) is that in the former $k$ is a part of the input while in the latter the value of $k$ is fixed.

\subsection{Overview of our results and organization of this paper}
In this paper, we  investigate the complexity of $\PathLenFixed{k}$ for different values of $k$ and $l$. We also make a distinction between connected and general (i.e. possibly disconnected) graphs. Our results are summarized in Table \ref{tbl:results}.

\begin{table}[htb]
\begin{center}
\begin{tabular}{l|c|c|}
& \textbf{Connected graphs} & \textbf{General graphs} \\
\hline
$k\leq 3$   & poly-time (\ref{thm:final}) & poly-time (\ref{thm:final}) \\
$k = 4$     & ?          & NP-hard (\ref{thm:LCPD_NPC_general}) \\
$k \geq 5$  & NP-hard (\ref{thm:LCPD_NPC_connected})    & NP-hard (\ref{thm:LCPD_NPC_general}) \\
\hline
$l\geq 2$   & \multicolumn{2}{c|}{NP-hard (\ref{thm:LCPD_NPC_l=2})} \\
\hline
\end{tabular}
\end{center}
\caption{\label{tbl:results}Our complexity results for $\PathLenFixed{k}$.}
\end{table}
Note that in all cases we fix either $k$ or $l$ (but not both). Observe that the case of $l=1$ is trivial.

In order to prove these results, we deal with the problem $\LCPDFixed{k}$ where $k$ is fixed.
We first show that $\LCPDFixed{4}$ is NP-complete for general graph and then we conclude that this implies that $\LCPDFixed{k}$ is NP-complete for all $k\geq 4$ for general graphs and for all $k\geq 5$ for connected graphs (Section \ref{sec:disconnected_five}). 

In the remainder of the paper, we construct a polynomial-time algorithm for $\PathLenFixed{3}$. We begin by showing in
Section \ref{sec:conn_disconn_four} an algorithm for 
$\PathLenFixed{k}$, $k=1,2,3$, for connected graphs. The algorithm recursively calls algorithms  for $\PathLenFixed{k'}$ for each $k'<k$ for general (possibly disconnected) graphs. We prove that the algorithm for $\PathLenFixed{k}$ is running in polynomial-time provided that the algorithms for $\PathLenFixed{k'}$ for each $k'<k$ are all polynomial-time. There is a trivial algorithm for $\PathLenFixed{0}$.

To deal with disconnected graphs we extend this algorithm to the so-called chunk graphs in Section \ref{sec:one_big_component}. A chunk graph has at most one `big' connected component with three or more vertices and all its other connected components are either isolated vertices or isolated edges. 

Finally, in Section~\ref{sec:disconnected} we show that $\PathLenFixed{k}$, $k=1,2,3$, for disconnected $G$ essentially reduces to $\PathLenFixed{k}$, $k=1,2,3$, for chunk graphs of $G$. Though each chunk graph of $G$ includes at most one big component of $G$, the isolated vertices and the isolated edges of $G$ can be distributed in many different ways between these big components to form chunk graphs of $G$. We show how to obtain an optimal distribution and thus optimal decomposition of $G$ into chunk graphs. Then,  we show how to construct a solution to $\PathLenFixed{k}$, $k=1,2,3$, for $G$ from the solutions of  $\PathLenFixed{k}$, $k=1,2,3$ for the chunk graphs of $G$.

%% file: disconnected_five.tex
\def\threepart{$3$-\textsc{partition}}

In this section we prove that the problem of finding a minimum length path decomposition of width $k\geq 4$ is NP-hard. To that end
it suffices to show that the decision problem $\LCPDFixed{4}$ is NP-complete, since this implies that $\LCPDFixed{k}$ is NP-complete for all $k\geq 4$ which follows from the following lemma.

\begin{lemma} \label{lem:NPcompletestepup}
Let $k\geq 1$. If $\LCPDFixed{k-1}$ is \textup{NP-complete} for general graphs, then $\LCPDFixed{k}$ is \textup{NP}-complete for connected graphs.
\end{lemma}
\begin{proof}
Let $k\geq 1$, let $G$ be a graph, and let $l\geq 1$. Construct an auxiliary connected graph $G'$ from 
$G$ by adding a vertex $v$ adjacent to all vertices in $V(G)$. 
We claim that the answer to $\LCPDFixed{k-1}$ is \YES~for $G,l$ if and only if the answer to $\LCPDFixed{k}$ is \YES~for $G', l$.
To see this, suppose that the answer to $\LCPDFixed{k-1}$ is \YES~for $G,l$ and let $\cP=(X_1,\hdots, X_l)$ be a path decomposition of $G$ with $\width(\cP)\leq k-1$ to witness this fact. Then, $\cP' = (X_1\cup\{v\}, \hdots, X_l\cup\{v\})$ is a path decomposition of $G'$ with $\width(\cP)\leq k$.
Conversely, suppose that the answer to $\LCPDFixed{k}$ is \YES~for $G',l$ and let $\cP'=(X_1',\hdots, X_l')$ be a path decomposition of $G'$ with $\width(\cP)\leq k$ to witness this fact. Define $s=\min\{i\in\{1,\ldots,l\}\st v\in X_i'\}$ and $t=\max\{i\in\{1,\ldots,l\}\st v\in X_i'\}$.
Then, $\cP = (X_s'\setminus\{v\}, \hdots, X_t'\setminus\{v\})$ is a path decomposition of $G$ with $\width(\cP)\leq k-1$.
\end{proof}

Thus, it remains to show that the problem $\LCPDFixed{4}$ is NP-complete. 
We prove this claim by showing a polynomial time reduction from the NP-complete \threepart~problem~\cite{GareyJohnson79} to $\LCPDFixed{4}$.
The input to \threepart~is an ordered list $S$ of $3m$ positive integers, which we will write as $S=(w_1,\ldots,w_{3m})$, and an integer $b$.
The answer to \threepart~is $\YES$ if and only if there exists a partition of the set $\{1,\hdots, 3m\}$ into $m$ sets $S_1,\ldots,S_m$ such that
\[\sum_{i\in S_j}w_i=b\textup{ for each }j=1,\ldots,m.\]
The problem remains NP-complete if we restrict the input to satisfy $b/4<w_i<b/2$ for all $i$. The latter restriction implies that for any feasible solution $S_1, \hdots, S_m$, it holds that $|S_j|=3$ for all $j=1,\ldots,m$. 

Given an instance of \threepart, we now construct a disconnected graph $G(S, b)$ in a few steps. 
In what follows, $m$ will always denote the number of required parts of the partition, i.e., $m = |S|/3$.

First, for each $i\in \{1,\hdots,3m\}$, we construct a connected graph $H_i$ as follows. Take $w_i$ copies of $K_3$, denoted by $K_3^{i,q}$, $q=1,\ldots,w_i$, and $w_i-1$ copies of $K_4$, denoted by $K_4^{i,q}$, $q=1,\ldots,w_i-1$. (The copies are taken to be mutually disjoint.) Then, for each $q=1,\ldots,w_i-1$, we identify two different vertices of $K_4^{i,q}$ with a vertex of $K_3^{i,q}$ and with a vertex of $K_3^{i,q+1}$, respectively.
This is done in such a way that each vertex of each $K_3^{i,q}$ is identified with at most one vertex from other cliques.
Thus, in the resulting graph $H_i$, each clique shares a vertex with at most two other cliques. Informally, the cliques form a `chain' in which the cliques of size $3$ and $4$ alternate. See Figure~\ref{fig:G(a)}(a) for an example of $H_i$ where $w_i=3$.

\begin{figure}[htb]
\begin{center}
\includegraphics[scale=0.8]{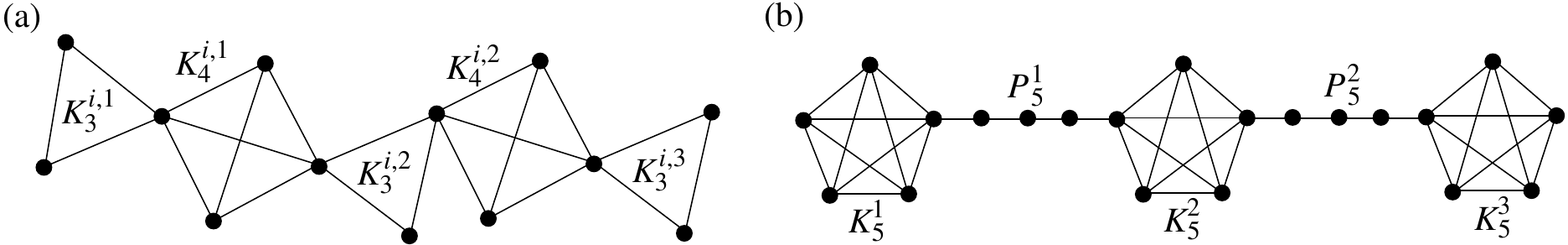}
\end{center}
\caption{(a) $H_i$, where $w_i=3$; (b) $H_{2,4}$}
\label{fig:G(a)}
\end{figure}
Second, we construct a graph $H_{m,b}$ as follows. Take $m+1$ copies of $K_5$, denoted by $K_5^1,\ldots,K_5^{m+1}$, and $m$ copies of the path graph $P_{b}$ of length $b$ ($P_b$ has $b$ edges and $b+1$ vertices), denoted by $P_{b}^1,\ldots,P_{b}^m$. (Again, the copies are taken to be mutually disjoint.) Now, for each $j=1,\ldots,m$, identify one of the endpoints of $P_{b}^j$ with a vertex of $K_5^j$, and  identify the other endpoint with a vertex of $K_5^{j+1}$. Moreover, do this in a way that ensures that, for each $j$, no vertex of $K_5^j$ is identified with the endpoints of two different paths.  See Figure~\ref{fig:G(a)}(b) for an example of $H_{2,4}$.

Let $G(S, b)$ be the graph obtained by taking the disjoint union of the graphs $H_1,\ldots,H_{3m}$ and the graph $H_{m,b}$. 
The input to the $\LCPDFixed{4}$ problem is the graph $G(S,b)$ and the integer $l=1-2m+2\sum_{i=1}^{3m} w_i$.

The subgraphs $K_3^{i,q}$, $K_4^{i,q}$ are called the \emph{cliques of} $H_i$, and the $K_5^j$ are called the cliques of $H_{m,b}$.
For brevity, all these cliques are called the \emph{cliques of} $G(S,b)$. Similarly, $P_{b}^1,\ldots,P_{b}^m$ are called the \emph{paths of} $G(S,b)$.
Observe that the number of cliques of $G(S, b)$ is exactly $l$. If $\cP$ is a path decomposition of $G(S, b)$ and a bag of $\cP$ contains all vertices of a clique of $G(S, b)$, then we say that the bag \emph{contains} this clique.

First we prove that if there exists a solution to \threepart~for the given $S$ and $b$, then there exists a path decomposition $\cP$ of $G(S, b)$ such that $\width(\cP)\leq 4$ (or, equivalently, in which all bags have size at most $5$) and $\length(\cP)=l$.
\begin{lemma} \label{lem:3-part->decomp}
If the answer to \threepart~is $\YES$ for $S$ and $b$, then the answer to $\LCPDFixed{4}$ is $\YES$ for $G(S, b)$ and $l=1-2m+2\sum_{i=1}^{3m} w_i$.
\end{lemma}
\begin{proof}
Let $S_1,\ldots,S_m$ be a solution to \threepart.
We say that a vertex of $P_{b}^j$, $j=1,\ldots,m$, at distance $d-1$ from the endpoint of $P_{b}^j$ identified with a vertex of $K_5^j$ is the \emph{$d$-th vertex of} $P_{b}^j$.
We construct a path decomposition $\cP$ as follows. 
\begin{enumerate}[label={\normalfont\textup{Step~\arabic*:}},align=left,leftmargin=*]
\item Let $\cP$ be initially the empty list. 
\item For each $j=1,\ldots,m$ do the following:
  \begin{enumerate}[label={\normalfont\textup{Step~2.\arabic*:}},align=left,leftmargin=*]
  \item Append $V(K_5^j)$ to $\cP$, and set $p:=0$.
  \item For each $i\in S_j$ do the following:
    \begin{enumerate}[label={\normalfont\textup{Step~2.2(\alph*):}},align=left,leftmargin=*]
    \item For each $q=1,\ldots,w_i$, first append $V(K_3^{i,q})\cup\{u,v\}$ to $\cP$, and if $q<w_i$, then also append $V(K_4^{i,q})\cup\{v\}$ to $\cP$, where $u$ and $v$ are the $(p+q)$-th and $(p+q+1)$-st vertices of $P_{b}^j$, respectively.
    \item Set $p:=p+w_i$.
    \end{enumerate}
  \end{enumerate}
\item Append $V(K_5^{m+1})$ to $\cP$.
\end{enumerate}

See Figure~\ref{fig:reduction_ex} for an example of this construction. It can easily be checked that, at the end of this algorithm, $\length(\cP)=m+1+\sum_{j=1}^m\sum_{i\in S_j}(2w_i-1)= l$ and, hence, $\cP$ consists of $l$ bags. Moreover, each bag has size $5$. 

Now we prove that $\cP$ satisfies \ref{def:path_dec}. First, due to Steps~2.1 and~3, some bag of $\cP$ contains $K_5^j$ for each $j\in\{1,\ldots,m+1\}$, hence every edge of $K_5^j$ appears in some bag of $\cP$. Similarly, due to Step~2.2(a), each clique $K_3^{i,q}$ and each clique $K_4^{i,q}$ appears in some bag of $\cP$, because $S_1,\ldots,S_m$ is a solution to \threepart. It follows that in order to show \ref{pathaxiom2}, it suffices to show that it holds for $e\in E(P_{b}^j)$. By definition, $\sum_{i\in S_j}w_i=|E(P_b^j)| = b$. Moreover,  $\{K_3^{i,q}\st i\in S_j, q=1,\hdots, w_i\}$ has cardinality $b$ and each of the cliques in this set together with the endpoints of a unique edge of $P_{b}^j$ form a bag of $\cP$. Therefore, $\cP$ has a bag that contains both endpoints of $e$. This proves that $\cP$ satisfies condition \ref{pathaxiom2}. Since $G(S, b)$ does not have any isolated vertices, \ref{pathaxiom1} follows.

Now we prove that $\cP$ satisfies condition \ref{pathaxiom3'} of the definition. Each vertex of a clique of $G(S, b)$ belongs to either exactly one bag or exactly two consecutive bags of $\cP$. Thus, condition \ref{pathaxiom3'} holds for such vertices. It remains to consider the internal vertices of the paths of $G(S, b)$ (their ends belongs to the cliques of $G(S, b)$). The $i$-th vertex $v\in V(P_{b}^j)$, $i=2,\ldots,b$,  either belongs to two consecutive bags of $\cP$, which occurs when the two edges incident to $v$ are in the bags together with cliques $K_3^{i,q}$ of two \emph{different} components of $G(S, b)$ (e.g. $u$ in Figure~\ref{fig:reduction_ex}), or it belongs to three consecutive bags of $\cP$, which occurs when the two edges incident to $v$ are in the bags together with two cliques $K_3^{i,q}$,  $K_3^{i,q+1}$, the case $q<w_i$  of Step~2.2(a),  of the \emph{same} component of $G(S,b)$ (e.g. $u'$ in Figure~\ref{fig:reduction_ex}). Then, $v$ is in a bag with $K_4^{i,q}$ as well.
\end{proof}

Before we continue, we give an example of the construction of the path decomposition in the proof of \ref{lem:3-part->decomp}.

\begin{example*}
Let $S=(1,1,1,2,2,3)$ (so, $m=2$) and $b=5$.  A solution to this instance of \threepart~is $S_1=\{1, 2, 6\}$ , and $S_2=\{3,4,5\}$ (clearly $w_1+w_2+w_6=w_3+w_4+w_5=5$). The graph  $H_{S,b}$, and the corresponding path decomposition $\cP$ constructed by the algorithm from the proof of \ref{lem:3-part->decomp} are given in Figure~\ref{fig:reduction_ex} (the gray color is used for some bags only to make it easier to distinguish the particular bags of this decomposition).
\end{example*}

\begin{figure}[ht]
\begin{center}
\includegraphics[scale=0.8]{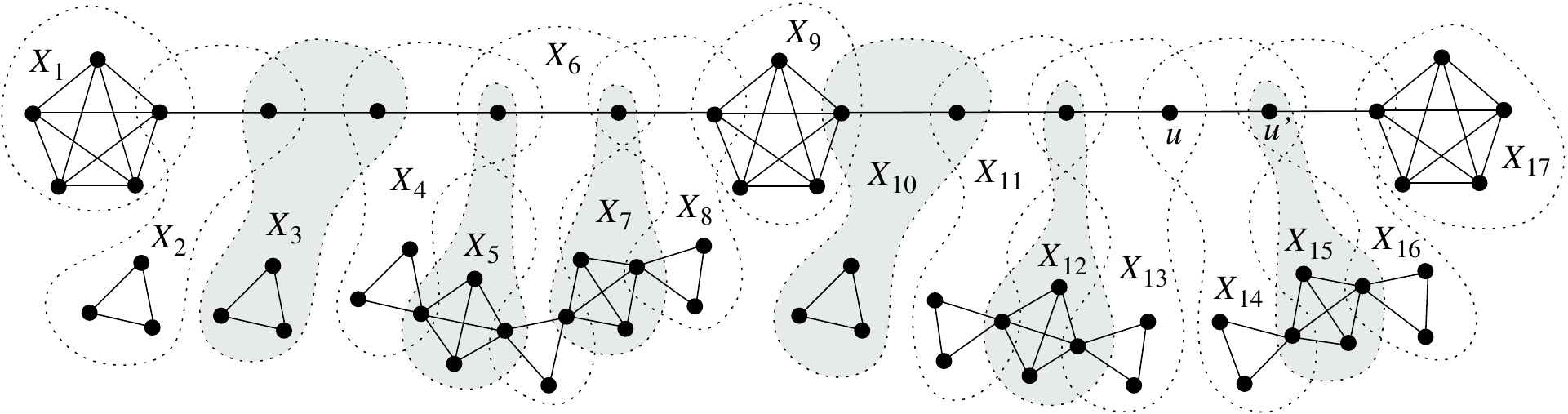}
\end{center}
\caption{$G((1,1,1,2,2,3),5)$ together with $\cP$ of width $4$ and length $17$. ($m=2, b=5$)}
\label{fig:reduction_ex}
\end{figure}

Before proving the reverse implication we need a few additional lemmas.
\begin{lemma} \label{lem:cliques_everywhere}
If $\cP$ is a path decomposition of $G(S, b)$ of width $4$ and length $l=1-2m+2\sum_{i=1}^{3m}w_i$, then each bag of $\cP$ contains exactly one clique of $G(S,b)$.
\qed
\end{lemma}
\begin{proof}
Each clique of $G(S, b)$ has size at least $3$. Moreover, any two cliques of $G(S, b)$ share at most one vertex, and no two cliques of size $3$ share a vertex.
Thus, each bag of $\cP = (X_1,\hdots, X_l)$ contains at most one clique of $G(S,b)$.  
However, it follows immediately from \ref{pathaxiom1}--\ref{pathaxiom3} that for every clique $K$ of $G(S,b)$, there exists
$i\in\{1,\hdots, l\}$ such that $V(K)\subseteq X_i$. Thus, 
since $l$ equals the number of cliques of $G(S,b)$, each bag of $\cP$ must contain exactly one clique of $G(S,b)$.
\end{proof}

We now show that we may assume without loss of generality that in a path decomposition of width $4$ of $G(S, b)$, the cliques
$K_5^1, \hdots, K_5^{m+1}$ appear in this order in the bags of the path decomposition.

\begin{lemma} \label{lem:order_on_H_mb}
Let $\cP = (X_1, \hdots, X_{l})$ be a path decomposition of width $4$ of $G(S, b)$ and let $c_1,\ldots,c_{m+1}$ be selected so that $X_{c_i}$ contains $K_5^i$ for each $i=1,\ldots,m+1$. Then, $c_1<c_2<\cdots<c_{m+1}$ or $c_1>c_2>\cdots>c_{m+1}$.
\end{lemma}
\begin{proof}
Suppose for a contradiction that the lemma does not hold. Thus, there exist $t_1,t_2,t_3$, $1\leq t_1<t_2<t_3\leq l$, such that $X_{t_i}$ contains $K_5^{j_i}$, $i=1,2,3$, where neither $j_1<j_2<j_3$ nor $j_1>j_2>j_3$. Consider the case when $j_2<j_1<j_3$ --- the other cases are analogous.  Take a shortest path $P$ between a vertex of $K_5^{j_1}$ and a vertex of $K_5^{j_3}$. Since $j_2<j_1<j_3$, $V(P)$ and $V(K_5^{j_2})$ are disjoint. By \ref{obs:O2}, there exists $v\in V(P)\cap X_{t_2}$. Thus, $X_{t_2}$ contains both $v$ and $V(K_5^{j_2})$, contrary to the fact that $|X_{t_2}|\leq \width(\cP)+1 = 5$.
\end{proof}

Moreover, the bags with the vertices of each subgraph $H_i$ form an interval of $\cP$ that falls between two cliques of $H_{m,b}$.
\begin{lemma} \label{lem:squeezing_H_i}
If $\cP$ is a path decomposition of width $4$ and length $l$ of $G(S, b)$, then for each $i \in \{1,\ldots,3m\}$ there exist $s$ and $t$ ($1\leq s<t\leq l$) such that $V(H_i) \subseteq X_s\cup\cdots\cup X_t$, $V(H_i)\cap X_p\neq\emptyset$ for each $p=s,\ldots,t$, and no clique of $H_{m,b}$ is contained in any of the bags $X_s,\ldots,X_t$.
\end{lemma}
\begin{proof}
Follows from \ref{obs:O2}, and from the facts that $\width(\cP)=4$ and $V(H_i)\cap V(H_{m,b})=\emptyset$ for each $i=1,\ldots,3m$.
\end{proof}
Finally, we have.

\begin{lemma} \label{lem:decomp->3-part}
If the answer to $\LCPDFixed{4}$ is $\YES$ for $G(S, b)$ and $l=1-2m+2\sum_{i=1}^{3m}w_i$, then the answer to \threepart~is $\YES$ for $S$ and $b$.
\end{lemma}
\begin{proof}
Let $\cP=(X_1,...,X_l)$ be a path decomposition of width $4$ and of length $l$ of $G(S,b)$.
By \ref{lem:cliques_everywhere}, each bag of $\cP$ contains exactly one clique of $G(S, b)$. Since each clique of $G(S, b)$ has size at least $3$ and $\width(\cP)=4$, no bag of $\cP$ contains the endpoints of two or more edges of a path $P_{b}^j$, $j=1,\ldots,m$, and the endpoints of any edge of $P_{b}^j$ can only share a bag with some clique $K_3^{i,q}$, $i\in\{1,\ldots,3m\}$, $q\in\{1,\ldots,w_i\}$.
Moreover, the total number of edges of the paths of $G(S, b)$ equals $mb$ ($mb$ is also the number of cliques $K_3^{i,q}$), which implies that the endpoints of each edge of each path of $G(S, b)$ share a bag with a unique $K_3^{i,q}$ for some $i\in\{1,\ldots,3m\}$ and $q\in\{1,\ldots,w_i\}$.

Let $X_{c_j}=V(K_5^j)$ for each $j=1,\ldots,m+1$.
By \ref{lem:order_on_H_mb}, $c_1<c_2<\cdots<c_{j+1}$ or $c_1>c_2>\cdots>c_{j+1}$.
Since $(X_l,\ldots,X_1)$ is a path decomposition of $G(S,b)$, we may assume without loss of generality that the former occurs.
Then, the endpoints of all $b$ edges of a path $P_{b}^j$, $j=1,\ldots,m$, must be included in the bags $X_{c_j+1},\ldots,X_{c_{j+1}-1}$ for otherwise the connectedness of $P_{b}^j$ and \ref{obs:O2} would imply a vertex of $P_{b}^j$ in
either $X_{c_j}$ or $X_{c_{j+1}}$ or both which results in a path decomposition of width at least 5, a contradiction.
Therefore, exactly $b$ cliques in $\{K_3^{i,q}\st i\in\{1,\ldots,3m\},q\in\{1,\ldots,w_i\}\}$ must be included in the bags $X_{c_j+1},\ldots,X_{c_{j+1}-1}$.
Moreover, \ref{lem:squeezing_H_i} implies that for each $i\in\{1,\ldots,3m\}$ there exists $j\in\{1,\ldots,m\}$ such that $V(H_i)\subseteq X_{c_j+1}\cup\cdots\cup X_{c_{j+1}-1}$. 
Define for each $j=1,\ldots,m$
\[S_j=\{i\in \{1,\ldots,3m\}\st V(H_i)\subseteq X_{c_j+1}\cup\cdots\cup X_{c_{j+1}-1}\}.\]
Due to the above arguments, $\sum_{i\in S_j}w_i=b$ for each $j=1,\ldots,m$. Therefore, the answer to \threepart~is $\YES$.
\end{proof}

\ref{lem:NPcompletestepup}, \ref{lem:3-part->decomp} and \ref{lem:decomp->3-part} imply the following:
\begin{theorem} \label{thm:LCPD_NPC_general}
The problem $\LCPDFixed{k}$ is \textup{NP}-complete for each $k\geq 4$.
\qed
\end{theorem}
Together with \ref{lem:NPcompletestepup}, this gives in addition the following theorem:
\begin{theorem} \label{thm:LCPD_NPC_connected}
The problem $\LCPDFixed{k}$ is \textup{NP}-complete for each $k\geq 5$, when the input is restricted to connected graphs.
\qed
\end{theorem}

We finish this section with a remark on the complexity of $\LCPD$ when $l$ is fixed.
The following theorem is a direct consequence of the NP-completeness of the vertex separator problem defined in~\cite{Gustedt93}.
\begin{theorem} \label{thm:LCPD_NPC_l=2}
The problem $\LCPD$ is \textup{NP}-complete for the given $G$, $k$ and $l\geq 2$.
\qed
\end{theorem}

%% file: generic_alg.tex
Let $G$ be a graph. We say that $\cP = (X_1, \hdots, X_{l})$ is a \textit{partial path decomposition} of $G$ if
\begin{enumerate}[label={\normalfont(\roman*)},align=left]
\item for each $\{u,v\}\in E\left(G\left[\bigcup_{i=1}^l X_i\right]\right)$, there exists $i\in \{1, \hdots, l\}$ such that $u,v\in X_i$, and
\item for each $i,j,k$ with $1\leq i\leq j\leq k\leq l$, it holds that $X_i\cap X_k\subseteq X_j$.
\end{enumerate}
Define $\pdspan(\cP) = \bigcup_{i=1}^l X_i$ to be the \textit{span} of $\cP$ and denote $G_{\cP} = G[\pdspan(\cP)]$.
$G_{\cP}$ is called the subgraph of $G$ \textit{covered by $\cP$}.

It follows that $\cP$ is a path decomposition of the induced subgraph $G_{\cP}$ and $V(G_{\cP})=\pdspan(\cP)$.
Notice that $G_\cP = G$ if and only if $\cP$ is a path decomposition of $G$.
Also note that any prefix of a path decomposition of $G$ is a partial path decomposition of $G$.

We say that a partial path decomposition $\cP =(X_1, \hdots, X_{l})$ \textit{extends to} a partial path decomposition $\cP' = (X'_1, \hdots, X'_{l'})$, with $l'\geq l$, if $X_i = X'_i$ for all $i \in\{1,\hdots, l\}$. We define the \textit{frontier} of $\cP$ to be $\delta(\cP) = \{x\in V(G_\cP)\st x\mbox{ has a neighbor in }V(G)\setminus V(G_\cP)\}$.

Consider the following generic and potentially exponential-time algorithm for finding a minimum-length path decomposition of width at most $k$ for a given graph $G$. We construct an auxiliary directed graph $\cG_k$ whose vertices are pairs $(F, X)$, where $F$ is an induced subgraph of $G$, $X\subseteq V(F)$, and $|X| \leq k+1$. Each pair $(F, X)$ represents the (perhaps empty) collection $\mbP(F, X)$ of all partial path decompositions $\cP=(X_1,\ldots,X_{\length(\cP)})$ of width at most $k$ that have the common property that $G_\cP = F$ and $X_{\length(\cP)} = X$ (i.e., the subgraph of $G$ covered by $\cP$ is $F$ and the last bag of $\cP$ is $X$). Notice that the partial path decompositions within $\mbP(F, X)$ may be of different lengths. There is an arc from $(F, X)$ to $(F', X')$ in $\cG_k$ if and only if \textit{every} partial path composition in $\mbP(F, X)$ extends to \textit{some} partial path decomposition in $\mbP(F', X')$ by adding exactly one bag, namely $X'$. We will also add to $\cG_k$ a special source vertex $s$ and a sink vertex $t$. There is an arc from the source vertex $s$ to every pair $(F, X)$ with $F = G[X]$, and an arc from every pair $(F, X)$ with $F = G$ to the sink vertex. For convenience, let $\mbP(s)$ contain the path decomposition of length $0$. We have the following result:

\begin{claim} \label{claim:generic_path_correspondence}
It holds $\cP\in\mbP(G,X)$ for some $X\subseteq V(G)$, where $\length(\cP)\leq l$, if and only if there exists a directed $s$-$t$ path in $\cG_k$ of length at most $l+1$.
\end{claim}
\begin{proof}
We argue by induction on $l$ that there exists a directed $s$-$v$ path of length $l$ in $\cG_k$, where $v=(F,X)$ is a vertex of $\cG_k$, if and only if a partial path decomposition of length $l$ belongs to $\mbP(F,X)$.
In the base case of $l=0$ we have that $v=s$ and the claim follows directly from the definition of $\mbP(s)$.
Suppose now that the induction hypothesis holds for some $l\geq 0$.

Let $P$ be an $s$-$v$ path of length $l+1$ in $\cG_k$.
Let $(v',v)$ be the last arc of $P$.
Let $P'$ be the path $P$ without $(v',v)$. Thus, $P'$ is an $s$-$v'$ path of length $l$ in $\cG_k$.
By the induction hypothesis, there exists a partial path decomposition $\cP'\in\mbP(v')$ of length $l$. Since $(v',v)\in E(\cG_k)$, $\cP'$ extends to some partial path decomposition $\cP\in\mbP(F,X)$ by adding exactly one bag, namely $X$. Thus $\length(\cP)=\length(\cP')+1$ as required.

Suppose now that $\cP\in\mbP(F,X)$ and $\length(\cP)=l+1$. Let $\cP=(X_1,\ldots,X_{l+1})$. By definition of $\cG_k$, $F=G_\cP$ and $X_{l+1}=X$. Thus, $v=(G_\cP,X_{l+1})$.
Now, let $\cP'=(X_1,\ldots,X_l)$ and $v'=(G_{\cP'},X_l)$.
Since every partial decomposition in  $\mbP(G_{\cP'},X_l)$ extends to some  partial decomposition in  $\mbP(G_{\cP},X_{l+1})$ by adding $X_{l+1}$, $(v',v)\in E(\cG_k)$.
By the induction hypothesis, there exists an $s$-$v'$ path $P'$ in $\cG_k$ of length $l$.
Then, $P'$ together with the arc $(v',v)$ forms the desired $s$-$v$ path of length $l+1$ in $\cG_k$.
\end{proof}

Having constructed $\cG_k$, we find a shortest path from $s$ to $t$ in $\cG_k$. By \ref{claim:generic_path_correspondence} this path (let its consecutive vertices be $s,(F_1,X_1),\ldots,(F_l,X_l),t$) corresponds to a path decomposition $(X_1,\ldots,X_l)$ of $G$. This clearly gives an exponential-time algorithm for finding a minimum length path decomposition. In the reminder of this section, we will turn this algorithm into a polynomial-time one by redefining the graph $\cG_k$ and reducing its size for connected $G$ and $k\in\{1,2,3\}$. 

\subsection*{Adapting the generic algorithm to ensure polynomial running time}

The aim of this section is to introduce some intuition on the construction of $\cG_k$ whose size is bounded by a polynomial in the size of $G$.
The formal definition of $\cG_k$ is given in Section~\ref{sec:connected_four}.

Our approach is to represent the pairs $(F, X)$ in an alternative way. We encode the graph $\cG_k$ in such a way that for a fixed set $X\subseteq V(G)$ the number of vertices of $\cG_k$ of the form $(F,X)$ is polynomially bounded. Since $|X| \leq k+1$ and $k$ is fixed, this reduces the number of vertices to a number bounded by a polynomial in the size of $G$. As a result, however, some path decompositions of $G$ no longer have corresponding $s$-$t$ paths in $\cG_k$ though we prove that minimum length path decompositions of $G$ still have corresponding shortest paths in $\cG_k$. The alternative encoding is as follows. 

Whenever possible, we represent a pair $(F, X)$ by a pair $(X,R(X))$, where $R(X)$ is a function that maps each non-empty subset $S\subseteq X$ into a triple $(\mbG,f,l)$ with the following properties:
\begin{enumerate}[label={\normalfont(\arabic*)},leftmargin=*]
\item $\mbG$ is a set with at least $\max\left\{|\cC_G^2(S)|-12,0\right\}$ components in $\cC_G^2(S)$,
\item $f\colon\cC_G^2(S)\rightarrow\{0,1\}$ is a function such that $f(H)=f(H')$ for all $H,H'\in\mbG$, and $f(H)=1$ implies that $V(H)\subseteq V(F)$,
\item $l = \left|\bigcup_{C\in \cC_G^1(S)} V(C)\cap V(F)\right|$, i.e., $l$ equals the number of single-vertex components in $\cC_G^1(S)$ that are covered by any partial path decomposition in $\mbP(F,X)$.
\end{enumerate}
We will show that it is possible to reconstruct $(F, X)$ from such a pair $(X, R(X))$.
The vertex set of our final auxiliary graph $\cG_k$ consists of all pairs $(X,R(X))$.
Since the number of such pairs is bounded by a polynomial in the size of $G$, we lose many vertices from the generic graph.
The arcs of $\cG_k$ will have weights and, informally speaking, the weight of an arc $(v,v')$ equals the number of bags that are added while extending any partial path decomposition in $\mbP(v)$ to a partial path decomposition in $\mbP(v')$.
Hence, contrary to the generic construction of $\cG_k$, some arcs in the new directed graph introduce several bags of a path decomposition.
We will show that the vertices that we drop from the generic graph are irrelevant for the length minimization, and that $\cG_k$ has the property that there exists a clean path decomposition  $\cP$ with $\length(\cP)\leq l$ if and only if there exists an $s$-$t$ path in $\cG_k$ of length at most $l$. (See \ref{lem:path_gives_decomposition} and \ref{lem:all_paths_present}.) The clean path decompositions are defined in Section~\ref{sec:clean} where we also observe that among all minimum length path decompositions of $G$ there always exists one that is clean.

%% file: properties.tex
\subsection{Bottleneck sets and bottleneck intervals}\label{sec:bottleneck}
In this section we consider a path decomposition $\cP=(X_1,\ldots,X_l)$ of a connected graph $G$, and a fixed set $S\subseteq V(G)$. 
Recall that $\cC_G(S)$ is the collection of connected components $H$ of $G-S$ such that every vertex in $S$ has a neighbor in $V(H)$. 
The elements of $\cC_G(S)$ are called $S$-\textit{components}. We call the components in $\cC_G^2(S)$ \textit{$S$-branches}, while the vertices of the graphs in $\cC_G^1(S)$ are called \textit{$S$-leaves}. Finally, let $|\cC_G^2(S)|=c$.

We begin by investigating the relative order in which the vertices of $S$ and $S$-branches appear in, and disappear from, the bags in $\cP=(X_1,\ldots,X_l)$.
For a (connected or disconnected) subgraph $H$  of $G$, define 
\[
\alpha_\cP(H)=\min\{i\st X_i\cap V(H)\neq\emptyset\}\quad\mbox{and}\quad 
\beta_\cP(H)=\max\{i\st X_i\cap V(H)\neq\emptyset\}.
\]
For $v\in V(G)$, we abbreviate $\alpha_\cP(G[\{v\}])$ and $\beta_\cP(G[\{v\}])$ as $\alpha_\cP(v)$ and $\beta_\cP(v)$, respectively. For convenience we define
\[
\maxalpha_\cP(S)=\max_{v\in S}\{\alpha_{\cP}(v)\}\quad\mbox{and}\quad 
\minbeta_\cP(S)=\min_{v\in S}\{\beta_{\cP}(v)\},
\]
for a non-empty $S\subseteq V(H)$.
Hence, $S\subseteq X_i$ if and only if $i\in\{\maxalpha_{\cP}(S),\ldots,\minbeta_{\cP}(S)\}$.
Whenever  $\cP$ is clear from the context, we drop it as a subscript. Clearly, we have 
\begin{equation} \label{eqn:VHalphabeta}
V(H)\subseteq X_{\alpha(H)}\cup\cdots\cup X_{\beta(H)}. 
\end{equation}
Informally, $X_{\alpha(H)}$ can be interpreted as the first bag that contains a vertex of $V(H)$, and $\alpha(H)$ as the start of $H$. Similarly, $X_{\beta(H)}$ is the last bag that contains a vertex of $V(H)$, and the $\beta(H)$ as the completion of $H$. By \ref{obs:O2}, if $H$ is connected, then any bag between these first and last bags must contain a vertex of $V(H)$. By definition, the converse is also true: no bag $X_i$ with $i$ outside the interval $[\alpha(H),\beta(H)]$ contains a vertex of $V(H)$.

The following lemma relates the start and the completion of an $S$-branch $H$ and the start and completion of $x\in S$.
\begin{lemma} \label{lem:basic2}
Let $G$ be a graph and let $\cP=(X_1, \hdots, X_l)$ be a path decomposition of $G$. Let $S\subseteq V(G)$, let $x\in S$ and let $H$ be an $S$-branch.
Then, the following statements hold:
\begin{enumerate}[label={\normalfont(\roman*)},align=left]
\item If $\alpha(x) \leq \alpha(H)$, then $x\in X_i$ for all $\alpha(x) \leq i \leq \alpha(H)$;
\item $\alpha(x)\leq\beta(H)$;
\item If $\beta(x) \geq \beta(H)$, then $x\in X_i$ for all $\beta(H)\leq i\leq \beta(x)$;
\item $\alpha(H)\leq\beta(x)$.
\end{enumerate}
\end{lemma}
\begin{proof}
By definition of an $S$-branch, there exists $v\in V(H)$ that is adjacent to $x$. Since $\{v,x\}\in E(G)$, it follows from \ref{pathaxiom2} that there exists $t$ such that $\{v, x\}\subseteq X_t$. Clearly, $\alpha(H)\leq t$.
To prove (i), note that the assumption $\alpha(x) \leq \alpha(H)$ implies that $\alpha(x)\leq t$. Since $x\in X_{\alpha(x)}$ and $x\in X_t$, it follows from \ref{pathaxiom3'} that $x\in X_i$ for all $\alpha(x)\leq i\leq t$. In particular, $x\in X_i$ for all $\alpha(x)\leq i\leq \alpha(H)$, as required. 
To prove (ii), observe that $x\in X_t$ implies that $\alpha(x) \leq t$, and $v\in X_t$ implies that $t \leq \beta(H)$. Thus,
$\alpha(x) \leq \beta(H)$.

Parts (iii) and (iv) follow from (i) and (ii) applied to $\cP'=(X_l, \hdots, X_1)$.
\end{proof}

Some $S$-branches can start even before the whole $S$ appears in the bags of $\cP$, see $X_1$ in Figure \ref{fig:Sbranches}, also
some $S$-branches can complete even after the whole $S$ no longer appears in the bags of $\cP$, see $X_{13}$ in Figure \ref{fig:Sbranches}. We now show that in either case this can only happen for a few $S$-branches whose number is limited by $k$. To that end, we adopt a convention that $H_1, \hdots, H_c$ denote the $S$-branches in $\cC_G^2(S)$ start-ordered so that $\alpha(H_1)\leq \alpha(H_2)\leq\cdots \leq \alpha(H_c)$, and $H^1, \hdots, H^c$ denote these $S$-branches completion-ordered so that $\beta(H^1)\leq \beta(H^2)\leq \cdots\leq \beta(H^c)$. We have the following lemma.

\begin{lemma} \label{lem:atmost3}
Let $G$ be a graph, let $\cP=(X_1, \hdots, X_l)$ be a path decomposition of width $k$ of $G$, and let $S\subseteq V(G)$, where $c>k$.
Then, the following two statements hold:
\begin{enumerate}[label={\normalfont(\roman*)},align=left]
\item $\alpha(H_i)\geq \alpha(x)$ for all $i\geq k+1$ and all $x\in S$,
\item $\beta(H^i)\leq \beta(x)$ for all $i\leq c-k$ and all $x\in S$.
\end{enumerate}
\end{lemma} 
\begin{proof}
For (i), since $\alpha(H_1)\leq \alpha(H_2) \leq \cdots \leq \alpha(H_c)$, it suffices to show that $\alpha(H_{k+1})\geq \alpha(x)$. 
So suppose for a contradiction that $\alpha(H_{k+1})<\alpha(x)$. Hence, $\alpha(H_j) < \alpha(x)$ for all $j=1,\hdots, k+1$.
It follows from \ref{lem:basic2}(ii) that, for each $j=1,\hdots,k+1$, there exists $v_j\in V(H_j)\cap X_{\alpha(x)}$. In particular, it follows that  $\{v_1, \hdots, v_{k+1}, x\}\subseteq X_{\alpha(x)}$, implying that $|X_{\alpha(x)}|> k+1$, contrary to the fact that $\cP$ has width $k$. This proves (i). Next, (ii) follows from (i) applied to $\cP'=(X_l, \hdots, X_1)$.
\end{proof}

We now focus on $S$-branches $H$ such that $\alpha(x)< \alpha(H)\leq \beta(H)< \beta(x)$ for all $x\in S$, since by \ref{lem:atmost3} there is only a constant number of branches that do \emph{not} meet this condition for fixed $k$. As we can not guarantee the existence of these $S$-branches for $S$ with small $c$, we limit ourselves to special sets $S$ refereed to as bottlenecks. 

A set $S\subseteq V(G)$ is a \textit{bottleneck set} if $S\neq\emptyset$ and $c \geq 13$.
We denote by $\cS$ the collection of all bottleneck sets of $G$.
Note that if $\pw(G)\leq 1$, then $G$ has no bottleneck sets.

\begin{example*}
To illustrate this concept, consider the graph $G$ in Figure \ref{fig:Sbranches}.
The sets $X_1, \hdots, X_{13}$ form a path decomposition of $G$.
The set $S=\{x\}$ is a bottleneck set.
There are $13$ $S$-branches, namely, the connected components of $G-\{x\}$.
Notice that $G-\{x\}$ has no components consisting of exactly one vertex, and therefore there are no $S$-leaves.
\end{example*}
\begin{figure}[htb]
\begin{center}
\includegraphics[scale=0.75]{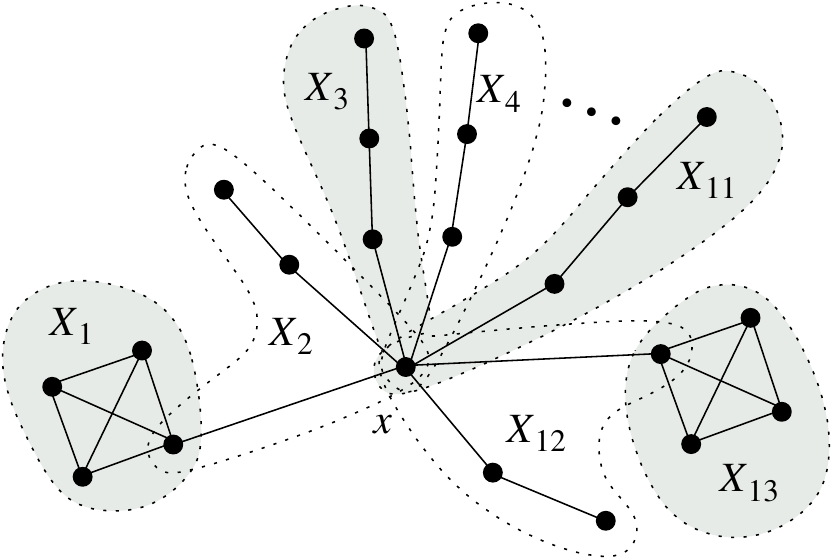}
\end{center}
\caption{\label{fig:Sbranches}Illustration of a bottleneck set $S=\{x\}$ and the corresponding $S$-branches.}
\end{figure}

The bottleneck sets are key for a couple more reasons. First, if $G$ has no bottleneck set, then the size of the auxiliary generic graph from Section~\ref{sec:general_alg} can be easily bounded by a polynomial in the size of $G$ since the number of $S$-branches is then bounded by a constant for any $S$. On the other hand,  if $G$ contains even a single bottleneck set $S$, then the number of vertices $(F,X)$ such that $S\subseteq X$ in $\cG_k$ can be exponential.
This follows from an observation that the number of induced subgraphs $F$ with $S\subseteq\border_G(F)$ is exponential in $|\cC_G^2(S)|$. However, we prove that all, except a constant number, $S$-branches in $\cC_G^2(S)$ are in consecutive bags $X_i,\ldots,X_j$ such that $S\subseteq X_p$ for each $p=i-1,\ldots,j+1$. We refer to the interval between $i$ and $j$ as the \textit{bottleneck interval} of $S$ and formally define it later. Since the $S$-branches in the bottleneck interval of $S$ always share bags with the whole $S$ we can recursively reduce the computation of a minimum length path decomposition of width $k$ of $G$ to the computation of a minimum length path decomposition of width $k-|S|$ for the branches in the bottleneck interval of $S$. This is another key reason behind the bottleneck sets.

For any bottleneck set $S$ and a path decomposition $\cP$, we define $I_{\cP}(S) = \{t_1(S), \hdots, t_2(S)\}$,
where  $t_1(S)$ and $t_2(S)$ are as follows:
\[t_1(S)  = \min \left\{ \alpha(H)\st H \mbox{ is an $S$-branch and } S\subseteq X_{\alpha(H)}\cap X_{\alpha(H)-1} \right\},\]
\[t_2(S)  = \max \left\{ \beta(H)\st H \mbox{ is an $S$-branch and } S\subseteq X_{\beta(H)}\cap X_{\beta(H)+1} \right\}.\]
We call $I_\cP(S)$ the \textit{bottleneck interval} associated with the set $S$.
Informally, let $X_i,\ldots,X_j$ be all bags in $\cP$ such that each of them contains $S$. Then, $t_1(S)$ is the start of the earliest $S$-branch to start in $\{i+1,\ldots,j-1\}$ and $t_2(S)$ is the completion of the latest $S$-branch to complete in $\{i+1,\ldots,j-1\}$.
For example, in Figure \ref{fig:Sbranches}, we have $t_1(S)=3$ and $t_2(S)=11$ for the bottleneck set $S= \{x\}$.

Notice that $I_\cP(S)$ depends on the path decomposition $\cP$.
Whenever $\cP$ is clear from the context we write $I(S)$ instead of $I_{\cP}(S)$. 
We show in \ref{lem:decomposition_structure1} that $t_1(S)$ and $t_2(S)$ are well defined and that $t_1(S)\leq t_2(S)$, which implies that $I(S)$ is non-empty.
Given the bottleneck interval $I(S)$, we color each $S$-branch
$H$ as follows: (see also \figref{fig:coloring})
\begin{itemize}[leftmargin=*]
\item color $H$ \textit{green} if $t_1(S) \leq \alpha(H) \leq \beta(H) \leq t_2(S)$;
\item color $H$ \textit{red} if $\alpha(H) < t_1(S) \leq \beta(H) \leq t_2(S)$;
\item color $H$ \textit{blue} if $t_1(S) \leq \alpha(H) \leq t_2(S) < \beta(H)$; 
\item color $H$ \textit{purple} if $\alpha(H) < t_1(S) \leq t_2(S) < \beta(H)$;
\item color $H$ \textit{gray} if $\beta(H)<t_1(S)$;
\item color $H$ \textit{black} if $\alpha(H)>t_2(S)$.
\end{itemize}
\begin{figure}[htb]
\begin{center}
\includegraphics{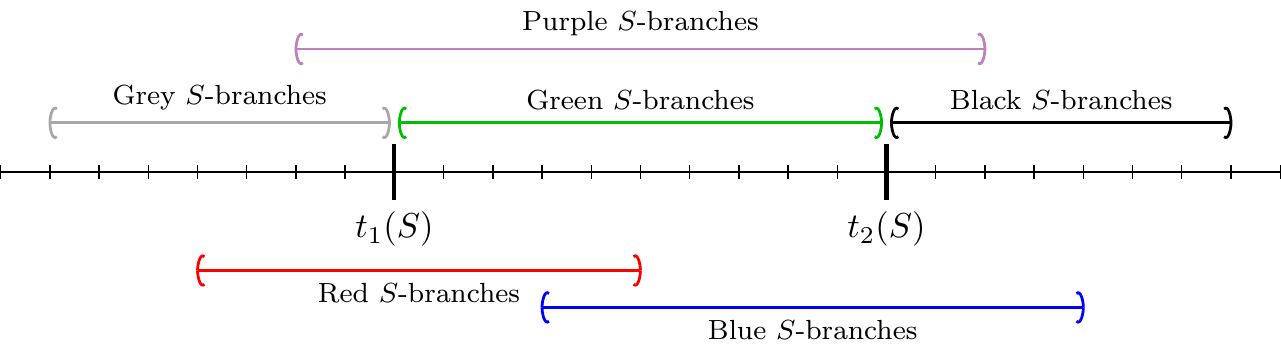}
\end{center}
\caption{\label{fig:coloring}Coloring of $S$-branches.}
\end{figure}

There are exactly two gray $S$-branches, exactly two black branches, and
the remaining branches are green for the bottleneck set $S= \{x\}$ in the graph $G$ in Figure \ref{fig:Sbranches}.

Since $t_1(S) \leq t_2(S)$ and $\alpha(H)\leq \beta(H)$, each $S$-branch is assigned exactly one color. 
Notice also that there exist $S$-branches $H$ and $H'$ (possibly equal) such that $\alpha(H) = t_1(S)$ and $\beta(H') = t_2(S)$.

\begin{lemma} \label{lem:decomposition_structure1}
Let $G$ be a connected graph, let $\cP=(X_1, \hdots, X_l)$ be a path decomposition of width $k\leq 3$ of $G$, and $S\in \cS$. Then, $I(S)$ is well-defined and non-empty,  and:
\begin{enumerate}[label={\normalfont(\roman*)},align=left]
 \item there is at least one green $S$-branch; \label{lem:decomposition_structure1_green}
 \item the number of $S$-branches colored red, purple or gray is at most $2k$; \label{lem:decomposition_structure1_rpg}
 \item the number of $S$-branches colored blue, purple or black is at most $2k$. \label{lem:decomposition_structure1_bpb}
\end{enumerate}
\end{lemma}
\begin{proof}
Let $H_1, H_2, \hdots, H_c$ be the $S$-branches start-ordered. Since $S$ is a bottleneck set, we have $c \geq 13$.
We first claim that
\begin{equation} \label{eq:H_q_start}
\alpha(G[S]) \leq \alpha(H_q) \leq \beta(G[S]) \textup{ for all } q\in \{k+1, \hdots, c\}.
\end{equation}
Let $q \in \{k+1, \hdots, c\}$ be selected arbitrarily.
By \ref{lem:atmost3}(i), $\alpha(x)\leq \alpha(H_{q})$ for each $x\in S$.
Thus, by \ref{lem:basic2}(i), $x\in X_{\alpha(H_{q})}$ for all $x\in S$. 
Hence, $S\subseteq X_{\alpha(H_{q})}$, and \eqnref{eq:H_q_start} follows.

Second, we claim that
\begin{equation} \label{eq:t_1_ok}
\alpha(G[S]) < \alpha(H_q) \leq \beta(G[S])\textup{ for all }q\in \{2k+1, \hdots, c\}.
\end{equation}
To prove this, it suffices to show that $\alpha(H_{2k+1}) > \alpha(G[S])$.
Suppose otherwise, i.e., $\alpha(H_{2k+1}) \leq \alpha(G[S])$.
Since, by \eqnref{eq:H_q_start}, $\alpha(H_{k+1})\geq\alpha(G[S])$, we obtain that $\alpha(H_{k+1}) = \cdots = \alpha(H_{2k+1}) = \alpha(G[S])$.
This implies that $X_{\alpha(G[S])}$ contains a vertex of each $S$-branch $H_{q}$ with $q\in \{k+1, \hdots,2k+1\}$.
Thus, $|X_{\alpha(G[S])}| \geq |S| + k + 1$, contrary to the fact that $|X_{\alpha(G[S])}| \leq  k+1$. 
This proves \eqnref{eq:t_1_ok}.

By applying this argument to $\cP'=(X_l, \hdots, X_1)$, we conclude that
\begin{equation} \label{eq:t_2_ok}
\alpha(G[S]) \leq \beta(H^q) < \beta(G[S])\textup{ for all }q\in \{1, \hdots, c-2k\}.
\end{equation}
Let $\cA = \cC_G^2(S)\setminus (\{H_1, \hdots, H_{2k}\}\cup \{H^{c-2k+1}, \hdots, H^c\})$. 
By \eqnref{eq:t_1_ok} and \eqnref{eq:t_2_ok}, $\alpha(G[S]) < \alpha(H) \leq \beta(H) < \beta(G[S])$ for all $H\in \cA$, and $|\cA| \geq c - 4k \geq 1$.
Therefore, $t_1(S)$ and $t_2(S)$ are well-defined, and satisfy
\[t_1(S) \leq \min\{\alpha(H)\colon H\in \cA\} \mbox{ and } t_2(S) \geq \max\{\beta(H)\colon H\in \cA\}.\]
It trivially follows that $t_1(S) \leq t_2(S)$ and hence $I(S)$ is non-empty.
For (i), notice that any $H\in \cA$ receives the color green.
Finally, \eqnref{eq:t_1_ok} implies \ref{lem:decomposition_structure1_rpg}, while \eqnref{eq:t_2_ok} gives \ref{lem:decomposition_structure1_bpb}.
\end{proof}

\begin{example*}
Consider the graph $G$ in Figure \ref{fig:bottleneck}. 
$G$ has three bottleneck sets, namely $\{s_1\}$, $\{s_2\}$ and $\{s_1,s_2\}$.
For the bottleneck $\{s_1\}$, we have $\{s_1\}\subseteq X_i$ for each $i=4,\ldots,38$.
Then, $t_1(\{s_1\})=5$, because $\alpha(H_5)=5$, and $t_2(\{s_1\})=37$, because $\beta(G[\{s_2\}\cup V(H_{13})\cup\cdots\cup V(H_{37})])=37$.
Thus, all $\{s_1\}$-branches are green except for $H_1,\ldots,H_4$, which are either gray  ($H_1,H_2,H_3$) or black ($H_4$).
For $\{s_1,s_2\}$ we have: $\{s_1,s_2\}\subseteq X_i$ for each $i=13,\ldots,37$, $t_1(\{s_1,s_2\})=14$ and $t_2(\{s_1,s_2\})=25$.
The branch $H_{13}$ is gray and the remaining $\{s_1,s_2\}$-branches, namely $H_{14},\ldots,H_{25}$ are green.
Thus, $I(\{s_1,s_2\})\subseteq I(\{s_1\})$.
Finally, $\{s_2\}\subseteq X_i$ for each $i=13,\ldots,37$, $t_1(\{s_2\})=26$ and $t_2(\{s_2\})=36$.
Hence, $I(\{s_2\})\subseteq I(\{s_1\})$ and $I(\{s_2\})\cap I(\{s_1,s_2\})=\emptyset$.
The green components in $\cC_G^2(\{s_2\})$ are $H_{26},\ldots,H_{36}$, the component $H_{37}$ is black, and the $\{s_2\}$-branch $G[\{s_1\}\cup V(H_1)\cup\cdots\cup V(H_{25})]$ is purple.
\end{example*}

By definition, any bottleneck $S$ has at least $13$ $S$-branches. The proof of \ref{lem:decomposition_structure1} shows that this number guarantees the existence of green $S$-branches for a bottleneck $S$. In the next section, we show that we can limit ourselves to the a special class of path decompositions, referred to as clean path decompositions, which have no red and no blue $S$-branches for any bottleneck $S$. However, gray and black $S$-branches are unavoidable since it may happen that $I(S)\varsubsetneq\{\alpha(G[S]),\ldots,\beta(G[S])\}$ as is the case in the example in Figure~\ref{fig:bottleneck}. We also remark that the restriction to clean path decompositions would make it possible to consider bottleneck sets as those having at least $7$ (rather than $13$) $S$-branches. Though this would improve the complexity of our polynomial-time algorithm, we do not make this attempt to optimize its running time.
\begin{figure}[htb]
\begin{center}
\includegraphics[scale=0.8]{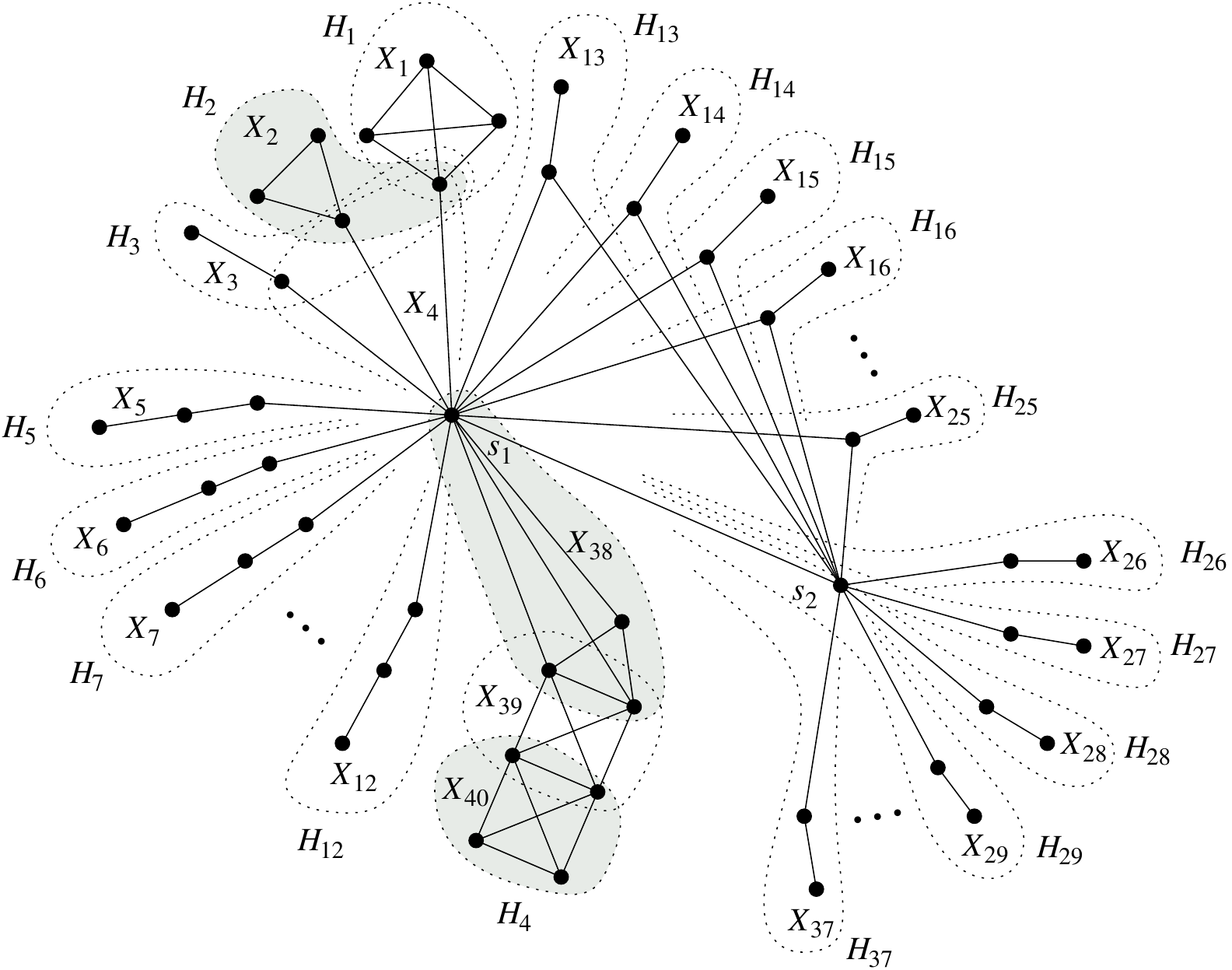}
\end{center}
\caption{\label{fig:bottleneck}A simple graph $G$ with a path decomposition $\cP$, $\width(\cP)=3$.}
\end{figure}

%% file: clean.tex
We show in the previous section that green $S$-branches appear only in the bottleneck interval $I(S)$ of $S\in \cS$. However, red, blue and purple $S$-branches may also appear in $I(S)$. (We say that an $S$-branch $H$ appears in $I(S)$ if there is $t\in I(S)$ such that $X_t\cap V(H)\neq \emptyset$.) Our goal in this section is to show that the search for minimum-length path decompositions of width $k\leq 3$ can be limited to a class of well-arranged path decompositions with no read and blue $S$-branches for any $S\in \cS$.
Moreover, any path decomposition  in this class suspends all purple branches in $I(S)$ so that only legacy vertices of purple $S$-branches that appear already in $\{1,..., t_1(S)-1\}$ may appear in $I(S)$ for any $S\in \cS$. We now formally define the well-arranged path decompositions.

Let $\cP=(X_1,\ldots,X_l)$ be a path decomposition of $G$.
We say that a subgraph $H$ of $G$ \emph{waits} in step $i$, $1\leq i\leq l$, if $X_i\cap V(H)\subseteq X_{i-1}\cap V(H)$. (In the latter statement we take $X_0=\emptyset$.) If $H$ does not wait in step $i$, then we say that $H$ \emph{makes progress} in step $i$.
For an interval $I\subseteq \{1,\hdots, l\}$, we say that a subgraph $H$ of $G$ \emph{waits in $I$} if $H$ waits in all steps $i\in I$,
and we say that $H$ \emph{makes progress in $I$} otherwise.
For $S\in \cS$, we say that a path decomposition is \textit{well-arranged with respect to $S$} if all components in $G-S$, except possibly $S$-leaves and green $S$-branches, wait in the interval $I(S)$. (Recall that not every component of $G-S$ is necessarily an $S$-component.) A path decomposition is called \textit{well-arranged} if it is well-arranged with respect to every bottleneck set. In this section, we will show that for every graph $G$ it is true that if $G$ has a path decomposition of length $l$ and width at most $k$, then $G$ has a well-arranged path decomposition of length $l$ and width at most $k$. 

To show the existence of well-arranged path decompositions, we will choose our path decomposition to be minimal in a certain sense. To make this precise, let us first order, for a given path decomposition $\cP$ of $G$, the bottleneck sets of $G$ as
$S_1, \hdots, S_p$ such that for all $i, j\in\{1,\ldots,p\}$ with $i\leq j$:
\begin{equation} \label{eq:S_order}
t_1(S_i) \leq t_1(S_j),\textup{ and if }t_1(S_i) = t_1(S_j),\textup{ then }t_2(S_i) \leq t_2(S_j).
\end{equation}
That is, we order the bottleneck sets by starting time $t_1(S)$ and, in case of a tie, by ending time $t_2(S)$.
We associate with $\cP$ the vector $\mbalpha(\cP) = (|I(S_1)|, \hdots, |I(S_p)|)$. 

\begin{definition} \label{def:clean}
A path decomposition $\cP$ of $G$ is called \emph{clean} if, among all path decompositions of width at most $\width(\cP)$ and length at most $\length(\cP)$ of $G$, the following holds:
\begin{enumerate}[label={\textnormal{\bfseries (C\arabic*)}}, leftmargin=*]
\item\label{it:clean:min-size} $\size(\cP)$ is minimum;
\item\label{it:clean:min-lex} subject to \ref{it:clean:min-size}, the vector $\mbalpha(\cP)$ is lexicographically smallest.
\end{enumerate}
\end{definition}
We explicitly note the following:
\begin{observation} \label{ob:clean}
If a graph has a path decomposition of width at most $k$ and length at most $l$, then it has a clean path decomposition of width at most $k$ and length at most $l$.
\end{observation}

\subsubsection{Basic characteristics of clean path decompositions}

We now prove some characteristics of clean path decompositions. The proofs will require the following lemma that strengthens  \ref{obs:O2} for connected $H$.

\begin{lemma} \label{lem:continuity}
Let $G$ be a graph, let $\cP$ be a path decomposition of $G$, and let $H$ be any connected subgraph of $G$.
Then, for each $t\in \{\alpha(H),\ldots,\beta(H)-1\}$, we have
$|X_t\cap X_{t+1}\cap V(H)| \geq 1$. 
\end{lemma}
\begin{proof}
Let $t\in\{\alpha(H),\ldots,\beta(H)-1\}$ be given. Define
$A = (X_{\alpha(H)}\cup \hdots \cup X_t)\cap V(H)$
and
$B = (X_{t+1}\cup \hdots \cup X_{\beta(H)})\cap V(H)$.
Since $X_{\alpha(H)}\cap V(H)$ and $X_{\beta(H)}\cap V(H)$ are non-empty,
it follows that $A$ and $B$ are non-empty. Moreover, $A\cup B = V(H)$ by \eqnref{eqn:VHalphabeta}.
If $A\cap B = \emptyset$, then $(A, B)$ is a partition of $V(H)$; 
by \ref{pathaxiom2} it then follows that there are no edges between $A$ and $B$, contrary
to the fact that $H$ is connected. Thus, there exists $v\in A\cap B$. 
By \ref{pathaxiom3'}, $v\in X_t\cap X_{t+1}$
and hence $v\in (X_t\cap X_{t+1})\cap V(H)$.
\end{proof}

We begin by showing that any $S$-branch, $S\in\cS$,  that starts in $I(S)$ does so with two vertices at a time; similarly, any $S$-branch that completes in $I(S)$ does so with two vertices at a time.

\begin{lemma} \label{lem:start_and_end_bag}
Let $G$ be a connected graph, let $\cP$ be a clean path decomposition, $S\in \cS$, 
and let $H$ be an $S$-branch. The following statements hold:
\begin{enumerate}[label={\normalfont (\roman*)}]
\item If $\alpha(H) \in I(S)$, then $|V(H)\cap X_{\alpha(H)}| \geq 2$.
\item If $\beta(H) \in I(S)$, then $|V(H)\cap X_{\beta(H)}| \geq 2$.
\end{enumerate}
\end{lemma}
\begin{proof}
(i) By the definition of $\alpha(H)$, we have $|V(H)\cap X_{\alpha(H)}| \geq 1$.
Suppose for a contradiction that $V(H)\cap X_{\alpha(H)} = \{v\}$.
It follows from \ref{pathaxiom2} that, for every $u\in N_H(v)$, there exists $\tau(u)\in\{\alpha(H),\hdots, \beta(H)\}$ 
such that $\{u,v\}\subseteq X_{\tau(u)}$.  By the assumption that $X_{\alpha(H)}\cap V(H)=\{v\}$, it follows that $\tau(u)\neq \alpha(H)$.
Moreover, by  \ref{lem:continuity}, we have $|X_{\alpha(H)+1}\cap X_{\alpha(H)}\cap V(H)|\geq 1$, which implies that $v\in X_{\alpha(H)+1}$.
Finally, because $\alpha(H)\in I(S)$, we have $S\subseteq X_{\alpha(H)+1}$.
Now construct a new path decomposition $\cP'$ from 
$\cP$ by replacing bag $X_{\alpha(H)}$ by $X_{\alpha(H)}\setminus \{v\}$.
Then, \ref{pathaxiom1} still holds for $\cP'$ because $v\in X_{\alpha(H)+1}$. To check \ref{pathaxiom2}, 
notice that the only edges that might violate this condition are the ones of the form $\{u,v\}$ where 
$u$ is a neighbor of $v$. Any neighbor $u$ of $v$ is either in $S$ or in $N_H(v)$. For $u\in S$, we have $\{u,v\}\subseteq X_{\alpha(H)+1}$;
for $u\in N_H(v)$, we have $\{u,v\}\subseteq X_{\tau(u)}$. Thus, \ref{pathaxiom2} holds for $\cP'$.
Finally, $v$ is the only vertex that might violate \ref{pathaxiom3'}. However, since $v\not\in X_j$ 
for $j < \alpha(H)$, condition \ref{pathaxiom3'} holds for $v$.
Thus, $\cP'$ is a path decomposition with $\size(\cP') < \size(\cP)$, contrary to the fact that $\cP$ is clean.
This proves (i). Part (ii) follows from the symmetry, i.e., by applying (i) to the reverse of $\cP$.
\end{proof}

For any $S\in \cS$, by \ref{lem:decomposition_structure1}, $I(S)$ is well-defined and hence the definition of $I(S)$ implies that there exists some $S$-branch that starts at $t_1(S)$, and similarly there exists some $S$-branch that completes at $t_2(S)$. Together with \ref{lem:start_and_end_bag} and the fact that $k\leq 3$, this gives the following useful corollary:

\begin{corollary}\label{lem:ISendpoints}
Let $G$ be a connected graph, let $\cP$ be a clean path decomposition of width $k\leq 3$ of $G$, and $S\in \cS$.
Then, the following statements hold:
\begin{enumerate}[label={\normalfont(\roman*)},align=left]
\item there is exactly one green $S$-branch $H$ such that $|X_{t_1(S)}\cap V(H)| \geq 2$.
\item there is exactly one green $S$-branch $H$ such that $|X_{t_2(S)}\cap V(H)| \geq 2$.
\end{enumerate}
\end{corollary}

We say that $H$ in (i) \textit{determines $t_1(S)$}, and similarly, $H$ in (ii) \textit{determines $t_2(S)$}.

Moreover, the following lemma shows that if an $S$-branch $H$, $S\in\cS$, makes progress in step $t\in I(S)$ of a clean path decomposition, then $X_t$ must contain at least two vertices of $V(H)$. 
\begin{lemma} \label{lem:more_than_one}
Let $G$ be a connected graph, let $\cP$ be a clean path decomposition, $S\in \cS$, 
and let $H$ be an $S$-branch. 
Then, for any $t\in I(S)$, if $H$ makes progress in step $t$, then $|X_{t}\cap V(H)|\geq 2$.
\end{lemma}
\begin{proof}
Clearly, since $H$ makes progress in step $t\in I(S)$, we have $t\in \{\alpha(H), \ldots, \beta(H)\}\cap I(S)$.
If $t = \alpha(H)$, then the result follows from \ref{lem:start_and_end_bag}.
So we may assume that $t\in \{\alpha(H)+1,\ldots,\beta(H)\}\cap I(S)$.
Then, there exists a vertex $u\in (X_t\setminus X_{t-1})\cap V(H)$.
However, by  \ref{lem:continuity}, we have $|X_t\cap X_{t-1}\cap V(H)|\geq 1$.
Thus, there is a vertex $v \in X_t\cap X_{t-1}\cap V(H)$ and $v\neq u$.
Therefore, $\{u,v\} \subseteq X_t\cap V(H)$, as required.
\end{proof}

Finally, \ref{lem:ISendpoints} and $k\leq 3$, give the following upper bounds on the numbers of red, blue and purple $S$-branches for $S\in \cS$.
\begin{lemma}\label{lem:onlyone}
Let $G$ be a connected graph, let $\cP$ be a clean path decomposition of width $k\leq 3$ of $G$, and $S\in \cS$.
Then, the following statements hold:
\begin{enumerate}[label={\normalfont(\roman*)},align=left]
\item there is at most one red or purple $S$-branch;
\item there is at most one blue or purple $S$-branch;
\item if there is a red, purple, or blue $S$-branch, then $k=3$ and $|S|=1$.
\end{enumerate}
\end{lemma}
\begin{proof}
By \ref{lem:ISendpoints}, there exist green $S$-branches $H_1$ and $H_2$ (possibly $H_1=H_2$) such that $\{u_i, v_i\}\subseteq X_{t_i(S)}\cap V(H_i)$ for each $i=1,2$.
Suppose that there exists an $S$-branch that is red of purple and let $H_1',\ldots,H_r'$ be all $S$-branches that are red or purple.
By \ref{lem:continuity}, for each $i=1,\ldots,r$ there exists $u_i'\in V(H_i')\cap X_{t_1(S)}$.
Hence, $S\cup \{u_1, v_1\}\cup\{u_1',\ldots,u_r'\} \subseteq X_{t_1(S)}$.
By the definition of the coloring, $u_i'\notin\{u_1,v_1\}$ for each $i\in\{1,\ldots,r\}$.
Thus, $r\leq 1$, which proves (i).

For (ii), suppose that there exist $q\geq 1$ $S$-branches that are blue or purple.
Denote those $S$-branches by $H_1'',\ldots,H_q''$.
By \ref{lem:continuity}, for each $i=1,\ldots,q$ there exists $u_i''\in V(H_i'')\cap X_{t_2(S)}$.
Hence, $S\cup \{u_2, v_2\}\cup\{u_1'',\ldots,u_r''\} \subseteq X_{t_2(S)}$.
By the definition of the coloring, $u_i''\notin\{u_2,v_2\}$ for each $i\in\{1,\ldots,q\}$.
Thus, $q\leq 1$ and (ii) follows.

Finally, if $r=1$ or $q=1$, then $S\cup\{u_1,v_1,u_1'\}\subseteq X_{t_1(S)}$ or $S\cup\{u_2,v_2,u_1''\}\subseteq X_{t_2(S)}$, respectively.
Because $|X_{t_i(S)}|\leq k+1 \leq 4$, $i=1,2$, this implies that $k=3$ and $|S| = 1$, as required in (iii).
\end{proof}

\subsubsection{The absence of red and blue $S$-branches}

We now show that clean path decompositions lack red and blue components.

We start with the following lemma, which provides a convenient way of re-arranging a path decomposition.
Notice that this lemma applies to all components of $G-S$ (i.e., $S$-branches, but also $S$-leaves
and other components of $G-S$ that are not $S$-components). Note also that we can always prepend or append a path decomposition with empty sets,
so that the condition that the path decompositions be of length exactly $|I(S)|$ is less restrictive than 
one may think at first sight.

\begin{lemma} \label{lem:movebranches}
Let $G$ be a connected graph, let $\cP= (X_1, \hdots, X_l)$ be a path decomposition of width $k\leq 3$ of $G$, and $S\in \cS$.
Write $t_1 = t_1(S)$ and $t_2 = t_2(S)$.
For each component $H$ of $G-S$, let $\cQ^H = (Y_{t_1}^H, \hdots, Y_{t_2}^H)$ be a path decomposition of length $|I(S)|$ 
for the graph
\[
\cY^H = H\left[\bigcup_{t\in I(S)} X_t\cap V(H) \right],
\]
satisfying the following conditions:
\begin{enumerate}[label={\normalfont(\roman*)},leftmargin=*, align=left]
\item if $\alpha_\cP(H) < t_1$, then $X_{t_1}\cap V(H) \subseteq Y_{t_1}^H$; and
\item if $\beta_\cP(H) > t_2$, then $X_{t_2} \cap V(H) \subseteq Y_{t_2}^H$.
\end{enumerate}
Define $\cP' = \left(X_1, \hdots, X_{t_1-1}, 
\quad S\cup \bigcup_{H} Y_{t_1}^H, \hdots, 
S\cup \bigcup_{H} Y_{t_2}^H, \quad X_{t_2+1}, \hdots,  X_l\right)$,
where the unions are taken over all components $H$ of $G-S$. Then $\cP'$ is a path decomposition 
of $G$.
\end{lemma}
\begin{proof}
For convenience, define
$\cP'_1 = \left(X_1, \hdots, X_{t_1-1}\right)$, $\cP'_2 = \left(S\cup \bigcup_{H} Y_{t_1}^H, \hdots, S\cup \bigcup_{H} Y_{t_2}^H\right)$,
and $\cP'_3 = \left(X_{t_2+1}, \hdots,  X_l\right)$.
Recall from \ref{lem:decomposition_structure1} that $I_{\cP}(S) \neq \emptyset$, and hence $\cP'_2$ is not
an empty list. Moreover, define $Z = G - \bigcup_{t\in I_{\cP}(S)} X_t$,
i.e. $Z$ is the subgraph of $G$ induced by all vertices not appearing in any bag $X_t$ with $t\in I_{\cP}(S)$. 
Notice that $V(G) = V(Z) \cup S \cup \bigcup_H V(\cY^H)$.
Write $\cP' = (X'_1, \hdots, X'_l)$.

To check \ref{pathaxiom2}, let $\{u,v\}\in E(G)$.  There are a few possibilities.
If $v\in V(Z)$, then, by \ref{pathaxiom2} for $\cP$, there exists some $t$ such that $\{u,v\}\in X_t$;
since $v\notin X_s$ for all $s\in I_{\cP}(S)$, it follows that $t\not\in I_{\cP}(S)$, and hence $X'_t = X_t$, so that $\{u,v\}\in X'_t$.
If $u,v \in S$, then $\{u,v\}\subseteq X'_{t_1}$.
If $u,v \in V(\cY^H)$ for some component $H$ of $G-S$, then \ref{pathaxiom2} for $\cQ^H$ implies that $\{u,v\}\subseteq Y_{t}^H$ for some $t\in I_{\cP}(S)$, and hence $\{u,v\}\subseteq X'_t$.
Finally, if $u\in S, v\in V(\cY^H)$ for some component $H$ of $G-S$, then by \ref{pathaxiom1} for $\cQ^H$, $v\in Y_t^H$ for some $t\in I_{\cP}(S)$, and hence $\{u,v\}\subseteq X'_t$.
This establishes property \ref{pathaxiom2}. 

Note that \ref{pathaxiom1} follows from \ref{pathaxiom2} because $G$ is connected.

Finally, in order to verify \ref{pathaxiom3'}, 
let $v\in V(G)$. If $v\in S\cup V(Z)$, then notice that $X'_t\cap (S\cup V(Z)) = X_t\cap (S\cup V(Z))$ for all 
$t\in \{1, \hdots, l\}$, and hence property \ref{pathaxiom3'} for $v$ follows immediately from 
the fact that \ref{pathaxiom3'} holds for $\cP$.
So we may assume that  $v\in V(\cY^H)$ for some component $H$ of $G-S$. Let $1\leq i< i'\leq l$ be such that 
$v\in X'_i$ and $v\in X'_{i'}$. We need to show that $v\in X'_j$ for all $j\in\{i,\hdots, i'\}$. 
Notice that, because $v\in V(\cY^H)$, there exists $t^*\in I_{\cP}(S)$ such that $v\in X_{t^*}$.
Let us go through the cases:
\begin{enumerate}[label={\normalfont(\arabic*)}, leftmargin=*, align=left]
\item $i\leq i'<t_1$ or $t_2<i\leq i'$. Then, $X'_j = X_j$ for all $j\in \{i, \hdots, i'\}$, and hence
the fact that $v\in X'_j$ follows directly from \ref{pathaxiom3'} for $\cP$. 
\item $t_1\leq i\leq i'\leq t_2$. Then, $v\in Y_i^H$ and $v\in Y_{i'}^H$ and hence 
the fact that $v\in Y_j^H\subseteq X'_j$, $j=i,\ldots,i'$, follows directly from \ref{pathaxiom3'} for $\cQ^H$. 
\item $i < t_1 \leq i' \leq t_2$.  Then, $v\in X_i$. Hence, by \ref{pathaxiom3'} applied to $i, t_1$ and $\cP$, 
it follows that $v\in X_{j}$ for all $j\in \{i, \hdots, t_1\}$. 
In particular, also $v\in X_{t_1}$. 
Since $v\in X_i$, we have $\alpha(H) < t_1$. Thus, by condition (i) in the statement of the lemma, it follows that 
$v\in Y_{t_1}^H \subseteq X'_{t_1}$.
Now it follows from (2) applied to $t_1, i'$ that $v\in X'_j$ for all $j \in \{t_1, \hdots, i'\}$.
\item $t_1 \leq i \leq t_2 < i'$. It follows from (3) and the symmetry (through $\cP'= (X_l, \hdots, X_1)$) that
$v\in X'_j$ for all $j\in \{i, \hdots, i'\}$.
\item $i< t_1\leq t_2 < i'$. It follows from (3) applied to $i, t_2$ that 
$v\in X'_j$ for all $j \in \{i, \hdots, t_2\}$,
and it follows from (4) applied to $t_2, i'$ that 
$v\in X'_j$ for all $j \in \{t_2, \hdots, i'\}$, as required.
\end{enumerate}
This proves the lemma.
\end{proof}

We make the following observation about shortening path decompositions.

\begin{observation} \label{obs:redundantbags}
Let $G$ be a connected graph and let $\cP= (X_1, \hdots, X_l)$ be a path decomposition of $G$.
Denote $X_0=X_{l+1}=\emptyset$.
Then, the deletion of any bag $X_t$, $t\in\{1,\ldots,l\}$, satisfying $X_t\subseteq X_{t-1}\cup X_{t+1}$ results in a path decomposition of width at most $\width(\cP)$ of $G$.
\end{observation}

We now prove the main result of this subsection. This result is the first step towards proving that clean path decompositions are well-arranged.

\begin{lemma} \label{lem:key1}
Let $G$ be a connected graph and $S\in \cS$.  Let $\cP=(X_1,\hdots,X_l)$ be a clean path decomposition of width $k\leq 3$ of $G$.
Then, there are no blue and no red $S$-branches.
\end{lemma}
\begin{proof}
First note that by \ref{lem:onlyone}(iii), if $k\leq 2$, then there are no red and no blue $S$-branches, and hence there is nothing
to prove. So we may assume from now on that $k=3$. 
Write $t_1 = t_1(S)$ and $t_2 = t_2(S)$.
Suppose for a contradiction that there is a red or a blue $S$-branch $H$. 
We will construct a new path decomposition $\cP'$ of width $k\leq 3$ and length $l$ for $G$,
which satisfies $\size(\cP')<\size(\cP)$, contrary to the fact that $\cP$ is clean.

It follows from \ref{lem:onlyone}(iii) that there are no purple $S$-branches and $|S| = 1$.
It also follows from \ref{lem:onlyone} that there may exist at most one blue $S$-branch and at most one red $S$ branch. Let $H$ be the red $S$-branch, if any exists; otherwise let $H$ be the empty graph. Similarly, let $H'$ be the blue $S$-branch, if any exists; otherwise let $H'$ be the empty graph. Notice that, by assumption, $H$ and $H'$ are not both empty graphs.

Since $G$ is connected, the fact that $|S|=1$ implies that every component of $G-S$ is an $S$-component.
Let $\cR \subseteq I_{\cP}(S)$ be the set of steps $t$ such that $|X_t\cap V(H)|\geq 2$, let $\cB \subseteq I_{\cP}(S)$ be the set of steps $t$ such that $|X_t\cap V(H')|\geq 2$, and let $\cZ = I_{\cP}(S)\setminus (\cR \cup \cB)$.  Let $r_1 <r_2 < \cdots <r_p$ be the steps in $\cR$, 
let $b_1 < b_2 < \cdots < b_q$ be the steps in $\cB$, and let  $z_1 < z_2 < \cdots < z_s$  be the steps in $\cZ$. We show that $\cR \cup \cB= \emptyset$ first. 
Suppose for a contradiction that $\cR \cup \cB\neq \emptyset$. \ref{lem:ISendpoints} and $k=3$ imply that $t_1, t_2\in \cZ$, so that $z_1 = t_1$ and $z_s = t_2$, and $\cR \cap \cB=\emptyset$.
\begin{figure}[htb]
\begin{center}
\includegraphics{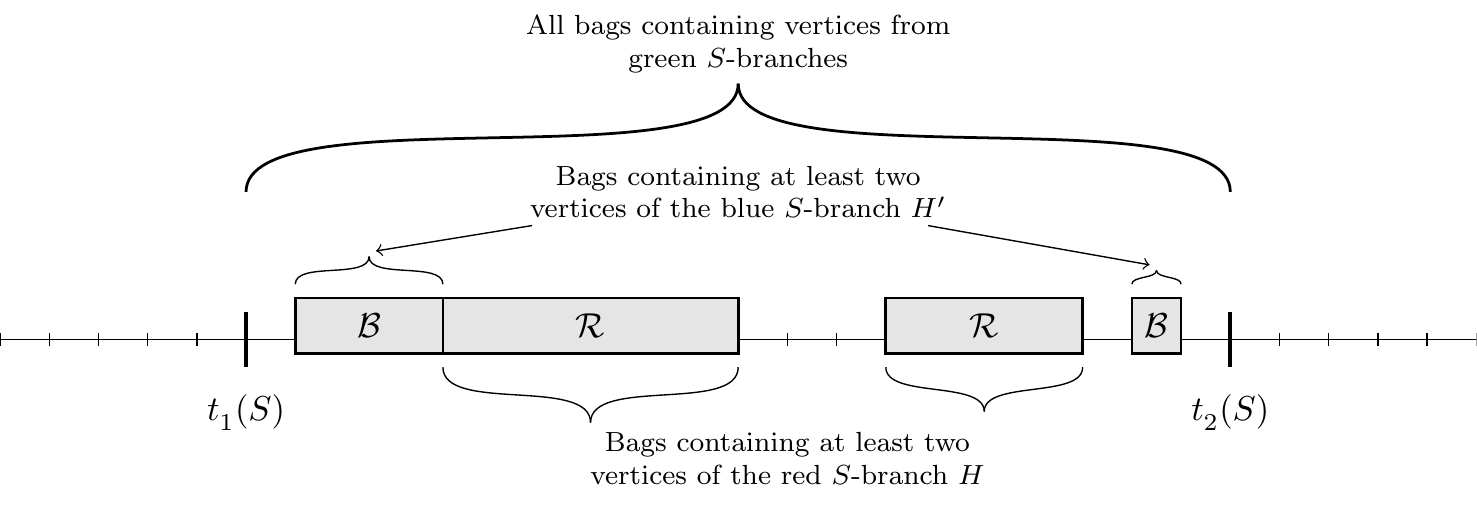}
\end{center}
\caption{$\cP$ is not well-arranged with respect to $S$.}
\end{figure}

For $t\in I_{\cP}(S)$, write $X_t = S_t\cup R_t \cup B_t \cup Z_t$,
where $S_t$ contains the vertex of $S$ and all $S$-leaves in $X_t$,
$R_t = X_t\cap V(H)$, $B_t = X_t\cap V(H')$, and $Z_t = X_t\setminus (S_t\cup V(H)\cup V(H'))$.
Thus, $R_t$ consists of the vertices of the red $S$-branch $H$, if any, in $X_t$, 
$B_t$ contains the vertices of the blue $S$-branch $H'$, if any, in $X_t$, 
$Z_t$ contains the vertices in $X_t$ that belong to green $S$-branches in $X_t$. 
Note that $S_t$, $R_t$, $Z_t$ and $B_t$ are pairwise disjoint.

Construct $\cP'$ from $\cP$ by replacing the bags 
\[ X_{t_1}, \quad \hdots, \quad X_{t_2},\]
of $\cP$ by the bags
\begin{equation} \label{eq:S'}
\begin{array}{l}
S_{r_1}\cup R_{r_1}, \quad\hdots, \quad S_{r_p}\cup R_{r_p},  \\
S_{z_1}\cup Z_{z_1}, \quad\hdots, \quad S_{z_r}\cup Z_{z_s}, \\
S_{b_1}\cup B_{b_1}, \quad\hdots, \quad S_{b_q}\cup B_{b_q}.
\end{array}
\end{equation}
Let $\cP' = (X_1', \hdots, X_l')$ be the result of the replacement. 
Let $\pi\colon I_{\cP}(S) \to I_{\cP}(S)$ be the permutation function that maps
$(t_1, \hdots, t_2)$ to $(r_1, \hdots,r_p, z_1, \hdots, z_s, b_1, \hdots, b_q)$.
We will show that $\cP'$ is a path decomposition of width at most $3$ of $G$ 
that satisfies $\size(\cP')<\size(\cP)$.

We first claim that $|X'_t|\leq 4$ for all $t\in\{1,\ldots,l\}$. To see this, note that 
$X'_t = X_t$ for all $t\not\in I_{\cP}(S)$, and hence $|X'_t|\leq 4$ 
trivially holds. Also, $X'_{t}\subseteq X_{\pi(t)}$ for all $t\in I_{\cP}(S)$, thus again $|X'_t|\leq 4$.

Next, we claim that $\size(\cP')<\size(\cP)$.
It suffices to show that  $|X'_{\pi(t_1)}| < |X_{t_1}|$ or $|X'_{\pi(t_2)}| < |X_{t_2}|$. 
By assumption, at least one of $\cB$, $\cR$ is non-empty.
If $\cB\neq\emptyset$, then since $H'$ is a blue $S$-branch there exists $w\in X_{t_2}\cap V(H')$. By \ref{lem:ISendpoints}(ii),  $|X_{t_2}\cap V(H')|=1$, thus, $\pi(t_2)\in \cR\cup \cZ$. Therefore, $X'_{\pi(t_2)}\subseteq X_{t_2}$, and $w\notin X'_{\pi(t_2)}$ by \eqnref{eq:S'} which proves $|X'_{\pi(t_2)}| < |X_{t_2}|$. Similarly, if $\cR\neq\emptyset$, then since $H$ is a red $S$-branch there is $u \in X_{t_1}\cap V(H)$. By \ref{lem:ISendpoints}(i),  $|X_{t_1}\cap V(H)|=1$, thus, $\pi(t_1)\in \cZ\cup\cB$. Therefore, $X'_{\pi(t_1)}\subseteq X_{t_1}$ and, by \eqnref{eq:S'}, $u\notin X'_{\pi(t_1)}$  which proves $|X'_{\pi(t_1)}| < |X_{t_1}|$. Thus, indeed $\size(\cP')<\size(\cP)$.

We now argue that $\cP'$ is a path decomposition of $G$. For every $S$-component $J$ such that $V(J)\cap X_t\neq\emptyset$ for some $t\in I_{\cP}(S)$ define
\[
\cY^{J} = J\left[\bigcup_{t\in I_{\cP}(S)} X_t\cap V(J)\right],
\]
i.e., $\cY^{J}$ is the subgraph of $J$ induced by all vertices that appear in bags
of $\cP$ in the interval $I_{\cP}(S)$. We have, by \ref{obs:redundantbags}:
\begin{itemize}[leftmargin=*]
\item $\cQ^H=(R_{r_1}, \hdots, R_{r_p}, \emptyset, \hdots, \emptyset)$ is a path decomposition of length $|I_{\cP}(S)|$ of $\cY^{H}$;
\item $\cQ^{H'}=(\emptyset, \hdots, \emptyset, B_{b_1}, \hdots, B_{b_q})$ is a path decomposition of length $|I_{\cP}(S)|$ of $\cY^{H'}$;
\item for any green $S$-branch $H^*$, $\cQ^{H^*}=(\emptyset, \hdots, \emptyset,
Z_{ z_1}\cap V(H^*), \hdots, Z_{z_s}\cap V(H^*), \emptyset, \hdots, \emptyset)$ is a path decomposition of length $|I_{\cP}(S)|$ of $\cY^{H^*}$;
\item for $J$ being an $S$-leaf $v$, $\alpha_{\cP}(v)\in I(S)$,
$\cQ^J=(S_{\pi(t_1)}\cap \{v\}, \hdots, S_{\pi(t_2)}\cap \{v\})$ is a path decomposition  of length $|I_{\cP}(S)|$ of $\cY^{J}$.
\end{itemize}

In order to apply \ref{lem:movebranches}, it remains to show that $X_{t_1}\cap V(H)\subseteq X'_{t_1}$ and $X_{t_2}\cap V(H')\subseteq X'_{t_2}$.
Notice that $X_{t_1}\cap V(H) = R_{t_1}$ and $X'_{t_1} = S_{r_1}\cup R_{r_1}$. Thus, it suffices to show that
$R_{t_1}\subseteq R_{r_1}$. This is trivial if there is no red $S$-branch. Otherwise, because $|R_{t_1}| = 1$ and $H$ waits in $\{t_1, \hdots, r_1-1\}$, it follows that $R_{t_1} = \cdots = R_{r_1-1}$. By \ref{lem:continuity}, $|R_{r_1-1}\cap R_{r_1}|\geq 1$ and hence $R_{t_1}\subseteq R_{r_1}$.
Similarly, it suffices to show that $B_{t_2}\subseteq B_{b_q}$. Again, this is trivial if there is no blue $S$-branch. Otherwise, because $|B_{b_q+1}| = 1$ and $H'$ waits in $\{b_q+1, \hdots, t_2\}$, it follows that $B_{b_q+1} = \cdots = B_{t_2}$. By \ref{lem:continuity}, $|B_{b_q}\cap B_{b_q+1}| \geq 1$ and hence $B_{t_2} \subseteq B_{b_q}$. 

Thus, by \ref{lem:movebranches}, $\cP'$ is a path decomposition of $G$, which proves that $\cR \cup \cB= \emptyset$. Therefore if $H$ is non-empty, then $X_{t_1}\cap V(H)=\cdots=X_{\beta(H)}\cap V(H)=\{v\}$. (Otherwise, $H$ makes progress in $I_{\cP}$, and thus \ref{lem:more_than_one} implies $\cR \neq \emptyset$, a contradiction.) Since
$\alpha(H)<t_1$, we have $X_{t_1-1}\cap X_{t_1}\cap V(H)=\{v\}$ by \ref{lem:continuity}. Moreover, by definition of $t_1$, $S\subseteq X_{t_1-1}$. Therefore, deleting $v$ from the bags $X_{t_1},\ldots,X_{\beta(H)}$ would result in a path decomposition $\cP'$ of $G$ such that $\size(\cP')<\size(\cP)$, since $t_1\leq \beta(H)$, contrary to the fact that $\cP$ is clean. Thus, no red $S$-branch exists. The proof that no blue branch exists follows by symmetry (take $(X_l,\hdots,X_1)$ instead of $\cP=(X_1,\hdots,X_l)$).
\end{proof}

\subsubsection{Clean path decompositions are well-arranged}

We now show that clean path decompositions are well-arranged, and thus, by \ref{ob:clean}, we may limit search for minimum-length path decomposition to well-arranged path decompositions. 

\def\Sstar{{S_m}}  

\begin{lemma} \label{lem:key2}
Let $G$ be a connected graph and let $\cP=(X_1,\ldots,X_{l})$ be a clean path decomposition of width $k\leq 3$ of $G$. Then, $\cP$ is well-arranged.
\end{lemma}
\begin{proof}
Let $S_1, \hdots, S_p$ be the bottleneck sets of $G$ ordered as in \eqnref{eq:S_order}. We will prove by induction on $m$ that $\cP$ is well-arranged with respect to the bottleneck sets $S_1, \hdots, S_m$,
for each $m=0,\hdots, p$. Thus, the case $m=p$ is the result of the lemma.

The base case $m=0$ is trivial. For the general case, let $1\leq m\leq p$ and assume inductively that $\cP$ is well-arranged with respect to each bottleneck set $S_1, \hdots, S_{m-1}$. Let $\mbalpha_m(\cP)$ be the vector with the first $m$ entries of $\mbalpha(\cP)$. 

Now suppose for a contradiction that $\cP$ is not well-arranged with respect to $\Sstar$. The latter implies that some component $H$ of $G-\Sstar$, which is neither an $\Sstar$-leaf nor a green $\Sstar$-branch, makes progress in some step in $I_{\cP}(\Sstar)$. It follows from \ref{lem:key1} that there no red and no blue $\Sstar$-branches. Thus, $H$ is either not an $\Sstar$-component or a purple $\Sstar$-branch. We deal with these two cases separately. 

\textsc{Case 1:} $H$ is not an $\Sstar$-component. If $|\Sstar|=1$, then every component of $G-\Sstar$ is an $\Sstar$-component, so it follows that $|\Sstar| \geq 2$. Therefore, by \ref{lem:onlyone}(iii), there are no purple $\Sstar$-branches. 
By \ref{lem:ISendpoints}, some green $\Sstar$-branches make progress in steps $t_1(\Sstar)$ and $t_2(\Sstar)$.
Thus, $k=3$, $|\Sstar|= 2$, and $X_{t_1(\Sstar)}$ and $X_{t_2(\Sstar)}$ contain no vertices of $V(H)$. Denote $\Sstar=\{s_1,s_2\}$.
Let $r_1=\alpha_{\cP}(H)$ and $r_2=\beta_{\cP}(H)$. Since $H$ makes progress in some step in $I_{\cP}(\Sstar)$, $t_1(\Sstar)<r_1\leq r_2 < t_2(\Sstar)$.
We may assume that $H$ is selected in such a way that $r_2$ is as large as possible.
Define
\[\cP'=\big(X_1,\ldots,X_{r_1-1},\quad X_{r_2+1}, \hdots, X_{t_2(\Sstar)},\quad X_{r_1}, \hdots, X_{r_2},\quad X_{t_2(\Sstar)+1},\ldots,X_l\big).\]
The proof that $\cP'$ is a path decomposition of $G$  follows from the following key observations.
First, only the vertices of $S_m$, green $S_m$-branches,  $S_m$-leaves, and components of $G-S_m$ that are not $S_m$-components can appear in $X_{t_1(S_m)}, \hdots X_{t_2(\Sstar)}$.
Second, no vertex of $H$ appears outside $X_{r_1}, \hdots, X_{r_2}$.
Third,
$|X_t\cap V(H)|\geq 1$ for each $t\in \{r_1,\hdots,r_2\}$, $|\Sstar|=2$ and the assumption that $r_2$ is as large as possible, imply that no $G-S_m$ component $H'\neq H$  with $|V(H)|\geq 2$ (this clearly includes $\Sstar$-branches) makes progress in steps $r_1,\hdots,r_2$.
Thus, no bag $X_t$ for $t\in \{r_1,\hdots,r_2\}$  contains a vertex of $G-S_m$ component $H'\neq H$ with $|V(H)|\geq 2$ for otherwise we could reduce the $\size(\cP)$ and thus contradict \ref{it:clean:min-size}. Therefore, each vertex in   $X_t\setminus (S_m\cup V(H))$,  $t\in \{r_1,\hdots,r_2\}$, belongs to a $G-S_m$ component $H'\neq H$ with $|V(H)|=1$ (this clearly includes $\Sstar$-leaves.) Therefore,
$X_{r_1-1}\cap X_{r_1}=S_m=X_{r_2}\cap X_{r_2+1}$. Finally, by \ref{lem:ISendpoints}, some green $\Sstar$-branch makes progress in step $t_2(\Sstar)$. Therefore, since $k=3$ and $|\Sstar|= 2$, we have $S_m=X_{t_2(S_m)}\cap X_{t_2(S_m)+1}$.

We claim that $\mbalpha_m(\cP') <_L \mbalpha_m(\cP)$, which  contradicts \ref{it:clean:min-lex}. To prove this claim, it suffices to show that 
\begin{equation} \label{eqn:lex2}
t'_1(S_i)=t_1(S_i) \textup{ and } t'_2(S_i)\leq t_2(S_i)
\end{equation}
for all $i = 1, \hdots, m$, with strict inequality $t'_2(S_i)< t_2(S_i)$ for $i=m$.

First note that $\alpha_{\cP'}(x) = \alpha_{\cP}(x)$ and $\beta_{\cP'}(x) = \beta_{\cP}(x)$ for each $x\in\Sstar$. Thus, $\maxalpha_{\cP}(S_m)=\maxalpha_{\cP'}(S_m)$ and $\minbeta_{\cP}(S_m)=\minbeta_{\cP'}(S_m)$. Let $H'$ and $H''$ be the green $S_m$-branches that determine
 $t_1(S_m)$ and $t_2(S_m)$ in $\cP$, respectively. We have $\alpha_{\cP'}(H')=\alpha_{\cP}(H')$
and $\beta_{\cP'}(H'')<\beta_{\cP}(H'')$. By definition, there is no $S_m$-branch $J$ such that $\maxalpha_{\cP}(S_m)<\alpha_{\cP}(J)<t_1(S_m)$ or $t_2(S_m)<\beta_{\cP}(J)<\minbeta_{\cP}(S_m)$. Also, we proved that there are no vertices of $S_m$-branches in $X_{r_1}, \hdots,  X_{r_2}$, thus the $S_m$-branches $H'$ and $H''$ determine $t'_1(S_m)$ and $t'_2(S_m)$ in $\cP'$, respectively. Therefore, $t_1'(\Sstar) = t_1(\Sstar)$ and $t_2'(\Sstar) < t_2(\Sstar)$, thus the inequality $t'_2(S_i)< t_2(S_i)$ for $i=m$ in \eqnref{eqn:lex2} is strict.

Now consider $S_i$, with $i \in \{1, \hdots, m-1\}$. Because $i < m$, we have that $t_1(S_i) \leq t_1(\Sstar)$. 
By the definition of $t_1(S_i)$, this implies that
\begin{equation} \label{eq:S_i_in_S_t}
S_i\subseteq X_t \textup{ for some } t<t_1(\Sstar).
\end{equation}
Let $u_1, u_2\in V(H')\cap X_{t_1(S_m)}$.
Since $k = 3$, we have $X_{t_1(S_m)} = \{s_1,s_2,u_1,u_2\}$. Let $H_i'$ and $H_i''$ be the green $S_i$-branches that determine
 $t_1(S_i)$ and $t_2(S_i)$ in $\cP$, respectively.
If $S_i$ contains a vertex that is not in $X_{t_1(S_m)}$, then, by \eqnref{eq:S_i_in_S_t}, $t_2(S_i) <\minbeta_{\cP}(S_i)< t_1(\Sstar)$.
Thus,  $H'_i$ and $H''_i$ determine $t'_1(S_i)$ and $t'_2(S_i)$ in $\cP'$.
Therefore, $t_1'(S_i) = t_1(S_i)$ and $t_2'(S_i) = t_2(S_i)$ and \eqnref{eqn:lex2} holds for the $i$.
So we may assume that $S_i \subseteq X_{t_1(S_m)}$.
If $u_j\in S_i$ for some $j\in\{1,2\}$, then $t_1(S_i) > \maxalpha_{\cP}(S_i) \geq t_1(S_m)$, a contradiction.
Thus,  $S_i \subseteq \{s_1, s_2\}$. By the symmetry and the fact that $S_i \neq S_m$, we may assume that $S_i = \{s_1\}$.

Since $\maxalpha_{\cP}(S_i)\leq \maxalpha_{\cP}(S_m)$ and $\minbeta_{\cP}(S_m)\leq \minbeta_{\cP}(S_i)$, we have $\maxalpha_{\cP}(S_i)=\maxalpha_{\cP'}(S_i)$ and $\minbeta_{\cP}(S_i)=\minbeta_{\cP'}(S_i)$.
Also, since $t_1(S_i) \leq t_1(\Sstar)$, $H_i'$ determines $t'_1(S_i)$ in $\cP'$.
Thus, $t_1'(S_i) = t_1(S_i)$.
By definition, there is no $S_i$-branch $J$ such that  $t_2(S_i)<\beta_{\cP}(J)<\minbeta_{\cP}(S_i)$.
Thus, both for $t_2(S_m)<\beta_{\cP}(H''_i)=t_2(S_i)$ and for $\beta_{\cP}(H''_i)=t_2(S_i)<r_1$, $H''_i$ determines $t'_2(S_i)$ in $\cP'$.
Thus, in both cases $t_2'(S_i) = t_2(S_i)$ and \eqnref{eqn:lex2} holds for the $i$.
Finally, suppose that $r_1\leq\beta_{\cP}(H''_i)\leq t_2(S_m)$.
We show that this case leads to a contradiction with the inductive assumption that $\cP$ is well-arranged with respect to $S_i$.
Consider the component $H^*$ in $G-S_i$ that contains $s_2$.
Since $|S_i|=1$, all $G-S_i$ components are $S_i$-components.
Moreover, $|V(H^*)|\geq 2$ because $H^*$ includes also $u_1$ and $u_2$, and hence $H^*$ is an $S_i$-branch.
We have, $\beta_{\cP}(H^*)\geq \minbeta_{\cP}(S_i)$, for otherwise $H^*$ is a green $S_i$-branch, and $\beta_{\cP}(H^*)\geq \beta_{\cP}(s_2)\geq \minbeta_{\cP}(S_m)>t_2(S_m)\geq\beta_{\cP}(H''_i)$. Therefore, since, by definition of $H''_i$, $\beta_{\cP}(J)\leq\beta_{\cP}(H''_i)$ for any green $S_i$-branch $J$, we get a contradiction. However, $\beta_{\cP}(H^*)\geq \minbeta_{\cP}(S_i)$ and \ref{lem:key1} imply that $H^*$ is a purple $S_i$-branch (observe that $\alpha_{\cP}(H^*)< t_2(S_i)$, since $\alpha_{\cP}(H^*)\leq \alpha_{\cP}(s_2)\leq \maxalpha_{\cP}(S_m)<t_1(S_m)<r_1\leq t_2(S_i)$). Moreover, $H^*$ makes progress in $t_1(S_m)$ for $S_m$-branch $H'$ makes progress in $t_1(S_m)$ and $H'$ is a subgraph of $H^*$. This contradicts the fact that $\cP$ is well-arranged with respect to $S_i$ for $t_1(S_m)\in I_{\cP}(S_i)$ (observe that
$t_1(S_i)\leq t_1(S_m)<r_1\leq t_2(S_i)$).

\bigskip

\textsc{Case 2:} $H$ is a purple $\Sstar$-branch. By \ref{lem:onlyone}, $|\Sstar| = 1$, and $k=3$.
Denote $\Sstar=\{s\}$.
Every component in $G-\Sstar$ is either an $\Sstar$-branch or an $\Sstar$-leaf. 
Because $H$ makes progress in $I_{\cP}(\Sstar)$, \ref{lem:more_than_one} implies that $|X_{t}\cap V(H)|\geq 2$ for some $t\in I_{\cP}(\Sstar)$. 
Let $r_2\in I_{\cP}(\Sstar)$ be latest step such that $|X_{r_2}\cap V(H)|\geq 2$ and let $r_1\leq r_2$ be the earliest step such that $|X_{t}\cap V(H)|\geq 2$ for all $t=r_1,\ldots,r_2$. For each $t\in I_{\cP}(\Sstar)$, write $X_t = Z_t \cup R_t\cup A_t$, where $Z_t$ is the set containing $\Sstar$ and all $\Sstar$-leaves contained in $X_t$, $R_t=X_t \cap V(H)$, and $A_t = X_t\setminus (R_t\cup Z_t)$.
Notice that $A_t$ contains precisely the vertices in $X_t$ that belong to green $\Sstar$-branches.
By \ref{lem:ISendpoints}, some green $\Sstar$-branches make progress in steps $t_1(\Sstar)$ and $t_2(\Sstar)$. Thus, \ref{obs:O2}, $k=3$, and definition of purple $S_m$-branch  imply
$|R_{t_1(\Sstar)}| = |R_{t_2(\Sstar)}| = 1$. 
Therefore $t_1(\Sstar)<r_1\leq r_2 < t_2(\Sstar)$. 

By \ref{obs:O2}, definition of purple $S_m$-branch, and the choice of $r_1$ and $r_2$, we have $|R_{r_1-1}| =1$, and $|R_t| = 1$ for all $t \in \{r_2+1, \hdots, t_2(\Sstar)\}$.
The latter also implies that $|Z_t\cup A_t| \leq 3$ for all $t \in \{r_2+1, \hdots, t_2(\Sstar)\}$.
Construct $\cP'$ from $\cP$ by replacing the bags
\[
X_{r_1}, \hdots, X_{t_2(\Sstar)}
\]
by
\begin{equation} \label{eq:no-purple}
Z_{r_2+1}\cup A_{r_2+1}\cup R_{r_1-1}, \quad \hdots, \quad Z_{t_2(\Sstar)}\cup A_{t_2(\Sstar)}\cup R_{r_1-1}, 
\quad Z_{r_1}\cup R_{r_1}, \quad \hdots, \quad Z_{r_2}\cup R_{r_2}.
\end{equation}
We obtain that $\width(\cP')\leq 3$, and by \ref{lem:movebranches}, $\cP'$ is a path decomposition of $G$.
If $A_t\neq\emptyset$ for some $t\in\{r_1,\ldots,r_2\}$, then, by \eqnref{eq:no-purple}, $\size(\cP')<\size(\cP)$, which contradicts \ref{it:clean:min-size}. So, $A_t=\emptyset$ for all $t\in\{r_1,\ldots,r_2\}$.

We claim again that $\mbalpha_m(\cP') <_L \mbalpha_m(\cP)$. As in Case 1, it suffices to show that \eqnref{eqn:lex2}
holds for all $i = 1, \hdots, m$, with strict inequality $t'_2(S_i)< t_2(S_i)$ for $i=m$.
First note that $\alpha_{\cP'}(s) = \alpha_{\cP}(s)$ and $\beta_{\cP'}(s) = \beta_{\cP}(s)$.
Thus, $\maxalpha_{\cP}(S_m)=\maxalpha_{\cP'}(S_m)$ and $\minbeta_{\cP}(S_m)=\minbeta_{\cP'}(S_m)$.
Let $H'$ and $H''$ be the green $S_m$-branches that determine $t_1(S_m)$ and $t_2(S_m)$ in $\cP$, respectively.
We have $\alpha_{\cP'}(H')=\alpha_{\cP}(H')$ and $\beta_{\cP'}(H'')<\beta_{\cP}(H'')$.
By definition, there is no $S_m$-branch $J$ such that  $\maxalpha_{\cP}(S)<\alpha_{\cP}(J)<t_1(S_m)$ or $t_2(S_m)<\beta_{\cP}(J)<\minbeta_{\cP}(S)$. Moreover, we proved that there are no vertices of green $S_m$-branches in $X_{r_1}, \hdots, X_{r_2}$, and by \ref{lem:key1} there are no red and blue $S_m$-branches, thus $S_m$-branches $H'$ and $H''$ determine
 $t'_1(S_m)$ and $t'_2(S_m)$ in $\cP'$. Therefore, $t_1'(\Sstar) = t_1(\Sstar)$ and $t_2'(\Sstar) < t_2(\Sstar)$, thus the inequality for $i=m$ in \eqnref{eqn:lex2} is strict.

Now consider $S_i$ with $i \in \{1, \hdots, m-1\}$. Because $i < m$, we have that $t_1(S_i) \leq t_1(\Sstar)$.
By definition of $t_1(S_i)$, this implies \eqnref{eq:S_i_in_S_t}.
Let  $u_1, u_2\in V(H')\cap X_{t_1(S_m)}$.
Let $w$ be the unique vertex of $H$ in $X_{t_1(S_m)}$.
Since $k = 3$, we have $X_{t_1(S_m)} = \{s,w,u_1,u_2\}$.
If $S_i$ contains a vertex that is not in $X_{t_1(S_m)}$, then, by \eqnref{eq:S_i_in_S_t}, $t_2(S_i) <\minbeta_{\cP}(S_i)< t_1(\Sstar)$.
Thus, $t_1'(S_i) = t_1(S_i)$ and $t_2'(S_i) = t_2(S_i)$ and \eqnref{eqn:lex2} holds.
So we may assume that $S_i \subseteq X_{t_1(S_m)}$.
If $u_j\in S_i$ for some $j\in\{1,2\}$, then $t_1(S_i) > \maxalpha_{\cP}(S_i) \geq t_1(S_m)$, a contradiction.
Thus, $S_i \subseteq \{s, w\}$.
Since $S_i \neq S_m$, this implies that either $S_i = \{s,w\}$ or $S_i = \{w\}$.

If $S_i = \{s,w\}$, then, by \ref{lem:ISendpoints}, $H'$ makes progress in $t_1(S_m)$.
However, $H'$ is a component in $G-S_i$ which is not an $S_i$-component.
Thus, if $t_1(S_m)\leq t_2(S_i)$ then we get a contradiction for $\cP$ is well-arranged with respect to $S_i$ (observe that $t_1(S_m)\in I_{\cP}(S_i))$.
If $t_1(S_m)> t_2(S_i)$, then $t_1'(S_i) = t_1(S_i)$ and $t_2'(S_i) = t_2(S_i)$ and \eqnref{eqn:lex2} holds.

Consider $S_i=\{w\}$. Let $H_i'$ and $H_i''$ be the green $S_i$-branches that determine $t_1(S_i)$ and $t_2(S_i)$ in $\cP$, respectively.
Since $\maxalpha_{\cP}(S_i)\leq t_1(S_m)$ and $t_2(S_m)\leq \minbeta_{\cP}(S_i)$, we have
$\maxalpha_{\cP}(S_i)=\maxalpha_{\cP'}(S_i)$ and $\minbeta_{\cP}(S_i)=\minbeta_{\cP'}(S_i)$.
Since $t_1(S_i) \leq t_1(\Sstar)$, $H_i'$ determines $t'_1(S_i)$ in $\cP'$. Thus, $t_1'(S_i) = t_1(S_i)$.
By definition, there is no $S_i$-branch $J$ such that  $t_2(S_i)<\beta_{\cP}(S_i)<\minbeta_{\cP}(S)$. Thus, both for
$t_2(S_m)<\beta_{\cP}(H''_i)=t_2(S_i)$ and for $\beta_{\cP}(H''_i)=t_2(S_i)<r_1$, $H''_i$ determines
 $t'_2(S_i)$ in $\cP'$. Thus, in both cases $t_2'(S_i) = t_2(S_i)$ and \eqnref{eqn:lex2} holds. 
Now suppose that $r_1\leq\beta_{\cP}(H''_i)\leq t_2(S_m)$. We show that this case leads to a contradiction with the inductive assumption that $\cP$ is well-arranged with respect to $S_i$.
The $S_i$-component $H^*$ that contains $s$ is an $S_i$-branch for it includes also $u_1$ and $u_2$.
Thus, $\beta_{\cP}(H^*)\geq \minbeta_{\cP}(S_i)$, for otherwise $H^*$ is a green $S_i$-branch and $\beta_{\cP}(H^*)\geq \beta_{\cP}(s)>t_2(S_m)\geq\beta_{\cP}(H''_i)$, which contradicts the choice of $H''_i$ as the branch that determines $t_1(S_i)$.
However $\beta_{\cP}(H^*)\geq \minbeta_{\cP}(S_i)$ and \ref{lem:key1} imply that $H^*$ is a purple $S_i$-branch (observe that $\alpha_{\cP}(H^*)< t_2(S_i)$, since $\alpha_{\cP}(H^*)\leq \alpha_{\cP}(s)= \maxalpha(S_m)<t_1(S_m)<r_1\leq t_2(S_i)$). Moreover, $H^*$ makes progress in $t_1(S_m)$ for $H'$ makes progress in $t_1(S_m)$ and $H'$ is a subgraph of $H^*$. This contradicts the fact that $\cP$ is well-arranged with respect to $S_i$ for $t_1(S_m)\in I_{\cP}(S_i)$.
\end{proof}

We close this section with the third key property of clean path decompositions.
A collection $\cB$ of sets is called \textit{strictly-nested} if for all $X, X'\in \cB$, ($X\neq X'$) and (either $X\cap X'=\emptyset$, or $X\subseteq X'$, or $X'\subseteq X$).
We have the following lemma.

\begin{lemma} \label{lem:IS-nested}
Let $G$ be a connected graph, let $\cP=(X_1,...,X_l)$ be a clean path decomposition of width $k\leq 3$ of $G$. Then, the collection $\{I(S)\colon S\in \cS\}$ is strictly-nested.
\end{lemma}
\begin{proof}
Let $S$ and $S'$ be two distinct bottleneck sets. We may assume without loss of generality that
$|S|\geq |S'|$. Write $I(S) = \{t_1, \hdots, t_2\}$ and $I(S') = \{t_1', \hdots, t_2'\}$.  
Suppose for a contradiction that ($I(S)=I(S')$) or ($I(S)\cap I(S')\neq\emptyset$, $I(S)\nsubseteq I(S')$ and $I(S')\nsubseteq I(S)$). Perhaps by taking $(X_l,...,X_1)$, we may assume that
\[
t_1 \leq t_1' \leq t_2 \leq t_2'.
\]
By definition of $t_1'$, there exists a green $S'$-branch $H'$ that determines $t_1'$. Fix $z\in X_{t'_1}\cap  V(H')$. 
By \ref{lem:key1} there are no red or blue $S$-branches, and by \ref{lem:key2} any purple $S$-branch and any component in $G-S$ that is not an $S$-component waits in $I(S)$, thus it follows that $z$ is a vertex of an $S$-component $H$ which is either an $S$-leaf or a green $S$-branch.

We first show that $H$ is a green $S$-branch. Let $x\in S\setminus S'$ (such $x$ exists because $|S|\geq |S'|$ and $S\neq S'$). If $x$ and $z$ are in the same
connected component (i.e., in $H'$) in $G-S'$, then $t'_1=\alpha(H') \leq \alpha(x) \leq \maxalpha(S)< t_1 \leq t_1'$, a contradiction. Thus, $x$ and $z$ are in different connected components of $G-S'$.
This implies that any path between $x$ and $z$ passes through a vertex in $S'$. 
Now suppose for a contradiction that $H$ is an $S$-leaf. Then $V(H)=\{z\}$ and the edge $(x,z)\in E(G)$ is a path between $x$ and $z$ that does not pass through a vertex in $S'$
since neither $x$ nor $s$ belong to $S'$. Therefore, $H$ is a green $S$-branch, thus includes some vertex $w\neq z$.

Now, since we just shown that $(x,z)\notin E(G)$, we can chose $w$ so that $(x,w)\in E(G)$. Therefore, $w\in S'\setminus S$ or at  least one vertex on any path between $w$ and $z$ in $H$ must be in $S'$. Thus, $S'\cap V(H) \neq\emptyset$. This, however, implies that $\beta(H)\geq\minbeta(S')> t_2' \geq t_2$, contrary to the fact that $\beta(H) \leq t_2$ (which follows from the fact that $H$ is
a green $S$-branch). This proves the lemma.
\end{proof}

%% file: connected_four.tex
In this section we turn the generic exponential-time algorithm from Section~\ref{sec:general_alg} into a polynomial-time algorithm that finds a minimum- length path decomposition of width at most $k$ of $G$ for given integer $k\leq 3$ and connected graph $G$. Recall that it can be checked in linear-time whether $\pw(G)\leq k$, see \cite{Bodlaender96,Bodlaender12}.
We formulate the algorithm in this section and prove its correctness in the next section.

Let $k\in\{1,2,3\}$ and let $G$ be a connected graph.
For every non-empty set $S\subseteq V(G)$, $|S|\leq k+1$, define
\begin{align} 
\cQ(S)=\bigl\{(\cf{G},f,l)\st & \cf{G}\subseteq\cC^2_G(S),|\cf{G}|\geq|\cC^2_G(S)|-12,\nonumber\\
& f\colon \cC^2_G(S) \to \{0,1\}, \ f(H)=f(H')\textup{ for all }H,H'\in\cf{G},\label{eq:Q(S)1}\\
& l\in\{0,\ldots,|\cC_G^1(S)|\}\bigr\}\nonumber
\end{align}
if $S\in\cS$ (i.e., if $S$ is a bottleneck set), and
\begin{equation} \label{eq:Q(S)2}
\cQ(S)=\{(\emptyset,f,l)\st \ f\colon \cC_G^2(S) \to \{0,1\}, \ l\in\{0,\ldots,|\cC_G^1(S)|\}\}
\end{equation}
if $S\notin\cS$.  

Every triple in $\cQ(S)$ provides information on which $S$-components have been covered (completed) by a partial path decomposition.
Here, $f(H)=1$ means that $H$ has been covered.

The first entry, $\cf{G}$, represents the collection of green $S$-branches. By \ref{lem:decomposition_structure1}, in any path decomposition, every bottleneck set $S$ has at most $12$ non-green $S$-branches, thus there are at least $|\cf{G}|\geq |\cC_G^2(S)|-12$ green $S$ branches. Moreover, by bottleneck definition, there must be at least one green $S$-branch. Thus, $\cf{G}$ is always non-empty for $S\in \cS$.  Since we only define the coloring of $S$-components for bottleneck sets $S$,  $\cf{G}$ is set to be empty for $S\notin \cS$.

The second entry, $f$, keeps track of which $S$-branches have been covered by a partial path decomposition.  Notice that we require that the green $S$-branches either all are covered or none of them is covered.

The third entry, $l$, counts the number of $S$-leaves that are covered by a partial path decomposition. 

Note that for $|S|>1$, there may exist connected components in $G-S$ that are not $S$-components, and thus $\cQ(S)$ does not provide any information about whether such components have been covered or not  by a partial path decomposition. This information, however, is not lost since it is provide by $\cQ(S')$ for some $S'\varsubsetneq S$.

Next, for $X\subseteq V(G)$, define $\cR(X)$ to be the set of all functions $R\colon 2^X\setminus\{\emptyset\}\to \bigcup_{S\subseteq X}\cQ(S)$ such that $R(S) \in \cQ(S)$ for each $S\subseteq X$, $S\neq\emptyset$.
Thus, each $R\in\cR(X)$ selects a triple from $\cQ(S)$ for each non-empty set $S\subseteq X$.
Now define the following edge-weighted directed graph $\cG_k$ with weights $w\colon E(\cG_k)\to\dZ$. 
The vertex set of $\cG_k$ is 
\[V(\cG_k)=\{(X,R)\st X\subseteq V(G), |X|\leq k+1,R\in\cR(X)\}\cup\{s,t\}.\]
In the following we denote $s=(\emptyset,\emptyset)$ and $\cC_G^1(S)=\{v_i(S)\st i=1,\ldots,|\cC_G^1(S)|\}$. Notice that if $S\neq S'$, then $\cC_G^1(S)$ and $\cC_G^1(S')$ are disjoint. 

Every vertex $v\in V(\cG_k)$ represents a set $\cV_v$ of the vertices of $G$ covered by any partial path decomposition that corresponds $v$.
We first show how to define $\cV_v$ for any given vertex $v = (X,R)\in V(\cG_k)$.
For each non-empty $S\subseteq X$, write $R(S)=(\cf{G}_S,f_S,l_S)$. Now, $\cV_v$ is given by
\begin{equation}\label{eq:Vvdef}
\cV_v =X\cup  \bigcup_{S\subseteq X,S\neq\emptyset} \left[\{v_1(S),\ldots,v_{l_S}(S)\} \quad \cup \quad \bigcup_{\mathclap{\substack{H\in \cC_G^2(S) \\ f_S(H)=1}}} V(H)\right].
\end{equation}

We are now ready to formally define the edge set of $\cG_k$ and the corresponding edge weights. The informal comments follow the definition.
Let there be an edge from $v$ to $t$ if and only if $\cV_v=V(G)$; the weight of such an edge is set to zero.
Then, let there be an edge from $v =(X,R)\in V(\cG_k)$ to $v' = (X',R')\in V(\cG_k)$ ($v\neq v'$) if and only if $\cV_v\subseteq \cV_{v'}$, $(X'\setminus \border_G(G[\cV_v]))\cap  \cV_v=\emptyset$  and one of the two following conditions holds:
\begin{enumerate}[label={\textnormal{\bfseries(G\arabic*)}},leftmargin=*,align=left]
 \item\label{def:G:step} $\cV_{v'}\setminus \cV_v\subseteq X'$ and $\border_G(G[\cV_v])= X'\cap \cV_v$. Set $w(v,v')=1$. We refer to the edge $(v,v')$ as a \textit{step edge}. 
 \item\label{def:G:jump} $\cV_{v'}\setminus \cV_v\nsubseteq X'$, and there exists a bottleneck set $S\subseteq X\cap X'$ with $R'(S)=(\mbG_S',f_S',l_S')$ such that:
\begin{enumerate}[label={\textnormal{\bfseries (G2\alph*)}}]
 \item\label{def:G:jump:a} $\cV_{v'}\setminus\cV_v=Y_S$, where
       $Y_S=\{v_{l_S+1}(S),\ldots,v_{l_S'}(S)\}\cup\bigcup_{H\in\mbG_S'}V(H)$, and $R(S)=(\mbG_S,f_S,l_S)$.
 \item\label{def:G:jump:b} $\border_G(G[\cV_v])=\bar S $, where $\bar S=X\cap X'$.
 \item\label{def:G:jump:c} $X'=\bar S$.
 \item\label{def:G:jump:d} there exists a path decomposition $\bar \cP=(\bar X_1,\hdots,\bar X_{\bar t})$ of width
\begin{equation}\label{eqn:subwidth}
\width(\bar\cP)\leq k-|\bar S|
\end{equation}
for the possibly disconnected graph $\bar G = G[Y_S]$. Set $w(v,v')=\bar t$, where $\bar t$ is taken as small as possible.
\end{enumerate}
We refer to the edge $(v,v')$ as a \textit{jump edge} and  to $S$ as the \textit{bottleneck set associated with $(v,v')$}. We refer to any path decomposition $(\bar X_1\cup \bar S,\hdots,\bar X_{\bar t}\cup \bar S)$ as a \textit{witness} of the jump edge.
\end{enumerate}

We naturally extend the weight function $w$ to the subsets of edges, i.e., $w(F)=\sum_{e\in F}w(e)$ for any $F\subseteq E(\cG_k)$.

An $s$-$t$ path in $\cG_k$  allows to construct a sequence of bags as follows.
We start with an empty sequence in $s$.
Then, whenever that path passes through a directed edge $(v,v')$ in $\cG_k$ we append the bag $X'$, if $(v,v')$ is a step edge, or a sequence of bags $(X'_1,\hdots,X'_{\bar t})$, where $X'_i = \bar S \cup \bar X_i$ for all $i=1, \hdots, \bar t$,  if $(v,v')$ is a jump edge, to the sequence.
We stop reaching $t$.

In the next section, we prove the following  theorem that is analogous to \ref{claim:generic_path_correspondence} for the generic exponential-time algorithm:

\begin{theorem} \label{thm:path_correspondence}
Let $k\in\{1,2,3\}$ and let $G$ be a connected graph.
There exists an $s$-$t$ path of weighted length $l$ in $\cG_k$ if and only if there exists a path decomposition of width at most $k$ and length $l$ of $G$.
\end{theorem}

Though we defer a formal proof of the theorem to the next section, we now provide some informal insights into the definition of $\cG_k$ and the above construction of a path decomposition for an $s$-$t$ path in $\cG_k$. The conditions \ref{def:G:step},\ref{def:G:jump:b}, and \ref{def:G:jump:c} guarantee that $X'$ includes the border of $G[\cV_{v}]$. This is intended to guarantee \ref{pathaxiom2} and \ref{pathaxiom3} in \ref{def:path_dec} for the partial path decompositions build for step edges and jump edges. For the latter the border is added to each bag of the path decomposition $\bar \cP$ and thus it is carried forward till the last bag of the partial path decomposition. Clearly this is necessary since, by \ref{def:G:jump:a} and \ref{def:G:jump:d}, $\bar \cP$ includes only the vertices of  green $S$-branches and $S$-leaves, however, both $S$ and possibly other vertices, for instance those of purple $S$-branches, belong both to the border of $G[\cV_{v}]$ and to the border of $G[\cV_{v'}]$. 
Finally, the condition $(X'\setminus \border_G(G[\cV_v]))\cap  \cV_v=\emptyset$ guarantees that vertices in $\cV_{v}$ without neighbors in $V(G)\setminus \cV_v$ are \emph{not} carried forward for it is unnecessary.

We are now ready to prove \ref{thm:connected}.

\textbf{Proof of \ref{thm:connected}:} It follows from \ref{thm:path_correspondence} that a solution to the problem $\PathLenFixed{k}$ can be obtained by constructing $\cG_k$ and finding a shortest $s$-$t$ path $P$ in $\cG_k$.

This approach leads to a polynomial-time algorithm, because $\cG_k$ can be constructed in time polynomial in $n$, where $n=|V(G)|$, given polynomial-time algorithms that find minimum-length path decompositions of widths $0,1,\ldots,k-1$ for disconnected graphs.
Indeed, by (\ref{eq:Q(S)1}) and  (\ref{eq:Q(S)2}),  $|\cQ(S)|=O(n^{13})$ for a given $S\subseteq V(G)$.
Thus, for a given $X\subseteq V(G)$ of size at most $k+1\leq 4$ the set $\cR(X)$ is of size $n^{O(2^k)}$.
Therefore, there are $n^{O(2^k)}$ nodes $(X,R)\in V(\cG_k)$ for a given set $X\subseteq V(G)$ of size at most $k$.
Since there are $O(n^4)$ possible sets $X$, we obtain that $|V(\cG_k)|$ is bounded by a polynomial in $n$.
Each jump edge of $\cG_k$, $k>0$, can be constructed in polynomial-time by using an algorithm that for any, possibly disconnected, graph  finds its path decomposition of width not exceeding $k-|\bar S|$, $|\bar S|>0$, whenever one exists. Finally, note that $\pw(G)=0$ implies that $G=K_1$.
This proves \ref{thm:connected}.\qed

Our algorithm can be significantly simplified for $k<3$.  If $k=1$, then $\cG_1$ contains no jump edges, and any path decomposition $\cP=(X_1,\ldots,X_l)$ for $G$ such that $X_i \neq X_{i+1}$,  $i=1,\ldots,l-1$, solves $\PathLenFixed{1}$.
If $k=2$, then each bottleneck is of size $1$.
However, since our focus is on settling the border between polynomial-time and NP-hard cases in the minimum length path decomposition problem, we do not attempt to optimize the running time of the polynomial-time algorithm.

\subsubsection{The proof of \ref{thm:path_correspondence}}
We start with the following observation.
\begin{observation}\label{obs:border_in_X}
For each $v=(X,R)\in V(\cG_k)$, it holds that $\border_G(G[\cV_v])\subseteq X$.
\end{observation}
\begin{proof}
Suppose for a contradiction that there exists a vertex $u\in \border_G(G[\cV_v])\setminus X$. By border definition $\border_G(G[\cV_v]) \subseteq \cV_v$, thus \eqnref{eq:Vvdef} implies that there exists a non-empty set $S\subseteq X$ such that $u$ is in some $S$-component $H$, and $V(H) \subseteq \cV_v$. Also by the definition, $u$ has a neighbor $u'\in V(G)\setminus \cV_v$ which, by $S\subseteq X \subseteq \cV_v$, implies that $u'\not\in S$. Therefore, since each neighbor of $u$ is either in $S$ or in $V(H)$, we get $u'\in V(H)$.
But, by~\eqnref{eq:Vvdef}, this implies that $u'\in \cV_v$, a contradiction.
\end{proof}

In the following lemma, we prove one direction of \ref{thm:path_correspondence}.
\begin{lemma} \label{lem:path_gives_decomposition}
Let $k\in\{1,2,3\}$ and let $G$ be a connected graph.
Let $P = v_0\sd v_1\sd v_2\sd \hdots\sd v_m$ be an $s$-$t$ path in $\cG_k$ and let $\ell$ be its weighted length.
Then, there exists a path decomposition $\cP$ of $G$ with $\width(\cP)\leq k$ and $\length(\cP)\leq\ell$.
\end{lemma}
\begin{proof}
Since $P$ is an $s$-$t$ path, we have $v_0 = s$ and $v_m = t$.
For brevity, we will write $(X_i,R_i) = v_i$, $\cV_i = \cV_{v_i}$, and $G_i=G[\cV_i]$ for each $i=0,\ldots,m-1$.
Also, for each $i=0, \hdots, m-1$, let $P_i = v_0$-$v_1$-$\hdots$-$v_i$ denote the prefix of $P$ formed by the first $i$ edges, 
and let $\ell_i$ be the weighted length of $P_i$. 

We prove, by  induction on $i=0,\ldots,m-1$, that there exists a path decomposition $\cP_i$ of $G_i$ that satisfies $\width(\cP_i)\leq k$, $\length(\cP_i) = \ell_i$, and $\border_G(G_i)$ is contained in the last bag of $\cP_i$. The claim is trivial for $i=0$. Indeed, because $\cV_{0} = \emptyset$, we may take $\cP_0$ to be a path decomposition of length $0$. For the generic case, suppose that the claim holds for some $i$, $0\leq i<m-1$. We now prove that it holds for $i+1$.
We have two cases, depending on the type of the edge $(v_i,v_{i+1})$.

\textsc{Case} 1: $(v_i, v_{i+1})$ is a step edge.
Recall that, by the definition of the step edge, we have
\begin{equation}\label{eq:tokens_movement_def}
\border_G(G_{i}) = X_{i+1}\cap \cV_i.
\end{equation}
We claim that $\cP_{i+1}:=(\cP_{i},X_{i+1})$ is a path decomposition of $G_{i+1}$ of width at most $k$ and of length $\length(\cP_{i+1}) \leq \ell_{i+1}=\ell_i+1$.
We will verify the conditions \ref{pathaxiom1}, \ref{pathaxiom2}, and \ref{pathaxiom3}.
\begin{enumerate}[label=\normalfont(\roman*),leftmargin=*,align=left]
\item Since $(v_i, v_{i+1})$ is a step edge, we have $\cV_{i+1}\setminus \cV_{i}\subseteq X_{i+1}$.
Moreover, by \eqnref{eq:Vvdef}, $X_{i+1}\subseteq \cV_{i+1}$.
Since, by the construction of the edge set of $\cG_k$, $\cV_i\subseteq\cV_{i+1}$, it follows that $\cV_{i+1} = \cV_{i} \cup X_{i+1}$.
Thus, by the induction hypothesis, $\pdspan(\cP_{i+1})=\cV_{i}\cup X_{i+1} = \cV_{i+1}$.
Hence, \ref{pathaxiom1} follows.
\item Let $\{x,y\}\in E(G_{i+1})$.
Note that $x,y\in V(G_{i+1})=\cV_{i+1}=\cV_i\cup X_{i+1}$.
We claim that some bag of $\cP_{i+1}$ contains both $x$ and $y$.
If $x$ and $y$ are both in $\cV_i$, then this follows from the induction hypothesis. If $x$ and $y$ are both in $\cV_{i+1}\setminus \cV_{i}$, then $\{x,y\}\subseteq X_{i+1}$, as required. So we may assume that $x\in \cV_{i+1}\setminus \cV_i$ and $y\in \cV_i$.
Thus, $y$ has a neighbor, i.e. $x$, in $G$ that is not in $\cV_i$, and it follows that $y\in\border_G(G_{i})$. Hence, by (\ref{eq:tokens_movement_def}), $y\in X_{i+1}$. Thus, since $\cV_{i+1}\setminus \cV_i\subseteq X_{i+1}$ (by the definition of a step edge), we have $\{x,y\}\subseteq X_{i+1}$, as required. This settles condition \ref{pathaxiom2}.
\item By \eqnref{eq:tokens_movement_def}, for all $j\leq i$ we have
$X_j\cap X_{i+1}\subseteq X_{i+1}\cap \cV_{i}=\border_G(G_i)$. Thus, by  \ref{obs:border_in_X}, $X_j\cap X_{i+1}\subseteq X_i$ which implies $X_j\cap X_{i+1}\subseteq X_j\cap  X_i$ for $j\leq i$. Therefore, since $\cP_i$ is a path decomposition of $G_i$, by the induction hypothesis we have $X_j\cap X_{i+1}\subseteq X_j\cap  X_i \subseteq X_p$ for each $j\leq p \leq i$. This proves that \ref{pathaxiom3} holds for $\cP_{i+1}$.
\end{enumerate}

Finally, it follows from the induction hypothesis and from the fact that $|X_{i+1}|\leq k+1$ that every bag of $\cP_{i+1}$ has size at most $k+1$.
Moreover,  $w(v_i,v_{i+1})=1$. Thus, $\ell_{i+1}=\ell_i+1$. On the other hand, $\length(\cP_{i+1}) = \length(\cP_i) + 1$. However, by the induction hypothesis, $\length(\cP_i) = \ell_i$. Thus,  $\length(\cP_{i+1})= \ell_{i+1}$.
Therefore, $\cP_{i+1}$ is a path decomposition of $G_{i+1}$ with $\width(\cP_{i+1})\leq k$ and $\length(\cP_{i+1})=\ell_{i+1}$.
Moreover, by \eqnref{eq:tokens_movement_def}, $\border_G(G_{i+1})\subseteq X_{i+1}$, which completes the proof of Case 1.\bigskip

\textsc{Case} 2: $(v_i, v_{i+1})$ is a jump edge.
Let $S$, $Y_S$ and $\bar\cP = (\bar X_1, \hdots, \bar X_{\bar t})$ be as in \ref{def:G:jump}.
Note that, by the definition of $\bar\cP$ in \ref{def:G:jump:d}, $\pdspan(\bar\cP)=Y_S$.
Let $\bar S$ be defined as in \ref{def:G:jump:b}, i.e., $\bar S = X_i\cap X_{i+1}$.
We claim that
\begin{equation} \label{eq:cP_i+1}
\cP_{i+1}:=(\cP_{i},\bar S\cup \bar X_1,\bar S\cup \bar X_2, \hdots, \bar S\cup \bar X_{\bar t})
\end{equation}
is a path decomposition of $G_{i+1}$ of width at most $k$ and of length $\length(\cP_{i+1}) = \ell_{i+1}=\ell_i+\bar t$. Since $\bar S\subseteq X_i \subseteq \cV_v$, \ref{def:G:jump:a} and  \ref{def:G:jump:d} imply that $\bar S$ is disjoint from $\bar X_j$ for all $j=1,\hdots,\bar t$. Thus $(\bar S\cup \bar X_1, \hdots, \bar S\cup \bar X_{\bar t})$ is a path decomposition of the subgraph $G[Y_S\cup \bar S]$.
We now verify that $\cP_{i+1}$ satisfies the conditions \ref{pathaxiom1}, \ref{pathaxiom2}, and \ref{pathaxiom3}.
\begin{enumerate}[label=\normalfont(\roman*),leftmargin=*,align=left]
\item Note that, by \ref{def:G:jump:c}, $\bar S= X_{i+1}$.
By \eqnref{eq:cP_i+1}, \ref{def:G:jump:b}, \ref{def:G:jump:d} and by the induction hypothesis,
\begin{equation} \label{eq:span_cP_i+1}
\pdspan(\cP_{i+1}) = \pdspan(\cP_i) \cup Y_S \cup \bar S =\cV_i\cup Y_S\cup X_{i+1} = \cV_i\cup Y_S=\cV_{i+1}.
\end{equation}
Thus, $\cP_{i+1}$ satisfies \ref{pathaxiom1}.
\item Let $\{x,y\}\in E(G_{i+1})$.
We claim that some bag of $\cP_{i+1}$ contains both $x$ and $y$. By \eqnref{eq:span_cP_i+1}, $x\in\cV_i$ or $x\in Y_S$, and $y\in\cV_i$ or $y\in Y_S$.

Let $x\in\cV_i$ first.
If $y\in\cV_i$, then the claim follows from the induction hypothesis.
If $y\in Y_S$, then $y\in\bar X_j$ for some $j\in\{1,\ldots,\bar t\}$, because, by \ref{def:G:jump:d}, $\bar\cP$ is a path decomposition of $G[Y_S]$. 
Moreover, if $y\in Y_S$, then, by \ref{def:G:jump:a} $x\in \border_G(G_i)$, and, by \ref{def:G:jump:b}, $x\in \bar S$.
Hence, by \eqnref{eq:cP_i+1}, the $(\ell_i+j)$-th bag of $\cP_{i+1}$ contains both $x$ and $y$.

Let now $x\in Y_S$.
The case when $y\in\cV_i$ follows by the symmetry.
Hence, $y\in Y_S$ and the claim follows from the fact that $\bar\cP$ is a path decomposition of $G[Y_S]$.
This proves condition \ref{pathaxiom2} for $\cP_{i+1}$.

\item To prove \ref{pathaxiom3} for $\cP_{i+1}$, we prove the equivalent \ref{pathaxiom3'} instead.
Let $x\in \cV_{i+1}$. By \eqnref{eq:span_cP_i+1}, $x\in Y_S$ or $x\in\cV_i$.

If $x\in Y_S=\pdspan(\bar\cP)$, then \ref{pathaxiom3'} follows from \eqnref{eq:cP_i+1}, the fact that $\bar\cP$ is a path decomposition of $G[Y_S]$ and $Y_S$ is disjoint from $\cV_i \cup \bar S$. The latter is due to \ref{def:G:jump:a}, \ref{def:G:jump:b} and \ref{def:G:jump:c}.

Thus, let $x\in\cV_i$.
If $x\in \cV_i\setminus\bar S$, then by \eqnref{eq:cP_i+1}, \ref{pathaxiom3'} follows from the induction hypothesis for $\cP_i$.
It remains to consider $x\in\bar S$.
Then, by \ref{def:G:jump:b}, $\bar S\subseteq X_i$. Thus, $x\in X_i$.
Now, let $j$ be the smallest index in $\{1,\ldots,\length(\cP_i)\}$ such that $x$ appears in the $j$-th bag of $\cP_i$. Then, due to the induction hypothesis, $x$ appears in each bag $X_{j'}$, $j'=j,\ldots,\length(\cP_i)$. Moreover,
by \eqnref{eq:cP_i+1}, $\bar S$ is included in each bag $X_{j'}$, $j'=\length(\cP_i)+1,\ldots,\length(\cP_i)+\bar t$, thus so is $x\in \bar S$. Therefore, $x$ appears in each bag $X_{j'}$, $j'=j,\ldots,\length(\cP_i)+\bar t$ which proves \ref{pathaxiom3'}.
\end{enumerate}

To complete the proof we need to show that each $x\in\border_G(G_{i+1})$ belongs to the last bag of $\cP_{i+1}$. To that end we observe that  by \eqnref{eq:cP_i+1} and by $\cV_i\subseteq\cV_{i+1}$,
\begin{equation} \label{eq:new_border}
\border_G(G_{i+1})\subseteq\border_G(G_i)\cup Y_S\cup\bar S.
\end{equation}

By \ref{def:G:jump:b}, $\border_G(G_i)=\bar S$.
We now argue that $x\notin Y_S$. Suppose for a contradiction that $x\in Y_S$. The set $Y_S$ contains only the vertices of $S$-components. Therefore, any neighbor of $x$ belongs  either to $Y_S$ or to $S$, and thus belongs to $\cV_{i+1}$ (observe that by \ref{def:G:jump} $S\subseteq \bar S$). Hence, $x\notin\border_G(G_{i+1})$ which leads to a contradiction. This proves that $x\in \border_G(G_{i+1}) \subseteq \bar S$ and thus $x$ appears in the last bag of $\cP_{i+1}$.

Finally, it follows from the induction hypothesis and condition \ref{def:G:jump:d} that every bag of $\cP_{i+1}$ has size at most $k+1$. Moreover,  $w(v_i,v_{i+1})=\bar t$. Thus, $\ell_{i+1}=\ell_i+\bar t$. On the other hand, by \eqnref{eq:cP_i+1}, $\length(\cP_{i+1}) = \length(\cP_i) + \bar t$. However, by  the induction hypothesis, $\length(\cP_i) = \ell_i$. Thus, $\length(\cP_{i+1}) = \ell_{i+1}$,
which completes the proof of Case 2.
\end{proof}

Having proved that every $s$-$t$ path of length $\ell$ in $\cG_k$ can be used to compute a path decomposition of width at most $k$ and length $\ell$ of $G$, we now prove that  if there exists a path decomposition of width $k$ and length $\ell$ of $G$, then one can find an $s$-$t$ path of length at most $\ell$ in $\cG_k$.

\begin{lemma} \label{lem:all_paths_present}
Let $k\in\{1,2,3\}$.
Let $G$ be a connected graph and let $\ell>0$. If there exists a path decomposition of width at most $k$ and length $\ell$ for $G$, then there exists an $s$-$t$ path $P$ of weighted length at most $\ell$ in $\cG_k$.
\end{lemma}
\begin{proof}
Let $\cP=(X_1,\ldots,X_\ell)$ be a path decomposition of $G$ with $\width(\cP)\leq k$. 
By \ref{ob:clean}, we may assume that $\cP$ is a clean path decomposition.
Let $\cP_i=(X_1,\hdots,X_i)$ for each $i\in\{1,\hdots,\ell\}$ and let for brevity $\cP_0$ be an empty list.
For every $S\in\cS$, let $\cf{G}_S$ be the set of green $S$-branches.
Note that all vertices of each $S$-branch in $\cf{G}_S$ are in $X_{t_1(S)}\cup\cdots\cup X_{t_2(S)}$.
For $S\subseteq V(G)$ such that $|S|\leq 4$ and $S\notin\cS$, set $\cf{G}_S:=\emptyset$.
Furthermore, for a non-empty $S\subseteq X_i$, let $l_S^i$ be the number of $S$-leaves  in  $\pdspan(\cP_i)$, $i=1,\ldots,\ell$.
Finally, for every  non-empty $S\subseteq X_i$, $i=1,\ldots,\ell$, and for each $H\in\cC_G^2(S)$ define the function $f_S^i$ as follows:
\[f_S^i(H) = \begin{cases}
1 & \mbox{if $V(H)\subseteq  \pdspan(\cP_i)$;} \\
0 & \mbox{otherwise.}
\end{cases}\]
Now, let $\prec$ be the partial order on $\cS$ defined by $S\prec S'$ if and only if $I(S)\subseteq I(S')$, where $S, S'\in \cS$.
Let $S_1,\ldots,S_p$ be all maximal elements of $\prec$.
By \ref{lem:IS-nested}, the sets $I(S_q)$ for $q=1,\ldots,p$ are pairwise disjoint.
Thus, without loss of generality, we assume that $1\leq t_1(S_1)\leq t_2(S_1)<t_1(S_2)\leq t_2(S_2)<\hdots<t_1(S_p)\leq t_2(S_p)\leq \ell$.

Define
\[I=\{0,\ldots,\ell\}\setminus\bigcup_{i=1,\ldots,p}\{t_1(S_i),\ldots,t_2(S_i)-1\}\]
and denote $I=\{s_1,\ldots,s_r\}$, where $s_1<s_2<\cdots<s_r$.
Observe that, by \ref {lem:IS-nested},for each $i\in I$ and for each $S\subseteq X_i$, $S\neq\emptyset$, the function $f_S^i$ satisfies the condition in (\ref{eq:Q(S)1}).
Thus, the following is a sequence of vertices of $\cG_k$:
\[P= s \sd (X_{s_1},R_{s_1}) \sd (X_{s_2},R_{s_2}) \sd \hdots \sd (X_{s_r},R_{s_r}) \sd t,\]
where $f_S^i$ and $l_S^i$ are used in each $R_i$, $i\in I$, i.e., $R_i(S)=(\mbG_S,f_S^i,l_S^i)$ for each non-empty $S\subseteq X_i$.
Denote $s=(\emptyset,\emptyset)=(X_0,R_0)$.
In the reminder of the proof, we show how to obtain an $s$-$t$ path of length $\ell$ in $\cG_k$ from the sequence $P$. 

We first prove that $G_{\cP_i}=G[\cV_{v_i}]$ and $\border (\cP_i)=\border_G (G[\cV_{v_i}])$, where $v_i=(X_i,R_i)$, for each $i\in I$.
By definition, $G[\cV_{v_{\ell}}]=G_{\cP}=G$ and $V(G[\cV_s])=\emptyset=V(G_{\cP_0})$ and hence let $0<i<\ell$ in the following.
We have $\border(\cP_i)\subseteq X_i$ by  \ref{pathaxiom2} and \ref{pathaxiom3}.
Moreover, $\border(\cP_i)\neq \emptyset$ since $G$ is connected and $i<l$.
Set $S=\border(\cP_i)$.
Then, for each connected component $H$ in $G-S$,  either $V(H)\subseteq \pdspan(\cP_i)$ or $V(H)\cap \pdspan(\cP_i)=\emptyset$.
Otherwise, $V(H)\cap \border(\cP_i)\neq\emptyset$, which is in contradiction with $H$ being a component in $G-S$.
Then, if $H$ is a connected component in $G-S$ but not an $S$-component such that $V(H)\subseteq \pdspan(\cP_i)$, then there exists an non-empty $S'\varsubsetneq S$ such that $H$ is an $S'$-component. Again, for  this $S'$-component either $V(H)\subseteq \pdspan(\cP_i)$ or $V(H)\cap \pdspan(\cP_i)=\emptyset$.
This proves that $G_{\cP_i}=G[\cV_{v_i}]$ which implies  $\border (\cP_i)=\border_G (G[\cV_{v_i}])$ for each $i\in I$.

Let $i\in I$ be selected in such a way that $i+1\in I$.
Now we argue that $((X_i,R_i),(X_{i+1},R_{i+1}))\in E(\cG_k)$.
Clearly, $\pdspan(\cP_i)\subseteq \pdspan(\cP_{i+1})$ and  $ \pdspan(\cP_{i+1})\setminus \pdspan(\cP_i)\subseteq X_{i+1}$.
Hence, $\cV_{v_i}\subseteq\cV_{v_{i+1}}$ and $\cV_{v_{i+1}}\setminus\cV_{v_i}\subseteq X_{i+1}$.
By \ref{pathaxiom2} and \ref{pathaxiom3}, $\border(\cP_i)\subseteq X_{i+1}\cap \pdspan(\cP_i)$.
Thus, since $\cP$ is a clean path decomposition, $\border(\cP_i)= X_{i+1}\cap \pdspan(\cP_i)$.
This implies that $\border_G(G[\cV_{v_i}])=\border(\cP_i)=X_{i+1}\cap\cV_{v_i}$.
Therefore, we just proved that there is a step edge from $(X_i,R_i)$ to $(X_{i+1},R_{i+1})$  in $\cG_k$.

We now consider $i\in I$ and $j\in I$ such that $j>i+1$ and $\{i+1,\ldots,j-1\}\cap I=\emptyset$.
Hence, there exists $q\in\{1,\ldots,p\}$ such that $t_1(S_q)=i+1$ and $t_2(S_q)=j$.
By \ref{lem:key2},
\begin{equation} \label{eq:X_i_to_X_j}
X_{i+1}\cup\cdots\cup X_j=(X_{i}\cap X_j)\cup V(\cf{G}_{S_q})\cup \Lambda,
\end{equation}
where $\Lambda$ is the set of all $S_q$-leaves in $X_{i+1},\hdots,X_j$.
Let $Y= V(\cf{G}_{S_q})\cup \Lambda$ for convenience.
We claim that there is a jump edge between $(X_{i},R_{i})$ and $(X_j\setminus Y,R_{j})$ in $\cG_k$.
Note that $S_q\subseteq X_{i}\cap X_j$.
By \eqnref{eq:X_i_to_X_j}, $X_j\setminus Y=X_i\cap X_j$ thus \ref{def:G:jump:c} is met.
By deleting the nodes other than those in $Y$ from each bag $X_{i+1},\hdots,X_j$ we obtain a path decomposition of length $j-i$ for the union of $S_q$-branches in $\cf{G}_{S_q}$ and $S_q$-leaves in $\Lambda$.
The width of this path decomposition is at most $k-|X_i\cap X_j|$ since $X_i\cap X_j$ belongs to each bag $X_t$, for $t=i+1,\hdots,j$, and no vertex in $X_i\cap X_j$ belongs to $Y$.
This implies condition \ref{def:G:jump:d}.

By \ref{lem:key2}, condition \ref{def:G:jump:a} is met.
As shown earlier, $\border(\cP_i)= X_{i+1}\cap \pdspan(\cP_i)$. By  \ref{pathaxiom3}, $X_i\cap X_j\subseteq X_{i+1}$, and by (\ref{eq:X_i_to_X_j}), $X_{i+1}\setminus Y \subseteq X_i\cap X_j$.
Thus, $\border(\cP_i)= X_i\cap X_j\cap \pdspan(\cP_i)=X_i\cap X_j=X_j\setminus Y=\border_G (G[\cV_{v_i}])$ and \ref{def:G:jump:b} is met.

Finally, clearly $\cV_{v_j}\setminus \cV_{v_i}\nsubseteq X_j\setminus Y$, because $Y\neq\emptyset$. Moreover, by \ref{def:G:jump:b}, $X_j \setminus Y=\border_G (G[\cV_{v_i}])$, which implies $((X_j\setminus Y)\setminus \border_G (G[\cV_{v_i}]))\cap \cV_{v_i}=\emptyset$.
Thus, we just proved that there is a jump edge  from $(X_{t_1(S_q)-1},R_{t_1(S_q)-1})$ to  $(X_j\setminus Y,R_{t_2(S_q)})$ with $\bar t \leq j-i$ in $\cG_k$.

By definition of $\cG_k$ and since $\ell>0$ there is a directed edge $(v_{\ell},t)$  in  $\cG_k$, and its weight is zero.
Therefore, we have shown how to build an $s$-$t$ path of weighted length at most $\ell$ in $\cG_k$ from the sequence $P$. This proves the lemma.
\end{proof}

%% file: connected_four_with_small.tex
Let $1 \leq k\leq 3$ be an integer, let $2 \leq \bagsize{1}, \bagsize{2}\leq k+1$ and 
let $C$ be a chunk graph with a big connected component $G$, $K_1$-components
$K_1^1,\hdots, K_1^{q_1}$ and $K_2$-components $K_2^1, \hdots, K_2^{q_2}$.
Let $\cG_k$ be the auxiliary graph for $G$ defined in Section \ref{sec:connected_four}. 
For every $v\in V(\cG_k)\setminus\{t\}$, define
\[W(v) = \{(v, i, j)\st i\in \{0,\ldots,q_1\}, j\in \{0,\ldots,q_2\}\}.\]
Now define the auxiliary graph $\cH_k = \cH_k( \bagsize{1}, \bagsize{2})$ corresponding to $C$ as follows.
The vertex set of $\cH_k$ is given by
\[V(\cH_k) = \{t'\}\cup\ \bigcup_{\mathclap{v\in V(\cG_k)\setminus\{t\}}} W(v).\]
Thus, $V(\cH_k)$ is constructed from $\cG_k$ by expanding every vertex $v$ of $\cG_k$ (except for the sink $t$) to a set of vertices $W(v)$ that includes all possible additions of small components to $v$. Let for brevity $s' = (s, 0, 0)$. Notice that $s' \in V(\cH_k)$. For any $x = (v,i,j)\in W(v)$, where $v = (X,R)\in V(\cG_k)\setminus\{t\}$, 
define
\[\cV'_x = \cV_v\cup \bigcup_{1\leq p\leq i} V(K_1^p) \cup \bigcup_{1\leq p\leq j} V(K_2^p).\]
The interpretation of this set is as follows.
Every path $P$ in $\cH_k$ from $s'$ to $x$ represents partial path decompositions of $C$. These partial decompositions cover $i$ isolated vertices, $j$ isolated edges and the vertices in $\cV_v$, and they can be extracted  from the consecutive vertices and edges of $P$. The construction of $\cH_k$ will guarantee that each of these partial path decompositions covers exactly those vertices of $C$ that are in $\cV'_x$.

Now, let us define the edge set of $\cH_k$. First, for any $x\in V(\cH_k)$,
define $\theta(x) = \bagsize{1}$ if $x=s'$, and $\theta(x) = k+1$ otherwise.
Similarly, for any $x\in V(\cH_k)$, let $\eta(x)=\bagsize{2}$ if $\cV_x'=V(C)$, and $\eta(x)=k+1$ otherwise.
There are four types of directed edges in $\cH_k$:
\begin{enumerate}[label={\textnormal{\bfseries (H\arabic*)}}, align=left]
\item\label{def:H:step} \textbf{Edges from a vertex in $W(v)$ to a vertex in $W(u)$ where $(u,v)$ is a step edge in $\cG_k$.}\\
Let $u, v\in V(\cG_k)\setminus\{t\}$ such that $(u,v)$ is a step edge in $\cG_k$, and let $x\in W(u)$ and $y\in W(v)$.
Let $a=|\border_G(G[\cV_u])\cup(\cV'_{y}\setminus \cV'_{x})|$.
Then, $(x,y)\in E(\cH_k)$ if and only if $\cV'_{x}\subseteq \cV'_{y}$, $a\leq \theta(x)$ and $a \leq \eta(y)$.
The weight of each edge of this type is $w'(x,y) = 1$.

\item\label{def:H:jump} \textbf{Edges from a vertex in $W(v)$ to a vertex in $W(u)$ where $(u,v)$ is a jump edge in $\cG_k$.}\\ 
Let $u, v\in V(\cG_k)\setminus\{t\}$ such that $(u,v)$ is a jump edge in $\cG_k$, and let $x=(u,i_u,j_u)\in W(u)$ and $y=(v,i_v,j_v)\in W(v)$ be such that $\cV'_{x}\subseteq \cV'_{y}$.
Denote $u=(X_u,R_u)$ and $v=(X_v,R_v)$.
Let $\cQ_{xy}'=(\bar Y_1,\ldots,\bar Y_{\bar z})$ be a path decomposition of width
\[\width(\cQ_{xy}')\leq k - |X_u\cap X_v|\]
of the graph $\bar{C}=\bar{G}\cup \bigcup_{i_u<p\leq i_v}K_1^p\cup\bigcup_{j_u<p\leq j_v}K_2^p$,
where $\bar{G}$ is defined in \ref{def:G:jump:d}.
Then, $(x,y)\in E(\cH_k)$ and $w'(x,y) = \bar z$, where $\bar z$ is taken as small as possible. We refer to any path decomposition $(\bar Y_1\cup (X_u\cap X_v),\hdots,\bar Y_{\bar z}\cup( X_u\cap X_v))$ as a \textit{witness} of the jump edge.

\item\label{def:H:inside} \textbf{Edges inside $W(v)$.}\\
Let $x,y\in W(v)$ for some $v\in V(\cG_k)\setminus\{t\}$, and let $a=|\border_G(G[\cV_v])\cup(\cV'_{y}\setminus \cV'_{x})|$.
Then, $(x,y)\in E(\cH_k)$ if and only if $\cV'_{x}\varsubsetneq \cV'_{y}$, $a\leq \theta(x)$ and $a \leq \eta(y)$.
The weight of each edge of this type is $w'(x,y) = 1$.

\item \label{def:H:to_t'} \textbf{Edges from a vertex in $W(v)$ to $t'$.}\\
For $x\in V(\cH_k)\setminus\{t'\}$, let $(x,t')\in E(\cH_k)$ if and only if $\cV_x'=V(C)$.
The weight of such edge is $w'(x,t')=0$.
\end{enumerate}

The parameter $\theta(x)$ guarantees that the first bag of the path decomposition that corresponds to an $s'$-$t'$ in $\cH_k$ path has the size at most $\bagsize{1}$. The restriction imposed by $\theta(x)$ is vacuous  for $x\neq s'$.
Similarly, the parameter $\eta(y)$ guarantees that the last bag of the path decomposition  that corresponds to an $s'$-$t'$ path has the size at most $\bagsize{2}$.
This  restriction imposed by $\eta(y)$ is vacuous  for $\cV_y'\neq V(C)$.

We also remark that our construction is fairly general, in the sense that one may argue that certain vertices and certain edges of $\cH_k$ can never be a part of a shortest $s'$-$t'$ path in $\cH_k$.
However, we proceed with this construction to avoid tedious analysis of special cases.

Before we continue with a formal analysis we give an intuition on the edge set of $\cH_k$.
We use the step edges $(u,v)$ of $\cG_k$ by adding, whenever possible, some vertices of small components to the bag that corresponds to $v$.
For a jump edge of $\cG_k$, we recalculate the path decomposition $\bar{\cP}$ used in \ref{def:G:jump:b} in such a way that the new path decomposition $\cQ_{xy}'$ in \ref{def:H:jump} covers the vertices in $\pdspan(\bar{\cP})$ and the vertices of some small components.
(If no small component vertices are added, then one may take $\cQ_{xy}'=\bar{\cP}$.)
Then, \ref{def:H:inside} allows the vertices of $K_1$- and $K_2$-components only to fill in a bag that corresponds to a vertex of $\cH_k$. In particular the edges with $v=s$ introduce bags  prior to $\alpha(G)$ and those with $\cV_v=V(G)$ introduce bags after $\beta(G)$  --- see the proof of \ref{lem:paths_for_one_big}.

Before we prove the main result of this section, we start with a useful observation:
\begin{observation}\label{obs:small_components}
Let $G'$ be a graph and let $\cP= (X_1, \hdots, X_l)$ be a clean path decomposition of $G'$.
Then, for every small component $H$ of $G'$, there exists a unique $i\in \{1, \hdots, l\}$ such that $X_i\cap V(H) \neq \emptyset$.
\qed
\end{observation}

\begin{lemma} \label{lem:paths_for_one_big}
Let $k\in\{1,2,3\}$.
Let $\bagsize{1},\bagsize{2}\in\{2,\ldots,k+1\}$ and let $C$ be a chunk graph.
There exists an $s'$-$t'$ path in $\cH_k$ of weighted length at most $l'$ if and only if there exists a  $(\bagsize{1}, \bagsize{2})$-path decomposition of $C$ of width at most $k$ and length at most $l'$.
\end{lemma}
\begin{proof}
Let $G$ be the unique big component of $C$.
First assume that there exists an $s'$-$t'$ path $P'=x_0\sd x_1\sd \hdots\sd x_{m'}\sd t'$ of weighted length $l'$ in $\cH_k$.
Denote $x_i = (v_i, a_i, a_i')$, $i=0,\ldots,m'$, and $v_i=(X_i,R_i)$.
Let $\cG_k$ be the auxiliary graph for $G$ defined in Section~\ref{sec:connected_four}.
Notice that $v_0 = s$, the source of $\cG_k$. Let for brevity $v_{m'+1} = t$, the sink of $\cG_k$.
From definition of $E(\cH_k)$ it follows that by replacing  any maximal subsequence $v_r\sd\hdots\sd v_{r'}$ such that $v_r=\cdots=v_{r'}$ in $v_0\sd \hdots \sd v_{m'}\sd v_{m'+1}$ (all $x_r, x_{r+1},\cdots,x_{r'}$ belong to $W(v_r)$) by the single vertex $v_r$, we obtain an $s$-$t$ path $P=u_0\sd \hdots\sd u_m\sd u_{m+1}$, where  $u_0=s$ and $u_{m+1}=t$, in $\cG_k$.

Now, construct $\cP=(\cR_1, \hdots, \cR_{m})$ and $\cP'=(\cR'_1, \hdots, \cR'_{m'})$ for the paths $P'$ and $P$ as follows. For each $i = 1, \hdots, m'$: 
\begin{enumerate}[label={\normalfont(\arabic*)}]
 \item If $(v_{i-1},v_i)=(u_{j-1},u_j)$ is a step edge in $\cG_k$, then set $\cR_i' = (X_i)$, and $\cR_j = (X_i)$.
 \item If $(v_{i-1},v_i)=(u_{j-1},u_j)$ is a jump edge in $\cG_k$, then let $(Y_1',\ldots,Y_{\bar z}')$ be a witness of $(v_{i-1},v_i)$ with $w(v_{i-1},v_i)=\bar z$. 
Set $\cR_i'=(Y_1',\ldots,Y_{\bar z}')$, and $\cR_j=(Y_1'\cap V(G),\ldots,Y_{\bar z}'\cap V(G))$. Observe that $\cR_j$ is a witness of $(u_{j-1},u_j)$, however, not necessarily with minimum length.
\item If $v_{i-1}=v_i=u_j$, then set $\cR_i'=\big(\border_G(\cV_{v_{j-1}})\cup(\cV_{x_i}'\setminus\cV_{x_{i-1}}')\big)$. 
\end{enumerate}

It follows from the proof of \ref{lem:path_gives_decomposition} that $\cP$ is a path decomposition of $G$ of width at most $k$.
Thus, by the construction of $\cP'$ in (1-3) and the fact that $P'$ is an $s'$-$t'$ path in $\cH_k$ it follows that $\cP'$ is a path decomposition of $C$ of width at most $k$ and of length $l'$. Moreover, by definition of $\cH_k$, the edges $(x_0,x_1)$ and $(x_{m'-1},x_{m'})$ are not jump edges in $\cH_k$.Therefore, the functions $\theta$ and $\eta$ in the definition of $\cH_k$ guarantee that the first bag and the last bag of $\cP'$ have sizes at most $\bagsize{1}$  and $\bagsize{2}$, respectively. Thus, $\cP'$ is also a $(\bagsize{1},\bagsize{2})$-path 
decomposition of $C$.

For the converse, assume that $\cP'=(X_1',\ldots,X_{l'}')$ is a $(\bagsize{1},\bagsize{2})$-path decomposition of $C$ with $\width(\cP')\leq k$.
We may assume without loss of generality that $\cP'$ is clean.  For $i=1,\ldots,l$, let $a_i$ and $b_i$ the numbers of vertices of $K_1$- and $K_2$-components respectively in $\bigcup^i_{j=1} X'_j$. By \ref{obs:small_components}, all vertices of a small component are either in $\bigcup^i_{j=1} X'_j$ or in $\bigcup^{l'}_{j=i+1} X'_j$.
Consider $\cP=(X'_1\cap V(G),\ldots,X'_{l'}\cap V(G))$. Let $Q=\{(p,q), 1\leq p<q\leq l'\}$ be the set of all pairs such that the sequence $X'_{p}\cap V(G)=\cdots=X'_{q}\cap V(G)$ is maximal. Delete all $X'_{p+1}\cap V(G),\ldots,X'_{q}\cap V(G)$ from $\cP$ for each pair $(p,q)\in Q$.
The resulting path decomposition $\cP''=(Y_1,\ldots,Y_l)$ of $G$, $l\leq l'$, is clean for $\cP'$ is clean. Let $h(i)$, $i=1,\ldots,l$, be such that $Y_i=X'_{h(i)}\cap V(G)$. By the proof of \ref{lem:all_paths_present},  there exists an $s$-$t$ path $P=s$-$v_1$-$\cdots$-$v_r$-$t$ of weighted length $l$ in the auxiliary graph $\cG_k$ for $G$,
and there are integers $1\leq s_1<\cdots<s_r\leq l$ such that if $s_{i+1}-s_{i}=1$, then $(v_{i},v_{i+1})$ is a step edge in $\cG_k$, and if $s_{i+1}-s_{i}>1$, then $(v_{i},v_{i+1})$ is a jump edge in $\cG_k$ with $(Y_{s_i+1},\ldots,Y_{s_{i+1}})$ being its witness. Define $w_i=(v_i,a_{h(i)},b_{h(i)})$ for $i=1,\ldots,r$. For each $(p,q)\in Q$, let $v_i$ be such that $h(i)=p$. Replace $w_i=(v_i,a_{p},b_{p})$ by $w_i, w^1_i=(v_i,a_{p+1},b_{p+1}),\ldots,w^{q-p}_i=(v_i,a_{q},b_{q})$ in $P$. All these new vertices are in $W(v_i)$ and the edges between them in $E(\cH_k)$ by \ref{def:H:inside}.  Let $u_1$-$\cdots$-$u_{l''}$ be the resulting sequence, where $u_i=(v'_i, a_i,b_i)$ for each $i\in\{1,\ldots,l''\}$. Clearly,  if $v'_i\neq v'_{i+1}$, then $(v'_i, v'_{i+1})$ is either a step edge or a jump edge in $\cG_k$. Moreover, if
$(v'_i, v'_{i+1})$ is a jump edge, then $(v'_i, v'_{i+1})=(v_{j},v_{j+1})$ for some $j$ and $(X'_{s_j+1},\ldots,X'_{s_{j+1}})$ is a witness of $(v_{j},v_{j+1})$.
Thus, by definition of $V(\cH_k)$, \ref{def:H:step}, and \ref{def:H:jump} the edges $(u_i,u_{i+1})$ belong to $E(\cH_k)$.
Therefore,  $P'=s'$-$u_1$-$\cdots$-$u_{l''}$-$t'$ is an $s'$-$t'$ path in $\cH_k$. The path is not longer than $l'$ since there is one to one correspondence between the jump edges in $P$ and in $P'$ and their witnesses are of the same length $s_{i+1}-s_i$.
\end{proof}

We conclude this section with a formal statement of the algorithm for computing a minimum length $(\bagsize{1},\bagsize{2})$-path decomposition of a chunk graph $C$ --- see Algorithm~\ref{alg:chunk}. Note that $\pw(C)>0$, because $C$ is assumed to contain a big component $G$.
\medskip
\begin{algorithm}
\caption{Finding a minimum length $(\bagsize{1},\bagsize{2})$-path decomposition of a chunk graph $C$.}
\label{alg:chunk}
\begin{algorithmic}
 \REQUIRE A chunk graph $C$, $k\in\{1,2,3\}$ and $\bagsize{1},\bagsize{2}\in\{2,3,4\}$.
 \ENSURE  A minimum length $(\bagsize{1},\bagsize{2})$-path decomposition of $C$ of width $k$ or `failure' if $\pw(C)>k$.
   \STATE Let $G$ be the big component of $C$.

   \IF {$\pw(G)>k$ (use the algorithm in \cite{Bodlaender96} to calculate $\pw(G)$)}
     \RETURN `failure'.
   \ENDIF

   \STATE Construct the auxiliary graph $\cH_k$.
   \STATE Find a shortest $s'$-$t'$ path $P$ in $\cH_k$.
   \STATE Use $P$ to construct the corresponding path decomposition $\cP$ of $C$ of the same length as $P$.
   \RETURN $\cP$.
\end{algorithmic}
\end{algorithm}
\medskip

Note that $\pw(G)\leq k$ if and only if $\pw(C)\leq k$.
We have $|W(v)| \leq (1+q_1)(1+q_2)$ for each $v\in V(\cG_k)$ and hence
$|V(\cH_k)| = O(q_1q_2 |V(\cG_k)|)$.
Finally, \ref{thm:connected} and \ref{lem:paths_for_one_big} imply \ref{thm:chunkgraphs}.

%% file: disconnected_four.tex
We start with a definition of parallel processing of big components.
\begin{definition} \label{def:parallel}
Let $G$ be a graph and let $\cP = (X_1, \hdots, X_l)$ be a path decomposition of $G$. We say that two big connected components $G_1$ and $G_2$ of $G$ are \emph{processed in parallel} if
\[|\{\alpha_{\cP}(G_1),\ldots,\beta_{\cP}(G_1)\}\cap\{\alpha_{\cP}(G_2),\ldots,\beta_{\cP}(G_2)\}|\geq 2.\]
\end{definition}
The main result of this subsection, \ref{lem:big_not_in_parallel}, shows that when constructing a minimum length path decomposition of width $k$, $k\leq 3$, we may limit ourselves to path decompositions with no two big components processed in parallel. 
Notice that the NP-completeness proof of Section \ref{sec:disconnected_five} shows that a minimum length path decomposition of width $4$ may require parallel processing. Thus, the parallel processing of big components is one of the main features distinguishing (in terms of the computational complexity) between the problems $\PathLenFixed{k}$, $k\leq 3$, and the problems $\PathLenFixed{k}$, $k\geq 4$.

\begin{lemma} \label{lem:big_not_in_parallel}
Let $G$ be a graph and let $k\in\{1,2,3\}$.
If $\pw(G)\leq k$, then there exists a minimum length path decomposition of width $k$ of $G$ such that no two big connected components of $G$ are processed in parallel.
\end{lemma}
\begin{proof}
For any path decomposition $\cP=(X_1,\ldots,X_l)$ of $G$, let $\zeta(\cP)$ be the smallest $t\in \{1 ,\hdots, l\}$ such that two big components are processed in parallel in step $t$, i.e., both components have a non-empty intersection with $X_t$ and with $X_{t+1}$ (set $\zeta(\cP)=\infty$ if no such $t$ exists). 
Assume without loss of generality that $\cP$ is a minimum length path decomposition of width $k$ of $G$ with maximum $\zeta(\cP)$.
We will show that no two big components are processed in parallel in $\cP$, by showing that if there are two such components, i.e., if $\zeta(\cP)\neq\infty$, then we can increase $\zeta$ without increasing the length of the path decomposition, contrary to our choice of $\cP$. This process is done in two stages. First, we construct a longer path decomposition  $\cP'$ from  $\cP$. Next, we show how  $\cP'$ can be shortened to a path decomposition $\cP''$ of length exactly $\length(\cP)$, but with a larger value of $\zeta$.

So suppose that there are two big connected components, say $G_1$ and $G_2$, that are processed in parallel in $\cP$. Let $\{t_1, \hdots, t_2\} =\{\alpha(G_1),\ldots,\beta(G_1)\}\cap\{\alpha(G_2),\ldots,\beta(G_2)\}$. We may assume that $G_1$ and $G_2$ are chosen so that $t_1=\zeta(\cP)$ and $\beta(G_1)\leq \beta(G_2)$. Let $q = t_2 - t_1 + 1$. By \ref{def:parallel}, $q\geq 2$. For notational convenience, define for $i=1,\hdots, q$, $Y_i = X_{t_1+i-1} \cap V(G_1)$ and $Z_i = X_{t_1+i-1}\setminus V(G_1)$. Observe that each of the sets $Y_i$ and $Z_i$ is non-empty, and $|Y_i|+|Z_i| = |Y_i\cup Z_i| \leq k+1$.
Now define a sequence (of subsets of $V(G)$) $\cP'$ and then we prove that $\cP'$ is a path decomposition of $G$ of width $k$ and of length $l + q$. We also introduce the two following partial path decompositions $\cQ_1$ and $\cQ_2$:
\[
\cP' = 
\begin{cases}
\biggl.(\underbrace{Y_1, \hdots, Y_q}_{:=\cQ_1}, \quad \underbrace{X_1, \hdots, X_{t_1-1}, \quad Z_1, \hdots, Z_q, \quad X_{t_2+1}, \hdots, X_{l}}_{:=\cQ_2})\biggr. & \mbox{if $\alpha(G_1) \geq \alpha(G_2)$;}\\
\biggl.(\underbrace{X_1, \hdots, X_{t_1-1}, \quad Y_1, \hdots, Y_q}_{:=\cQ_1}, \quad \underbrace{Z_1, \hdots, Z_q, \quad X_{t_2+1}, \hdots, X_{l}}_{:=\cQ_2})\biggr. & \mbox{if $\alpha(G_1) < \alpha(G_2)$.} 
\end{cases}
\]
We claim that $\cP'$ is a path decomposition of $G$.
First suppose that $\alpha(G_1) \geq \alpha(G_2)$.
Then $t_1=\alpha(G_1)$ and $t_2=\beta(G_1)$.
Therefore, by \ref{obs:O2}, $V(G_1)\cap X_t \neq\emptyset$ if and only if $t_1\leq t\leq t_2$.
Thus, $\cQ_1$ is a path decomposition of $G_1$, and $\cQ_2$ is a path decomposition of $G-G_1$.
Therefore, $\cP$ is a path decomposition of $G$.
Next, suppose that $\alpha(G_1) < \alpha(G_2)$.
Clearly, $\cP'$ satisfies \ref{pathaxiom1}.
Also, since there is no edge between $x\in Y_i$ and $y\in Z_i$ in $G$, and $\cP$ is a path decomposition of $G$, it follows that \ref{pathaxiom2} is met by $\cP'$.
Finally, to show that $\cP'$ satisfies \ref{pathaxiom3}, we argue that
\[A= (X_1\cup\cdots\cup X_{t_1-1})\cap (Z_1\cup\cdots\cup Z_q\cup X_{t_2+1}\cup\cdots\cup X_l)=\emptyset.\]
Suppose for a contradiction that $A\neq\emptyset$ and let $x\in A$ be selected arbitrarily.
Since $x$ appears in at least two bags of $\cP$, by \ref{obs:small_components},  $x$ belongs to a big component $G'$.
However, $G_2\neq G'$ since $\alpha(G_2)=t_1$, and  $G_1\neq G'$ since $\beta(G_1)=t_2$. Therefore, $G_1$ and $G'$ are two big components that are processed in parallel starting at $t'<t_1$, contrary to our choice of $G_1$ and $G_2$.
Hence, $A=\emptyset$.
Moreover, we observe that since $\beta(G_1)=t_2$, it follows that there is no $x$ such that $x\in Y_1\cup\cdots\cup Y_q$
and $x\in X_{t_2+1}\cup\cdots\cup X_l$. Thus, \ref{pathaxiom3} is met by $\cP'$. Therefore, $\cP'$ is a path decomposition of $G$.

In what follows, we describe an algorithm that takes the path decomposition $\cP'$ as an input and returns a path decomposition $\cP''$ of length exactly $\length(\cP)$.
We consider the case of $k=3$, as the other cases are analogous.
The algorithm uses the following two length decreasing operations that preserve the property that $\cP'$ is a path decomposition:
\begin{itemize}
\item If $\cP'$ has a bag $Z_i$ or $Y_i$ of cardinality one, then, due to \ref{lem:continuity}, we may delete it.
\item If $\cP'$ has $s\geq 2$ consecutive non-empty bags $X_{t+1}, \hdots, X_{t+s}$ such that $|X_{t+1}\cup \hdots \cup X_{t+s}|\leq 4$, then we may replace $X_{t+1},\hdots, X_{t+s}$ by one bag containing $X_{t+1}\cup \hdots \cup X_{t+s}$. (i.e., we merge the bags $X_{t+1}, \hdots, X_{t+s}$.)
\end{itemize}

We define three types of subintervals of $\{1, \hdots, q\}$. For a subinterval $J = \{a, \hdots, b\}$, we say that
\begin{itemize}
\item $J$ is of Type A if $|J|\geq 2$ and $|Y_j| = 2$ for all $j\in J$;
\item $J$ is of Type B if $|J|\geq 3$, $|Y_a| = |Y_b| = 2$, and $|Y_j| \in \{1,3\}$ for all $j\in J\setminus\{a,b\}$;
\item $J$ is of Type C if $|Y_j|=2$ for at most one index $j\in J$.
\end{itemize}
We now partition the interval $\{1, \hdots, q\}$ in the following way. Let $\cA$ be the collection of all maximal subintervals $J$ of $\{1, \hdots, q\}$ of type A. Next, let $\cB$ be the collection of all maximal subintervals of $\{1,\hdots, q\}\setminus \bigcup \cA$ of type B. Finally, let $\cC$ be the collection of all maximal subintervals of $\{1,\hdots,q\}\setminus \bigcup (\cA\cup \cB)$. By construction, the intervals in $\cC$ are of type C. Moreover, the intervals in $\cA\cup \cB\cup \cC$ form a partition of $\{1, \hdots, q\}$. 

The length-reduction subroutine consists of using three different subroutines, one for each type A, B, C. These subroutines all work as follows. Given an interval $\{a, \hdots, b\}$, they take two partial path decompositions $(Y_a, \hdots, Y_b)$ and $(Z_a, \hdots, Z_b)$, and return two new partial path decompositions $(Y'_{1}, \hdots, Y'_{a'})$ and $(Z'_1, \hdots, Z'_{b'})$ that satisfy: $a' + b' = b-a+1$, $\bigcup_{i=1}^{a'} Y'_i = \bigcup_{i=a}^b Y_i$, $\bigcup_{i=1}^{b'} Z'_i = \bigcup_{i=a}^b Z_i$, $Y_a\subseteq Y'_1$, $Y_b\subseteq Y'_{a'}$, $Z_a\subseteq Z'_1$, and $Z_{b}\subseteq Z'_{b'}$. Thus, we may replace the partial path decompositions $(Y_a, \hdots, Y_b)$ and $(Z_a, \hdots, Z_b)$ in $\cP'$ by $(Y'_{1}, \hdots, Y'_{a'})$ and $(Z'_1, \hdots, Z'_{b'})$, respectively. After running the appropriate subroutine for each of the intervals and performing these replacements, we end up with a path decomposition of $G$ of length exactly $\length(\cP)$.

Let $J=\{a,\hdots,b\}\in\cA\cup\cB\cup\cC$. We consider the following cases:
\begin{list}{}{}
\item[Case 1:] $J$ is of Type A. We use the following subroutine which decreases the length of $\cP$ by $|J|$. Notice that, since $|Y_j|=2$ for all $j\in J$, it follows that $|Z_j|\leq 2$ for all $j\in J$.

\medskip
\begin{algorithmic}
\STATE Set $\cY = \cZ = \emptyset$.
\IF {$|J|$ is odd}
	\STATE Append $Y_{a}\cup Y_{a+1}\cup Y_{a+2}$ to $\cY$. \COMMENT{Merge $Y_a, Y_{a+1}, Y_{a+2}$}
	\STATE Append $Z_{a}\cup Z_{a+1}$ and $Z_{a+2}$ to $\cZ$. \COMMENT{Merge $Z_a$ and $Z_{a+1}$}
	\STATE $a'=a+3$.
\ELSE
	\STATE $a'=a$.
\ENDIF
\IF {$a'<b$}
	\FOR {$j = a', a'+2, \hdots, b-1$}
		\STATE Append $Y_j \cup Y_{j+1}$ to $\cY$. \COMMENT{Merge $Y_j$ and $Y_{j+1}$}
		\STATE Append $Z_j \cup Z_{j+1}$ to $\cZ$. \COMMENT{Merge $Z_j$ and $Z_{j+1}$}
	\ENDFOR
\ENDIF
\end{algorithmic}
\medskip

To see that this subroutine is correct, observe that if $Y_a, Y_{a+1}, Y_{a+2}$ all have cardinality $2$, then, since $G_1$ is connected, by \ref{lem:continuity}, $Y_a\cap Y_{a+1}\neq\emptyset$ and $Y_{a+1}\cap Y_{a+2}\neq\emptyset$. Hence $|Y_a\cup Y_{a+1}\cup Y_{a+2}|\leq 4$.

\item[Case 2:] $J$ is of Type B. For each $i\in \{a+1,\hdots, b-1\}$, if $|Y_i| = 1$, then set $Y_i := \emptyset$; if $|Y_i|=3$, then set $Z_i := \emptyset$.
Next, let $p>a$ be the smallest index such that $Y_p\neq\emptyset$.
Such a $p$ exists, because $Y_b\neq\emptyset$.
Set $Y_a := Y_a\cup Y_p$, and $Y_p := \emptyset$.
Since $G_1$ is connected, by \ref{lem:continuity}, $Y_a\cap Y_p\neq\emptyset$, thus $|Y_a\cup Y_p| \leq 4$.
Finally, let $r>a$ be the smallest index such that $Z_r\neq\emptyset$.
Such $r$ exists, because $Z_b\neq\emptyset$.
Set $Z_a := Z_a\cup Z_r$, and $Z_r := \emptyset$.
Again, by \ref{lem:continuity}, $|Z_a\cup Z_r| \leq 4$. The removal of empty bags decreases the length of $\cP'$ by exactly $|J|$.

\item[Case 3:] $J$ is of Type C. Then, we use the following subroutine to decrease the length of $\cP$ by $|J|$. We may assume without loss of generality that $|Y_b| \neq 2$ because otherwise we can run the subroutine backwards.

\medskip
\begin{algorithmic}
\STATE Set $i=a$ and $\cY = \cZ = \emptyset$.
\WHILE {$i \leq b$}
	\IF {$|Y_i| =  1$}
		\STATE Append $Z_i$ to $\cZ$ and set $i := i + 1$. \COMMENT{Remove $Y_i$}
	\ELSIF{$|Y_i| = 3$}
		\STATE Append $Y_i$ to $\cY$ and set $i := i + 1$. \COMMENT{Remove $Z_i$}
	\ELSIF{$|Y_i| = 2$}
		\IF{$|Y_{i+1}| = 1$}
			\STATE Append $Y_i$ to $\cY$. \COMMENT{Remove $Y_{i+1}$}
			\STATE Append $Z_i\cup Z_{i+1}$ to $\cZ$. \COMMENT{Merge $Z_i, Z_{i+1}$}
		\ELSIF{$|Y_{i+1}| = 3$}
			\STATE Append $Y_i\cup Y_{i+1}$ to $\cY$. \COMMENT{Merge $Y_i, Y_{i+1}$}
			\STATE Append $Z_{i}$ to $\cZ$. \COMMENT{Remove $Z_{i+1}$}
		\ENDIF
		\STATE Set $i := i + 2$.
	\ENDIF
\ENDWHILE
\end{algorithmic}
To see that this subroutine is correct, notice that if $|Y_i|=2$, then $i < b$ by the assumption that $|Y_b|\neq 2$.
Thus, $Y_{i+1}$ and $Z_{i+1}$ are well-defined. Moreover, because $J$ is an interval of Type C, $|Y_{i+1}|\in \{1,3\}$.
Also notice that if $|Y_i|=2$ and $|Y_{i+1}| = 1$, then $|Z_i|\leq 2$ and $|Z_{i+1}|\leq 3$.
Since, due to \ref{lem:continuity}, both $Z_i$ and $Z_{i+1}$ contain vertices from the connected component $G_2$, it follows that $|Z_i\cup Z_{i+1}| \leq 4$.
\end{list}
From cases 1-3 we obtain that $\length(\cP'')=\length(\cP')-q$ as required.
To see that $\zeta(\cP'') > \zeta(\cP)$, observe that $\zeta(\cP) = t_1$ and $\zeta(\cP'') \geq t_1+1$. 
\end{proof}

We remark that we obtain a stronger analogue of \ref{lem:big_not_in_parallel} for $k\in \{1, 2\}$.
\begin{observation} \label{ob:big_no_overlaps}
Let $G$ be a graph and let $k\in\{1,2\}$.
If $\pw(G)\leq k$, then there exists a minimum length path decomposition $\cP$ of width $k$ of $G$ such that for any two big components $G_1$ and $G_2$ of $G$ it holds $\{\alpha_{\cP}(G_1),\ldots,\beta_{\cP}(G_1)\}\cap\{\alpha_{\cP}(G_2),\ldots,\beta_{\cP}(G_2)\}=\emptyset$.
\qed
\end{observation}

We end with the following corollary that follows immediately from \ref{lem:big_not_in_parallel} and \ref{lem:continuity}.

\begin{corollary}\label{corollary:big_components}
Let $G$ be a graph and let $k\in\{1,2,3\}$.
If $\pw(G)\leq k$, then there exists a minimum length path decomposition $\cP= (X_1, \hdots, X_l)$ of width $k$ of $G$ such that for every big component $G'$ of $G$, $|X_i\cap V(G')|\geq 2$ for $i=\alpha(G'),\ldots,\beta(G')$.
\qed
\end{corollary}

%% file: optimal_characteristics.tex
We now deal with two issues alluded to earlier in Section 5. One is the appropriate concatenation of minimum length decompositions of given chunk graphs. The other is the appropriate selection of a minimum length path decomposition for a given chunk graph. Note that the sizes of the first and the last bag of a minimum length path decomposition of a chunk graph may impact the length of its subsequent concatenation with minimum length decompositions of other chunk graphs. By \ref{ob:big_no_overlaps}, we obtain the following.
\begin{observation} \label{ob:chunks:concatenation}
Let $G$ be a graph with $\pw(G)\leq k$, $k\in\{1,2\}$, and let $c$ be the number of big connected components of $G$.
There exist chunk graphs $C^1, \hdots, C^c$ such that $G = C^1 \cup \cdots \cup C^c$ and
$\cP = (\cQ_1,\ldots,\cQ_c)$ is a minimum length path decomposition of width at most $k$ of $G$,
where $\cQ_i$ is any minimum length path decomposition of $C^i$, $i=1,\ldots,c$.
\qed
\end{observation}
Hence, we assume $k=3$ for the reminder of this subsection.

We first distinguish four types of path decompositions of chunk graphs. We say that a path decomposition $\cQ=(X_1,\ldots,X_{l})$ of a chunk graph $C$ is of
\begin{itemize}
\item Type~$\bA$, if $|X_1|\leq 2$ and $|X_{l}|\leq 2$;
\item Type~$\bB_1$, if $|X_1|\leq 2$ and $|X_{l}|>2$; 
\item Type~$\bB_2$, if $|X_1|>2$ and $|X_{l}|\leq 2$; and
\item Type~$\bC$, if $|X_1|>2$ and $|X_{l}|>2$.
\end{itemize}
We  say that $\cQ$ is of Type $\bB$ if $\cP$ is either of Type $\bB_1$ or of Type $\bB_2$. 
Notice that if $\cQ$ is of Type $\bB_1$, then $(X_l,\ldots,X_1)$ is of Type $\bB_2$, and vice versa.

We say that a path decomposition $\cQ$ of a chunk graph $C$ is \emph{type-optimal} if $\cQ$ has minimum length and
\begin{itemize}
 \item $\cQ$ is of Type~$\bC$ if no minimum length path decomposition of Type~$\bA$ or~$\bB$ of $C$ exists, or
 \item $\cQ$ is of Type~$\bB$ if no minimum length path decomposition of Type~$\bA$ of $C$ exists, or
 \item $\cQ$ is of Type~$\bA$ otherwise.
\end{itemize}

Let $\cP_1=(X_1^1,\ldots,X_{l_1}^1)$ and $\cP_2=(X_1^2,\ldots,X_{l_2}^2)$ be two path decompositions of graphs $G^1$ and $G^2$, respectively. We define the \emph{concatenation} of $\cP_1$ and $\cP_2$, denoted $\cP_1 \oplus \cP_2$, as follows
\begin{align*}
\cP_1 \oplus\cP_2 = 
\begin{cases}
(X_1^1,\ldots,X_{l_1-1}^1,X_{l_1}^1\cup X_1^2,X_2^2,\ldots,X_{l_2}^2), & \mbox{if $|X_{l_1}|\leq 2$ and $|X_1^2|\leq 2$;} \\
(X_1^1,\ldots,X_{l_1}^1,X_1^2,\ldots,X_{l_2}^2), & \mbox{otherwise.}
\end{cases}
\end{align*}
Clearly, if $\width(\cP_1)\leq 3$ and $\width(\cP_2)\leq 3$, then $\cP_1\oplus\cP_2$ is a path decomposition of $G^1\cup G^2$ of width at most $3$.
Observe that $\length[\cP_1]+\length[\cP_2]-1\leq\length[\cP_1 \oplus\cP_2]\leq\length[\cP_1]+\length[\cP_2]$.
Then, let
\[\cP_1\oplus \cP_2\oplus \cdots \oplus \cP_c=(\cdots((\cP_1\oplus\cP_2)\oplus\cP_3)\cdots)\oplus\cP_c.\]

We now show that any minimum length path decomposition of width $k=3$ of any graph $G$ can be expressed as
the concatenation of type-optimal path decompositions of chunk graphs whose union is $G$.

\begin{lemma} \label{lem:concatenation}
Let $G$ be a graph with $\pw(G)\leq 3$ and let $c$ be the number of big connected components of $G$.
There exist chunk graphs $C^1, \hdots, C^c$ such that $G = C^1 \cup \cdots \cup C^c$ and for each $i=1,\ldots,c$ there exists a type-optimal path decomposition $\cQ_i$ of $C^i$ such that
$\cP = \cQ_1\oplus \cdots \oplus \cQ_c$ is a minimum length path decomposition of width at most 3 of $G$.
\end{lemma}
\begin{proof}
By \ref{lem:big_not_in_parallel}, there exists a minimum length path decomposition $\cP=(X_1,\ldots,X_l)$ of $G$, $\width(\cP)\leq 3$, in which no two big components of $G$ are processed in parallel.
Let $G_1,\ldots,G_c$ be all big components of $G$, and let $L = V(G)\setminus \bigcup_{i=1}^c V(G_i)$.
We claim that $\alpha(G_i)\neq \alpha(G_{j})$ for all distinct $i,j\in \{1,\ldots,c\}$.
Suppose $\alpha(G_i)= \alpha(G_{j})=t^*$ for some $i\neq j$.
By \ref{corollary:big_components} and by the fact that $|X_{t^*}| \leq 4$, we obtain that $X_{t^*}$ contains exactly 2 vertices of each of $G_i$, $G_j$. Since $G_i$ and $G_j$ are big components, it follows that $X_{t^*+1}$ also contains at least one vertex from each of $G_i$, $G_j$. But this implies that $G_i$ and $G_j$ are processed in parallel, a contradiction.
Thus, we may assume without loss of generality that $\alpha(G_i) < \alpha(G_{i+1})$ for each $i=1,\ldots,c-1$. It follows from \ref{lem:big_not_in_parallel} that
$\beta(G_i) \leq \alpha(G_{i+1})$. 

Now, let us define the chunk graphs $C^1, \hdots, C^c$ and the corresponding path decompositions $\cQ_1,\ldots,\cQ_c$.
Define $\alpha_1 = 1$, and let $\alpha_i=\alpha(G_i)$ for $i=2,\hdots,c$.
Define $\omega\colon L\to \{1,\hdots, c\}$ by $\omega(v) = \max\{i\colon \alpha(v) \geq \alpha_i, 1 \leq i \leq c\}$.
For $i\in\{1,\hdots,c\}$, let $C^i = G\left[ V(G_i)\cup \omega^{-1}(i) \right]$ and let
\[
\cQ_i = \left( X_{\alpha(C^i)} \cap V(C^i) , \hdots, X_{\beta(C^i)} \cap V(C^i) \right).
\]
By this construction, we have $\alpha(G_i)=\alpha(C^i)$, and,  moreover, if $\beta(G_{i-1})<\beta(C^{i-1})$, then $\beta(C^{i-1})<\alpha(C^i)$  for $i=2,\hdots,c$. Thus, if $\beta(C^{i-1})=\alpha(C^i)$, then  $\beta(G_{i-1})=\alpha(G_i)$. Therefore, by \ref{corollary:big_components}, $| X_{\beta(C^{i-1})}|=2$ and $| X_{\alpha(C^{i})}|=2$ in such case. Finally, if $\beta(C^{i-1})<\alpha(C^i)$ for any $i\in\{2,\hdots,c\}$, then $|X_{\beta(C^{i-1})} \cap V(C^{i-1})|\geq 3$ or  $|X_{\alpha(C^{i})} \cap V(C^{i})|\geq 3$. Otherwise, $|X_{\beta(C^{i-1})}\cup X_{\alpha(C^{i})}|\leq 4$ and $\cP$ would not be a minimum-length path decomposition, thus contradiction. Therefore, we just proved that $\cP=\cQ_1\oplus\cdots\oplus\cQ_c$. 
Possibly by choosing $\cP$ differently, we may assume without loss of generality that $\cP$ has the maximum number of indices $i\in \{1,\hdots,c\}$ such that $\cQ_i$ is type-optimal for $C^i$. 

We finish the proof by showing that, for each $i\in \{1, \hdots, c\}$, $\cQ_i$ is a type-optimal path decomposition of $C^i$.
Suppose for a contradiction that this claim does not hold for some $i\in \{1,\hdots,c\}$.
Let $\cQ'_i$ be a type-optimal path decomposition of $C^i$.

If $\length(\cQ_i)=\length(\cQ'_i)$, then by the definition of the types of decompositions,
\[\length(\cQ_{1}\oplus\cdots\oplus\cQ_{i-1}\oplus\cQ'_i\oplus\cQ_{i+1}\oplus\cdots\oplus\cQ_{c})\leq\length(\cP),\]
which contradicts our choice of $\cP$.
Hence, $\length(\cQ_i)>\length(\cQ'_i)$.
If $\cQ_i$ is not of Type~$\bA$ or $\cQ'_i$ is not of Type $\bC$, then
\[\length(\cQ_{1}\oplus\cdots\oplus\cQ_{i-1}\oplus\cQ'_i\oplus\cQ_{i+1}\oplus\cdots\oplus\cQ_{c})\leq\length(\cP),\]
and if $\cQ_i$ is of Type~$\bA$ and $\cQ'_i$ is of Type~$\bC$, then
\[\length[\cQ_{1}\oplus\cdots\oplus\cQ_{i-1}\oplus\cQ_{i+1}\oplus\cdots\oplus\cQ_{c}\oplus\cQ'_i]\leq\length(\cP),\]
which again contradicts our choice of $\cP$. This completes the proof of this case.
\end{proof}

\ref{lem:concatenation} implies that, we may construct a minimum-length path decomposition of width at most $k$ of graph $G$ 
by constructing a type-optimal path decomposition of each chunk graph of $G$ separately, and then concatenating the resulting type-optimal path decompositions. 
The only two caveats here are: the optimal number of $K_1$- and $K_2$-components in each chunk graph and the optimal ordering of the type-optimal path decompositions. We deal with the latter in the next subsection.

\subsection{Optimal ordering of path decompositions of chunk graphs} \label{sec:ordering}

Let us now assume that path decompositions $\cQ_1,\ldots,\cQ_c$ of chunk graphs $C^1,\hdots, C^c$, respectively, are given and $G = C^1\cup \hdots \cup C^c$. In this subsection, our goal is to find an optimal order, i.e., the one that minimizes the length of the resulting path decomposition of $G$, in which to concatenate $\cQ_i$'s.
Note that, by \ref{ob:chunks:concatenation}, the $\cQ_i$'s can be concatenated in any order when $k<3$.
Thus, we assume that $k=3$ in the remainder of this section.
To obtain the order we determine the permutation of $\cQ_i$'s and, if a particular $\cQ_i$ is of type $\bB$, then determine whether $\cQ_i$ should be of Type $\bB_1$ or of Type $\bB_2$ in the concatenation (hence, such a $\cQ_i$ may be reversed before producing the concatenation).

Let $a$, $b_1$, $b_2$ denote the number of $\cQ_i$'s of Type~$\bA$, $\bB_1$, $\bB_2$, respectively, $i=1,\ldots,c$, and let $b = b_1+b_2$. 
We say that the sequence $\cQ_1,\hdots,\cQ_c$ is \textit{in normal form} if $b_1 = \floor{b/2}$, $b_2 = \ceil{b/2}$, and they are ordered as follows: \begin{enumerate}[label={\normalfont(\roman*)},align=left]
\item If $b = 0$, then $\cQ_1, \hdots, \cQ_{a}$ are of Type~$\bA$ and $\cQ_{a+1}, \hdots, \cQ_{c}$ are of Type~$\bC$; (Informally, using the Kleene star notation, the pattern is non-empty and belongs to $\bA^*\bC^*$.) 
\item If $b=1$, then
$\cQ_{1}, \hdots, \cQ_{c-a-1}$ are of type $\bC$;
$\cQ_{c-a}$ is of Type~$\bB_2$; and
$\cQ_{c-a+1}, \hdots, \cQ_{c}$ are of type $\bA$.  (Informally: the pattern belongs to $\bC^*\bB_2\bA^*$.)
\item If $b>1$, then $\cQ_1$ is of Type~$\bB_2$; 
$\cQ_{2}, \hdots, \cQ_{a+1}$ are Type~$\bA$;
$\cQ_{a+2}$ is of Type~$\bB_1$;
$\cQ_{a+3}, \hdots, \cQ_{a+c+2}$ are of Type~$\bC$; and
$\cQ_{a+c+3}, \hdots, \cQ_{c}$ alternate Type $\bB_2$ and
$\bB_1$, starting with Type $\bB_2$. (Informally: the pattern belongs to $\bB_2\bA^*\bB_1\bC^*(\bB_2\bB_1)^*$ for an even $b$, and to $\bB_2\bA^*\bB_1\bC^*\bB_2(\bB_1\bB_2)^*$) for an odd $b$.)
\end{enumerate}

The following lemma implies that the normal form is an optimal way of ordering and reversing the path decompositions
$\cQ_1, \hdots, \cQ_c$ that results in a minimum length  
path decomposition of $G$. For convenience, define 
\[
\mu(a,b) 
=
\begin{cases}
\max(0,a - 1) & \mbox{if $b = 0$,} \\
a + \floor{b/2} & \mbox{if $b > 0$.} \\
\end{cases}
\]

\begin{lemma}\label{lemma:properordering}
Let $G$ be a graph with $\pw(G)=3$. Let $C^1,\hdots, C^c$ be any chunk graphs such that $G = C^1\cup \hdots \cup C^c$.
Let $\cQ_i$ be a path decomposition of width at most $3$ of $C^i$, $i=1,\ldots,c$.
Let $a$ and $b$ be the numbers of path decompositions among $\cQ_1,\ldots,\cQ_c$ of Type~$\bA$ and of Type~$\bB$, respectively.
Then, 
\begin{equation} \label{eqn:minlenab}
\min_{\pi, \mbf\in\{0,1\}^c} \biggl\{ \length\left(\cQ_{\pi(1)}^{f_1}\oplus\hdots \oplus\cQ_{\pi(c)}^{f_c}\right) \biggr\}
= 
\sum_{i=1}^c \length(\cQ_i) - \mu(a,b).
\end{equation}
where the minimization is over all permutations $\pi\colon\{1,\ldots,c\}\to\{1,\ldots,c\}$, $\cQ_i^0 = \cQ_i$, $\cQ_i^1$ is the reverse of $\cQ_i$ and $\mbf=(f_1,\ldots,f_c)$.
\end{lemma}
\begin{proof}
Consider an arbitrary permutation $\pi\colon\{1,\ldots,c\}\to\{1,\ldots,c\}$ and a 
vector $\mbf=(f_1,\ldots,f_c)\in \{0,1\}^{c}$. Let
$\nu = \nu(\pi,\mbf)$ denote the number of pairs $(i,i+1)$, $i\in\{1,\ldots,c-1\}$, such that
$\length[\cQ_{\pi(i)}^{f_i}\oplus \cQ_{\pi(i+1)}^{f_{i+1}}] = \length[\cQ_{\pi(i)}^{f_i}] + \length[\cQ_{\pi(i+1)}^{f_{i+1}}] - 1$. 
We refer to these pairs as \textit{matchups}. Clearly, since each chunk graph
$C^i$ has a big connected component,
\begin{equation} \label{eq:chunks}
\length[\cQ_{\pi(1)}^{f_1}\oplus \cdots\oplus \cQ_{\pi(c)}^{f_c}] = \sum_{i=1}^c \length[\cQ_i] - \nu(\pi,\mbf).
\end{equation}

By  \ref{corollary:big_components}, and by definition of $\oplus$, each matchup $(i,i+1)$ requires that $\cQ_{\pi(i)}^{f_i}$ is of Type $\bA$ or of Type $\bB_2$, and $\cQ_{\pi(i+1)}^{f_{i+1}}$ is of Type $\bA$ or of Type $\bB_1$.
We therefore have that if $b>0$, then $\nu(\pi,\mbf)\leq a + \min\{b_1,b_2\} \leq \mu(a,b)$. Moreover, if $b=0$, then $\nu(\pi,\mbf)\leq \max\{0,a - 1\}  = \mu(a,b)$. Since $\pi$ and $\mbf$ were chosen arbitrarily, it follows that $\nu(\pi,\mbf) \leq  \mu(a,b)$ for every permutation $\pi$ and $\mbf\in \{0,1\}^c$,
which, by \eqnref{eq:chunks}, proves the ``$\geq$'' direction of \eqnref{eqn:minlenab}.

It remains to prove the ``$\leq$'' direction of \eqnref{eqn:minlenab}.
Clearly, there are $\pi$ and $\mbf$ such that the sequence $\cQ_{\pi(1)}^{f_1}\oplus \cdots\oplus \cQ_{\pi(c)}^{f_c}$ is in normal form. 
Now, it is straightforward to check that the number of matchups for 
$\cQ_1, \hdots, \cQ_c$ is exactly $\mu(a,b)$,
thus proving the lemma. 
\end{proof}

Notice that the proof of \ref{lemma:properordering} together with
\ref{lem:concatenation} immediately
give an algorithm for graphs $G$ with no $K_1$- and $K_2$-components, namely, find a type-optimal
path decomposition for each connected component $G_i$ separately, order the resulting
path decompositions so that their sequence is in normal form, and concatenate them. 

Therefore, it remains do show how many $K_1$- and $K_2$-components need to be added to each $G_i$ to make up a chunk graph $C^i$. Section~\ref{sec:disconnected_general} deals with this question.

%% file: disconnected_general.tex
\newcommand\phmin{\mbox{\phantom{$-$}}}

We assume that $G$ is a graph with big connected components $G_1,\ldots,G_c$, $c\geq 1$, $K_1$-components $K_1^1,\ldots,K_1^{q_1}$, and $K_2$-components $K_2^1,\ldots,K_2^{q_2}$. 

Clearly we could find all possible chunk graphs for $G$ by enumerating all distributions of $K_1$- and $K_2$-components of $G$ among the big components of $G$. 
However, the running time of such a procedure is not, in general, polynomial in the size of $G$. We overcome this problem by designing a dynamic programming procedure that eliminates unnecessary distributions leaving only those that can possibly lead to minimal length path decomposition of $G$. The procedure calls Algorithm~\ref{alg:chunk} to determine type-optimal path decomposition for each distribution (i.e. for given chunk graphs) it considers worth trying. We now give details of the procedure.

For any vector $\mbs = (s_1, s_2) \in \dZ^2_+$ and for any integer $i\in\{1,\ldots,c\}$, construct $H_i(\mbs)$ by taking $G_i$, $s_1$ isolated vertices and $s_2$ isolated edges, and let $\cQ_i(\mbs)$ be a 
type-optimal path decomposition of $H_i(\mbs)$.
Let $\tau(\cQ_i(\mbs))$ be the type of $\cQ_i(\mbs)$.
For $m\in\{1,\ldots,c\}$, $r_1\in\{0,\ldots,q_1\}$ and $r_2\in\{0,\ldots,q_2\}$, define
\[
\cD_m(r_1,r_2) = 
\left\{\mbd\colon \{1,\ldots,m\}\to\dZ^2_+
\ \middle|\  \sum_{i=1}^m \mbd_1(i) = r_1 \mbox{ and }
\sum_{i=1}^m \mbd_2(i) = r_2\right\},
\]
where $\mbd_1(i)$ and $\mbd_2(i)$ are the first and second entry of the vector $\mbd(i)$, respectively.
The set $\cD_m(r_1,r_2)$ represents all assignments of $r_1$ $K_1$-components
and $r_2$ $K_2$-components to the first $m$ big components 
$G_1, \hdots, G_m$. 
Thus, $\mbd_1(i)$ and $\mbd_2(i)$ are the quantities of $K_1$- and $K_2$-components, respectively, that are assigned to the big component $G_i$, $i=1,\ldots,m$.
Thus, once $\mbd$ is fixed, the chunk graphs are fixed. 
And with those fixed, \ref{lem:concatenation} and \ref{lemma:properordering} show that by concatenating their type-optimal path decompositions in the normal form we obtain a minimum length path decomposition of $G$.

Let $\mbd \in \cD_m(r_1,r_2)$.
Let $\#\bA(\mbd)$ and $\#\bB(\mbd)$ be the numbers of path decompositions of Type~$\bA$ and of Type~$\bB$, respectively, among the type-optimal path decompositions
$\cQ_1(\mbd(1)), \hdots, \cQ_m(\mbd(m))$.
Define 
\begin{equation} \label{eqn:lenadef}
\length(\mbd) = \sum_{i=1}^m \length[\cQ_i(\mbd(i))] - 
\mu(\mbd).
\end{equation}
Moreover, let
\[
\omega(\mbd) = 
\begin{cases}
1, & \mbox{if $\#\bA(\mbd)=\#\bB(\mbd)=0$;} \\
2, & \mbox{if $\#\bA(\mbd)>0$ and $\#\bB(\mbd)=0$;} \\
3, & \mbox{if $\#\bB(\mbd)>0$ and $\#\bB(\mbd)$ is odd;} \\
4, & \mbox{if $\#\bB(\mbd)>0$ and $\#\bB(\mbd)$ is even.}
\end{cases}
\]
We will call $\omega(\mbd)$ the \textit{configuration} of $\mbd$.
It divides all possible normal-form path decompositions for $\cQ_1(\mbd(1)), \hdots, \cQ_m(\mbd(m))$ into four different categories, depending on the number of  path decompositions  of Type~$\bA$ and~$\bB$ among $\cQ_1(\mbd(1)), \hdots, \cQ_m(\mbd(m))$. 
Our dynamic programming algorithm works as follows.
It builds a $4$-dimensional array $\phi^*_m$, $m\in \{1,\ldots,c\}$, with the entries  $\phi^*_m(r_1,r_2,t)$
where $r_1\in\{0,\ldots,q_1\}$, $r_2\in\{0,\ldots,q_2\}$ and $t\in\{1,2,3,4\}$, 
which satisfy
\begin{equation} \label{defn:dpquantity1}
\phi^*_m(r_1,r_2,t) =\min\big\{\length(\mbd) \st \mbd \in \cD_m(r_1,r_2) \mbox{ and } \omega(\mbd) = t\big\}.
\end{equation}

We set $\phi^*_m(r_1,r_2,t) =\infty$ for $\big\{\mbd\st \mbd \in\cD_m(r_1,r_2) \mbox{ and } \omega(\mbd) = t\big\}=\emptyset$. It follows from \ref{lem:concatenation} and \ref{lemma:properordering} that for each $t\in\{1,2,3,4\}$, $\phi_c^*(q_1, q_2,t)<\infty$ is the minimum integer such that
there exists a path decomposition $\cP$ of width $3$ of $G$ with
$\length(\cP)=\phi_c^*(q_1, q_2,t)$ and with the corresponding vector $\mbd$ in configuration $t$, $\omega(\mbd)=t$.
Thus, $l=\min_{t\in\{1,2,3,4\}}\{\phi_c^*(q_1, q_2,t)\}$ is the minimum integer such that there exists
a path decomposition of width $3$ and length $l$ of $G$. Clearly, $\phi_c^*(q_1, q_2,t)<\infty$ for some $t\in\{1,2,3,4\}$.
It remains to show that the values of $\phi^*_m(r_1,r_2,t)$ can be recursively
calculated in polynomial time, which is what the next two lemmas show.
We start with the following crucial observation:

\begin{lemma} \label{lem:dp_key}
Let $m\in\{2,\ldots,c\}$. Let $\bar\mbd\in \cD_m(r_1, r_2)$ and let 
$\mbd\in \cD_{m-1}\big(\sum_{i=1}^{m-1}\bar\mbd_1(i), \sum_{i=1}^{m-1}\bar\mbd_2(i)\big)$ 
be be such that $\mbd(i)=\bar\mbd(i)$ for each $i=1,\ldots,m-1$. If $\tau(\cQ_m(\bar\mbd(m)))=\bC$, then $\mu(\#\bA(\bar\mbd),\#\bB(\bar\mbd)) = \mu(\#\bA(\mbd),\#\bB(\mbd))$ and $\delta(\mbd, \bar \mbd)=0$. Otherwise, 
\[
\delta(\mbd, \bar \mbd) =  \mu(\#\bA(\bar\mbd),\#\bB(\bar\mbd)) - \mu(\#\bA(\mbd),\#\bB(\mbd)) = 
\begin{cases}
0, & \mbox{if $(\omega(\mbd), \omega(\bar\mbd))\in \{(1,2), (1,3), (4,3), (2,3)\};$}\\
1, & \mbox{otherwise.}
\end{cases}
\]
In particular, $\delta(\mbd, \bar \mbd)$ only depends on $\omega(\mbd)$,  $\omega(\bar\mbd)$, and $\tau(\cQ_m(\bar\mbd(m)))$.
\end{lemma}
\begin{proof}
If $\cQ_m(\bar\mbd(m))$ is of Type $\bC$, the result follow from the fact that $\cQ_i(\mbd(i))$ and $\cQ_i(\bar\mbd(i))$ are of the same type for each $i=1,\ldots,m-1$.
So we may assume that $\cQ_m(\bar\mbd(m))$ is not of Type~$\bC$.
The possible transitions from $\omega(\mbd)$ to $\omega(\bar\mbd)$ when $\cQ_m(\bar\mbd(m))$ is of Type $\bA$ or $\bB$ are graphically shown in \figref{fig:transitions}.
\begin{figure}[htb]
\begin{center}
\includegraphics{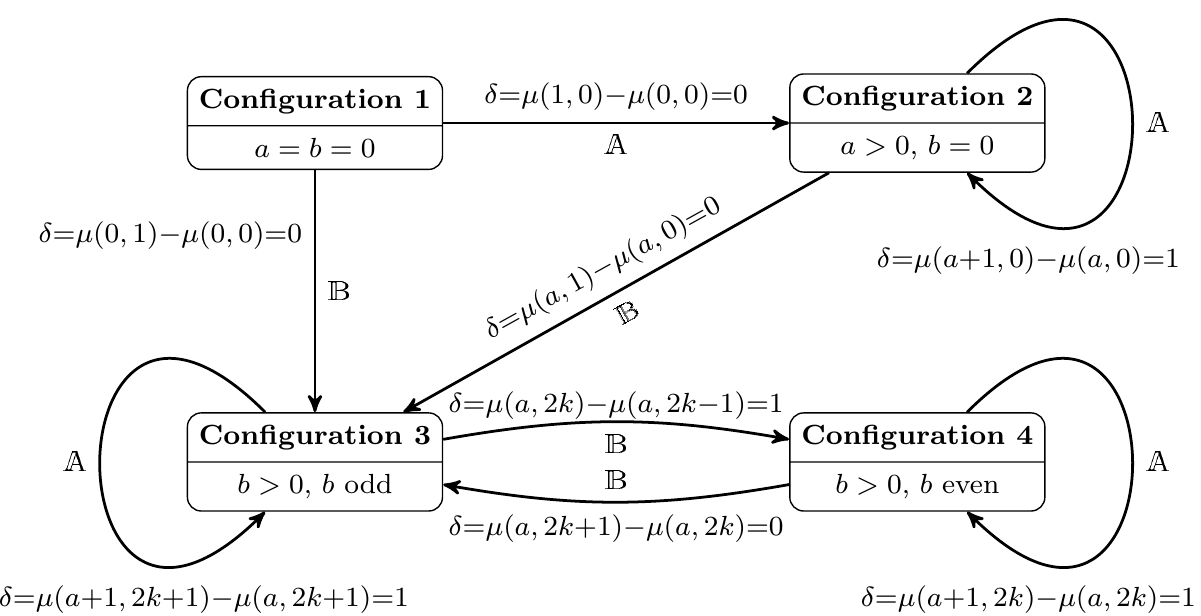}
\end{center}
\caption{Transitions between different configurations. \label{fig:transitions}}
\end{figure}
The nodes represent configurations $\omega\in\{1,2,3,4\}$ and the edges represent pairs $(\omega(\mbd), \omega(\bar\mbd))$.
Each edge has two labels, one that indicates the value of $\delta(\mbd, \bar \mbd)$ and the other that indicates the type of $\cQ_m(\bar\mbd(m))$. 
Notice that $\#\bA(\bar\mbd) - \#\bA(\mbd) \in \{0,1\}$ and $\#\bB(\bar\mbd) - \#\bB(\mbd) \in \{0,1\}$. It is now straightforward to check that the edges depicted in the figure are the only possible pairs $(\omega(\mbd), \omega(\bar\mbd))$ and that the value of $\delta(\mbd, \bar \mbd)$ is exactly as given in the figure. 
\end{proof}

Let  $h:\{1,2,3,4\}\times \{\dA,\dB,\dC\} \to \{1,2,3,4\}$ be the transition function  and $g:\{1,2,3,4\}\times \{\dA,\dB,\dC\} \to \{0,1\}$ be transition weights $\delta$ from Figure \ref{fig:transitions}. We assume $h(t,\dC)=t$ and $g(t,\dC)=0$ for any configuration $t$.

\begin{lemma} \label{lemma:dp}
Let $m\in\{2,\ldots,c\}$.
If there exists a polynomial-time algorithm that, for any chunk graph $C$, and for any $\bagsize{1},\bagsize{2}\in\{2,\ldots,k+1\}$,
either constructs a minimum length $(\bagsize{1},\bagsize{2})$-path decomposition of width $k$ of $C$, 
or concludes that no such path decomposition exists, and if the array $\phi^*_{m-1}$ is given, then $\phi^*_m(r_1, r_2, t)$ can be computed in polynomial time for each $r_1\in\{0,\ldots,q_1\}$, $r_2\in\{0,\ldots,q_2\}$, $t\in\{1,2,3,4\}$.
\end{lemma}
\begin{proof} Fix $r_1$, $r_2$, $t$, and $m$. In  order to prove the lemma, we will show that the following recursive dynamic programming relation holds:
\begin{equation}\label{eqn:dynprogeqn}
\phi^*_m(r_1,r_2,t) = \min_{\substack{0\leq i\leq r_1\\ 0\leq j\leq r_2 \\ 1 \leq t'\leq 4}} 
\bigl[ \phi_{m-1}^*(r_1 - i, r_2-j,t')+\length(\cQ_m(i,j))-g\big(t',\tau(\cQ_m(i,j))\big) 
\st h(t',\tau(\cQ_m(i,j)) = t
\bigr],
\end{equation}
where the right hand side of (\ref{eqn:dynprogeqn}) equals $\infty$ if $\phi^*_{m-1}(r_1-i,r_2-j,t')= \infty$ or  $h(t',\tau(\cQ_m(i,j)) \neq t$ for each $i\in\{0,\ldots,r_1\}$, $j\in\{0,\ldots,r_2\}$ and $t'\in\{1,2,3,4\}$. 

This suffices because the right hand side of \eqnref{eqn:dynprogeqn} can be calculated in polynomial time.
Indeed, the array $\phi^*_{m-1}$ is given by assumption; $\cQ_m(i,j)$, its length, $\length(\cQ_m(i,j))$, and its type, $\tau(\cQ_m(i,j))$, can be calculated in polynomial time by assumption; thus, since the initial configuration $t'$ is given for each entry $\phi^*_{m-1}$, both $h(t',\tau(\cQ_m(i,j)))$ and $g(t',\tau(\cQ_m(i,j)))$  can  be readily calculated.

We now prove \eqnref{eqn:dynprogeqn}. First, we observe that if for each $i\in\{0,\ldots,r_1\}$, $j\in\{0,\ldots,r_2\}$ and $t'\in\{1,2,3,4\}$, $\phi^*_{m-1}(r_1-i,r_2-j,t')= \infty$ or $h(t',\tau(\cQ_m(i,j)) \neq t$, then the right hand side of  \eqnref{eqn:dynprogeqn} equals $\infty$. On the other hand, the definition of $\phi^*_{m-1}$ then implies that the set $\big\{\mbd\st \mbd \in\cD_m(r_1,r_2) \mbox{ and } \omega(\mbd) = t\big\}$ is empty and thus $\phi^*_m(r_1,r_2,t)=\infty$ and  \eqnref{eqn:dynprogeqn} holds as required. Therefore, we assume now that there are  $i\in\{0,\ldots,r_1\}$, $j\in\{0,\ldots,r_2\}$ and $t'\in\{1,2,3,4\}$ such that $\phi^*_{m-1}(r_1-i,r_2-j,t')< \infty$ and $h(t',\tau(\cQ_m(i,j)) = t$. We  prove \eqnref{eqn:dynprogeqn} by proving that the inequality holds in 
both directions.

``$\leq$'': Let $\mbd^{i,j,t'}\in\cD_{m-1}(r_1-i,r_2-j)$ be a vector such that $\length(\mbd^{i,j,t'}) = \phi^*_{m-1}(r_1-i,r_2-j,t')$ and $\omega(\mbd^{i,j,t'}) = t'$.
Let $\bar\mbd^{i,j,t'}\in \cD_m(r_1,r_2)$ be the extension of $\mbd^{i,j,t'}$ to $\{1,\ldots,m\}$ that satisfies $\bar \mbd_1^{i,j,t'}(m) = i$, $\bar \mbd_2^{i,j,t'}(m) = j$, and $\omega(\bar \mbd^{i,j,t'}) = t$. By \ref{lemma:properordering},
\begin{align*}
\phi_m^*(r_1, r_2, t) & \leq \length(\bar\mbd^{i,j,t'}) 
= \sum_{l=1}^m \length[\cQ_l(\bar\mbd^{i,j,t'}(l))] - \mu(\#\bA(\bar\mbd^{i,j,t'}),\#\bB(\bar\mbd^{i,j,t'})) \\
& = \left[\sum_{l=1}^{m-1} \length[\cQ_l(\mbd^{i,j,t'}(l))] - \mu(\#\bA(\mbd^{i,j,t'}),\#\bB(\mbd^{i,j,t'}))\right] + \length(\cQ_m(i,j)) - \delta(\mbd^{i,j,t'}, \bar\mbd^{i,j,t'}) \\
& = \phi_{m-1}^*(r_1-i,r_2-j,t')+ \length(\cQ_m(i,j))-\delta(\mbd^{i,j,t'}, \bar\mbd^{i,j,t'})\\
& = \phi_{m-1}^*(r_1-i,r_2-j,t')+ \length(\cQ_m(i,j))
-g(t',\tau(\cQ_m(i,j))),
\end{align*}
where $\delta(\mbd^{i,j,t'} \bar\mbd^{i,j,t'})=g(t',\tau(\cQ_m(i,j)))$ in the last equality follows from \ref{lem:dp_key}  and $\omega(\mbd^{i,j,t'}) = t'$. Finally, $\omega(\bar \mbd^{i,j,t'}) = h(t',\tau(\cQ_m(i,j)) = t$ as required.

``$\geq$'':
Let $\bar\mbd^* \in \cD_m(r_1,r_2)$ be such that $\length(\bar\mbd^*) = \phi_m^*(r_1,r_2,t)$. Thus, $\omega(\bar\mbd^*) = t$. Let $i = \bar\mbd_1^*(m)$, $j = \bar\mbd_2^*(m)$, let $\mbd^*\in \cD_{m-1}(r_1-i, r_2-j)$ be the restriction of $\bar\mbd^*$ to $\{1,\ldots,m-1\}$, and let $t' = \omega(\mbd^*)$. 
By \ref{lemma:properordering}, 
\begin{align*}
\phi_m^*(r_1, r_2, t) & = \length(\bar\mbd^*)  
= \sum_{l=1}^{m} \length[\cQ_l(\bar\mbd^*(l))] - \mu(\#\bA(\bar\mbd^*),\#\bB(\bar\mbd^*)) \\
& = \left[\sum_{l=1}^{m-1} \length[\cQ_l(\mbd^*(l))] - \mu(\#\bA(\mbd^*),\#\bA(\mbd^*))\right] + \length(\cQ_m(i,j)) - \delta(\mbd^*, \bar\mbd^*) \\
& \geq \phi_{m-1}^*(r_1 - i, r_2 - j, t') +\length(\cQ_m(i,j)) - \delta(\mbd^*, \bar\mbd^*)\\
& = \phi_{m-1}^*(r_1 - i, r_2 - j, t') +\length(\cQ_m(i,j)) - g(t',\tau(\cQ_m(i,j)))\\
\end{align*}
where $\delta(\mbd^*, \bar\mbd^*)=g(t',\tau(\cQ_m(i,j)))$ in the last equality follows from \ref{lem:dp_key}  and $\omega(\mbd^{i,j,t'}) = t'$. Finally, $\omega(\bar\mbd^*) = h(t',\tau(\cQ_m(i,j)) = t$ as required.
\end{proof}

We are now ready to prove \ref{thm:anygraph}.

\textbf{Proof of \ref{thm:anygraph}:}
Let $k=3$. The array 
$\phi^*_1$ can be directly computed using Algorithm~\ref{alg:chunk} for chunk graphs.
Thus, induction on $m$ and \ref{lemma:dp} (given $\phi^*_{m-1}$, $\phi^*_m$ can be computed with the help of Algorithm~\ref{alg:chunk}) imply that \ref{thm:anygraph} follows for $k=3$.

Let $k<3$.
By \ref{ob:chunks:concatenation}, finding a minimum-length path decomposition of $G$ reduces to determining of the assignment of small components to the big components, i.e., it reduces to partitioning $G$ into chunk graphs and then to concatenating (in any order) the minimum-length path decomposition of the chunk graphs.
Hence, in the case of $k\in\{1,2\}$ we use the dynamic programming algorithm in which all $(\bagsize{1},\bagsize{2})$-path decompositions of chunk graphs are of the same type.
Thus, each minimum length path decomposition of a chunk graph is type optimal.
The decomposition can be computed by using Algorithm~\ref{alg:chunk}.
\qed

We finish this section by stating the main results of this paper. By the fact that the problem $\PathLenFixed{0}$ is trivial, and by \ref{thm:chunkgraphs} and \ref{thm:anygraph} we obtain:
\begin{theorem} \label{thm:final}
Given a graph $G$ and an integer $k\leq 3$, there exists a polynomial-time algorithm that computes a minimum length path decomposition of width at most $k$ of $G$, or concludes that no such path decomposition exists.
\qed
\end{theorem}